\documentclass[runningheads]{llncs}

\usepackage{amsmath,amssymb}
\usepackage{graphicx}
\usepackage{relsize}
\usepackage[absolute]{textpos}

\usepackage[lined,boxed,commentsnumbered]{algorithm2e}
\usepackage{framed}
\usepackage{setspace}
\usepackage{mathtools}
\usepackage[inline]{enumitem}
\usepackage{tikz}

\usepackage{ntheorem}

\theoremstyle{plain}
\theorembodyfont{}
\theoremsymbol{}
\theoremprework{}
\theorempostwork{}
\theoremseparator{.}

\newtheorem{observation}[proposition]{Observation}
\newtheorem{myclaim}{Claim}

\makeatletter
\newcommand{\removelatexerror}{\let\@latex@error\@gobble}
\makeatother

\renewenvironment{framed}[1][\hsize]
   {\MakeFramed{\hsize#1\advance\hsize-\width \FrameRestore}}%
   {\endMakeFramed}

\newcommand{\Mod}[1]{\ (\mathrm{mod}\ #1)}
\newcommand{\Thcref}[1]{Theorem~\ref{#1}}
\newcommand{\Prcref}[1]{Proposition~\ref{#1}} 
\newcommand{\Clcref}[1]{Claim~\ref{#1}} 
\newcommand{\Obcref}[1]{Observation~\ref{#1}}   
\newcommand{\Lecref}[1]{Lemma~\ref{#1}}
\newcommand{\Cocref}[1]{Corollary~\ref{#1}}
\newcommand{\biglefttriangle}{\text{\larger[3]$\vartriangleleft$}}
\newcommand{\nagents}{\mathrm{\# agents}}
\newcommand{\nitems}{\mathrm{\# items}}

\newcommand{\Tcref}[1]{T.\ref{#1}}
 
\newcommand{\Ccref}[1]{C.\ref{#1}}

\newcommand*{\myproofname}{Proof}

\newcommand{\defproblem}[3]{
\vspace{1mm}
\noindent\fbox{
  \begin{minipage}{0.96\textwidth}
  \begin{tabular*}{\textwidth}{@{\extracolsep{\fill}}lr} #1 \\ \end{tabular*}
  {\bf{Input:}} #2  \\
  {\bf{Question:}} #3
  \end{minipage}
}\vspace{1mm}}

\newcommand{\ThreeSat}{\textsc{3-SAT}}
\newcommand{\assign}{\textsc{Assignment}}
\newcommand{\goassign}{\textsc{$\alpha$-Globally Optimal Assignment}}
\newcommand{\serDictFeas}{\textsc{Serial Dictatorship Feasibility}}

\newcommand{\MCIS}{\textsc{Multicolored Independent Set}}

\newcommand{\pc}{\textsc{$k \times k$ Permutation Clique}}
\newcommand{\ind}{\mathrm{ind}}
\newcommand{\bb}{\mathrm{b}}
\newcommand{\dd}{\mathrm{d}}

\newcommand{\XX}{\mathcal{X}}
\newcommand{\CC}{\mathcal{C}}
\newcommand{\OO}{\mathcal{O}}

\newcommand{\yes}{\textsf{Yes}}
\newcommand{\no}{\textsf{No}}
\newcommand{\yesinstance}{\yes-instance}

\newcommand{\NP}{\textsf{NP}}

\newcommand{\coNPpoly}{\textsf{coNP/poly}}
\newcommand{\FPT}{\textsf{FPT}}
\newcommand{\XP}{\textsf{XP}}
\newcommand{\WO}{\textsf{W[1]}}
\newcommand{\WOH}{\textsf{W[1]}-hard}

\newcommand{\NPH}{\textsf{NP}-hard}

\newcommand{\paraH}{\textsf{para-NP}-hard}
\newcommand{\ETH}{\textsf{ETH}}

\input{tikzit.sty}

\tikzstyle{vertex}=[fill=white, draw=black, shape=circle, minimum size=1.2cm]
\tikzstyle{rec}=[fill=white, draw=black, shape=rectangle, minimum width=1.4cm, minimum height=0.7cm]

\tikzstyle{edge}=[->]

\begin{document}\sloppy

\title{Parameterized Analysis of Assignment Under Multiple Preferences}
\author{Barak Steindl \and Meirav Zehavi}
\authorrunning{Barak Steindl and Meirav Zehavi}

\institute{Ben Gurion University of the Negev, Beer-Sheva, Israel}

\maketitle

\begin{abstract}
The \assign\ problem is a fundamental and well-studied problem in the intersection of Social Choice, Computational Economics and Discrete Allocation. In the \assign\ problem, a group of agents expresses preferences over a set of items, and the task is to find a {\em pareto optimal} allocation of items to agents. We introduce a generalized version of this problem, where each agent is equipped with {\em multiple} incomplete preference lists: each list (called a {\em layer}) is a ranking of items in a possibly different way according to a different criterion. We introduce the concept of {\em global optimality}, which extends the notion of pareto optimality to the multi-layered setting, and we focus on the problem of deciding whether a {\em globally optimal} assignment exists. We study this problem from the perspective of Parameterized Complexity: we consider several natural parameters such as the number of layers, the number of agents, the number of items, and the maximum length of a preference list. We present a comprehensive picture of the parameterized complexity of the problem with respect to these parameters.
\end{abstract}
\section{Introduction}
The field of resource allocation problems has been widely studied in recent years. A fundamental and one of the most well-studied problems in this field is the \assign\ problem\footnote{The problem is called \assign\ in all relevant literature. Although this name is somewhat generic, to be consistent with the literature, we use it here as well.} \cite{RePEc:ecm:emetrp:v:66:y:1998:i:3:p:689-702,Bogomolnaia2001ANS,Abdulkadirog1999HouseAW,10.1007/978-3-540-30551-4-3,RePEc:eee:jetheo:v:52:y:1990:i:1:p:123-135,AAAI1817396,ijcai2017-12,soton425734,gourves:hal-01741519}. In the \assign\ problem we are given a set of $n$ agents, and a set of $m$ items. Each agent (human, company, or any other entity) has strict preferences over a subset of items, and the objective is to allocate items to agents in an ``optimal" way. Different notions of optimality have been considered in the literature, but the one that has received the most attention is {\em pareto optimality} (see, e.g., \cite{RePEc:ecm:emetrp:v:66:y:1998:i:3:p:689-702,ijcai2017-12,soton425734}). Intuitively, an assignment $p$ is called {\em pareto optimal} if there is no other assignment $q$ that is at least good as $p$ for all the agents and also strictly better than $p$ for at least one agent.

Besides its theoretical interest, the problem has also practical importance. Algorithms for the \assign\ problem have applications in a variety of real-world situations, such as assigning jobs to workers, campus houses to students, time stamps to users on a common machine, players to sports teams, graduating medical students to their first hospital appointments, and so on \cite{doi:10.3138/infor.45.3.123,singh2012comparative,faudzi2018assignment,doi:10.4103/0256-4602.78092}.

A simple and well-studied allocation mechanism is the greedy serial dictatorship mechanism, introduced by Abdulkadiroglu and Sönmez \cite{RePEc:ecm:emetrp:v:66:y:1998:i:3:p:689-702}. In the serial dictatorship mechanism, we draw a random permutation on the agents from the uniform distribution. The agent ordered first in the permutation is assigned to its top choice, the agent ordered second is assigned to its top choice among the remaining items and so on. Abdulkadiroglu and Sönmez \cite{RePEc:ecm:emetrp:v:66:y:1998:i:3:p:689-702} proved that applying this mechanism results in a pareto optimal assignment.

In the \assign\ problem, each agent has exactly one preference list. The preference lists may represent a single subjective criterion according to which each agent ranks the items. However, they may also represent a combination of different such criteria: each agent associates a score to each item per criterion, and a single preference list is derived from some weighted sum of the scores. In many cases, it is unclear how to combine scores associated with criteria of inherently incomparable nature - that is like ``comparing apples with oranges". Additionally, even if a single list can be forcefully extracted, most data is lost.\footnote{Our new generalized model provides the ability to limit the amount of data that can be ignored.} 

Thus, the classic model seems somewhat restrictive in real world scenarios where people rely on highly {\em varied} aspects to rank other entities. For example, suppose that there are $n$ candidates who need to be assigned to $n$ positions. The recruiters may rank the candidates for each position according to different criteria, such as academic background, experience, impression by the interview, and so on \cite{KINICKI1985117,alderfer1970personal}. Moreover, when assigning campus houses to students, the student may rank the houses by multiple criteria such as their location (how close the house is to their faculty), rent, size etc \cite{wu2016students}. This motivates the employment of multiple preference lists where each preference list (called a {\em layer}) is defined by a different criterion.

In many real-world scenarios, the preferences of the agents may sometimes depend on external circumstances that may not be completely known in advance such as growth of stocks in the market, natural phenomena, outbreak of pandemics \cite{zeren2020impact,TOPCU2020101691} and so on. In such cases, each layer in our generalized model can represent a possible ``state'' of the world, and we may seek an assignment that is optimal in as many states as possible. For instance, suppose that there is a taxi company with $n$ taxis and $m$ costumers ($n > m$) that want to be picked at a specific time in future. The ``cost'' of each taxi depends on the time taken to reach the costumer from the starting location of the taxi. Many factors (that may not be completely known a-priori) may affect the total cost such as road constructions, weather, car condition and the availability of the drivers \cite{cools2010assessing,saleh2017study}. The firm may suggest different possible scenarios (each represents a layer). For each scenario, the costumers may be ranked differently by the taxis, and an assignment that is pareto optimal in as many layers as possible will cover most of the scenarios and will give the lowest expected total cost. 

Furthermore, it is not always possible to completely take hold of preferences of some (or all) agents due to lack of information or communication, as well as security and privacy issues \cite{nass2009value,browne2000lives}. In addition, even if it is technically and ethically feasible, it may be costly in terms of money, time, or other resources to gather all information from all the agents \cite{mulder2019willingness}. In these cases, we can ``complete the preferences'' using different assumptions on the agents. As a result, we will have a list of preference profiles that represent different possible states of the world. An assignment that is pareto optimal in as many preference profiles as possible will be pareto optimal with high probability.

Our work is inspired by the work of Chen et al.~\cite{10.1145/3219166.3219168}, who studied the \textsc{Stable Marriage} problem under multiple preference lists. In contrast to the \assign\ problem, the \textsc{Stable Marriage} problem is a {\em two-sided} matching problem, i.e.~it consists of two disjoint sets of agents $A$ and $B$, such that each agent strictly ranks the agents of the opposite set (in the \assign\ problem, only the agents rank the items). The objective in the \textsc{Stable Marriage} problem is to find a matching (called a {\em stable matching}) between $A$ and $B$ such that there do not exist agents $a \in A$ and  $b \in B$ that are not matched to each other but rank each other higher than their matched partners.\footnote{In Section \ref{sec:preliminaries}, we further argue that the \assign\ and \textsc{Stable Marriage}, problems, being based on different concepts of stability, are very different problems.}

Chen et al.~\cite{10.1145/3219166.3219168} considered an extension of the \textsc{Stable Marriage} problem where there are $\ell$ layers of preferences, and adapted the definition of stability accordingly. Specifically, three notions of stability were defined: {\em $\alpha$-global stability}, {\em $\alpha$-pair stability}, and {\em $\alpha$-individual stability}. In their work, Chen et al.~\cite{10.1145/3219166.3219168} studied the algorithmic complexity of finding matchings that satisfy each of these stability notions. Their notion of $\alpha$-global stability extends the original notion of stability in a natural way, by requiring the sought matching to be stable in (at least) some $\alpha$ layers. Our notion of {\em $\alpha$-global optimality} extends pareto optimality in the same way, by requiring an assignment to be pareto optimal in some $\alpha$ layers.  
In contrast to \cite{10.1145/3219166.3219168}, our research focuses more on the perspective of parameterized complexity.  We take into account a variety of parameters (\cite{10.1145/3219166.3219168} focuses only on the parameter $\alpha$), and we provide parameterized algorithms, parameterized hardness results, and \ETH -based lower bounds regarding them (\cite{10.1145/3219166.3219168} focuses only on parameterized hardness results). We also study the kernelization complexity of the problem: We provide kernels, and lower bounds on kernel sizes.

Although the \assign\ problem can be solved in polynomial time using a mechanism called ``serial dictatorship" \cite{RePEc:ecm:emetrp:v:66:y:1998:i:3:p:689-702}, we show that the problem becomes much harder when multiple preference lists are taken into account. So, in this paper, we study the parameterized complexity of deciding whether a globally optimal assignment exists with respect to various parameters.

\smallskip
\noindent\textbf{Contributions and Methods.} \begin{table*}[t]
\centering
\resizebox{\columnwidth}{!}{
\begin{tabular}{l | c | c | c}
Parameter $k$ &\:\:\:\:\:Complexity Class\:\:\:\:\:&\:\:\:\:\:\:\:\:Running Time\:\:\:\:\:\:\:\:& Polynomial Kernel? \\
\hline
$\ell + d$ & \paraH\ [\Tcref{theorem:threesat-reduction}]& - & -\\
$\ell + (m - d)$ & \paraH\ [\Tcref{theorem:threesat-reduction}]& - & -\\
$n$& \FPT & $\OO^{*}(n!)^{\dagger}$ [\Tcref{theorem:goassign-is-fpt-num-of-agents} + \Tcref{theorem:optimal-agents-items-alpha}] & no [\Tcref{theorem:nopoly-kernel-agents-items-alpha}]\\
$m$ & \XP, \WOH\ [\Tcref{w-one-hardness-items-alpha}] &  $(nm)^{\OO(m)}$ [\Tcref{xp-for-items}] & -\\
$m + \alpha$ & \XP, \WOH\ [\Tcref{w-one-hardness-items-alpha}] &  $(nm)^{\OO(m)}$ [\Tcref{xp-for-items}]& -\\
$n+m+\alpha$ & \FPT & $\OO^*(n!)^{\dagger}$ [\Tcref{theorem:goassign-is-fpt-num-of-agents} + \Tcref{theorem:optimal-agents-items-alpha}]& no [\Tcref{theorem:nopoly-kernel-agents-items-alpha}]\\
$m + (\ell - \alpha)$ & \XP, \WOH\ [\Tcref{theorem:w-one-hardness-items-l-alpha}] &  $(nm)^{\OO(m)}$ [\Tcref{xp-for-items}] & -\\
$n + m + (\ell - \alpha)$\:\: & \FPT & $\OO^{*}(n!)^{\dagger}$ [\Tcref{theorem:goassign-is-fpt-num-of-agents} + \Tcref{theorem:optimal-agents-items-ell-alpha}] & no [\Tcref{theorem:nopoly-kernel-agents-items-alpha}]\\
$m + \ell$     & \FPT & $\OO^{*}(((m!)^{\ell+1})!)$ [\Ccref{corollary:items-plus-ell-algorithm}] & no [\Tcref{theorem:nopoly-kernel-items-l}]\\
$n + \ell$& \FPT & $\OO^{*}(n!)$ [\Tcref{theorem:goassign-is-fpt-num-of-agents}]& yes [\Tcref{theorem:poly-kernel}]\\
$n + m + \ell$ & \FPT & $\OO^{*}(n!)$ [\Tcref{theorem:goassign-is-fpt-num-of-agents}] & yes [\Tcref{theorem:poly-kernel}]\\
\end{tabular}
}
\caption{Summary of our results for \goassign. Results marked with $\dagger$ are proved to be optimal under the exponential-time hypothesis.}
\label{tab:results-summary}
\end{table*} One important aspect of our contribution is conceptual: we are the first to study pareto optimality (in the \assign\ problem) in the presence of multiple preference lists. This opens the door to many future studies (both theoretical and experimental) of our concept, as well as refinements or extensions thereof (see Section \ref{sec:conclusion}). In this work, we focus on the classical and parameterized complexity of the problem.

We consider several parameters such as the number of layers $\ell$, the number of agents $n$ (also denoted by $\nagents$), the number of items $m$ (also denoted by $\nitems$), the maximum length of a preference list $d$, and the given number of layers $\alpha$ for which we require an assignment to be pareto optimal. The choice of these parameters is sensible because in real-life scenarios such as those mentioned earlier, some of these parameters may be substantially smaller than the input size. For instance, $\ell$, $\alpha$ and $\ell - \alpha$ are upper bounded by the number of criteria according to which the agents rank the items. Thus, they are likely to be small in practice: when ranking other entities, people usually do not consider a substantially large number of criteria (further, up until now, attention was only given to the case where $\ell = \alpha = 1$). For instance, when sports teams rank candidate players, only a few criteria such as the player's winning history, its impact on its previous teams, and physical properties are taken into account \cite{fearnhead2011estimating}. In addition, the parameter $\ell - \alpha$ may be small particularly in cases where we want to find an assignment that is optimal with respect to as many criteria as possible. Moreover, in various cases concerning ranking of people, jobs, houses etc., people usually have a limited number of entities that they want or are allowed to ask for \cite{clinedinst2019state}. In these cases, the parameter $d$ is likely to be small. Moreover, in small countries (such as Israel), the number of universities, hospitals, sports teams and many other facilities and organizations is very small \cite{chernichovskystate,educationInIsrael}. Thus, in scenarios concerning these entities, at least one among $n$ and $m$ may be small. A summary of our results is given in Table \ref{tab:results-summary}.

\smallskip
\noindent\textbf{Fixed-Parameter Tractability and ETH Based Lower Bounds.}
We prove that \goassign\ is {\em fixed-parameter tractable} (\FPT) with respect to $n$ by providing an $\OO^{*}(n!)$ time algorithm that relies on the connection between pareto optimality and serial dictatorship. We then prove that the problem admits a polynomial kernel with respect to $n + \ell$ and that it is \FPT\ with respect to $m + \ell$ by providing an exponential kernel. We also prove that the problem is {\em slice-wise polynomial} (\XP) with respect to $m$ by providing an $m^{\OO(m)}\cdot n^{\OO(m)}$ time algorithm. In addition, we prove that $\OO^{*}(2^{\OO(t\log{t})})$ is a tight lower bound on the running time (so, our $\OO^{*}(n!)$ time algorithm is essentially optimal) under \ETH\ (defined in Section \ref{sec:preliminaries}) for even larger parameters such as $t = n + m + \alpha$ and $t = n + m + (\ell - \alpha)$ using two linear parameter reductions from the \pc\ problem. Lastly, we prove that the problem is \WOH\ with respect to $m + \alpha$ and $m + (\ell - \alpha)$ using two parameterized reductions from \MCIS.

\smallskip
\noindent\textbf{NP-Hardness.}
We first prove that the problem is \NPH\ for any fixed values of $\alpha$ and $\ell$ such that $2 \leq \alpha \leq \ell$ using a polynomial reduction from the \serDictFeas\ problem that relies on a reduction by Aziz el al.~\cite{ijcai2017-12}. We also define three polynomial-time constructions of preference profiles given an instance of \ThreeSat, and we rely on them in the construction of polynomial reductions from \ThreeSat\ to the problem, such that in the resulting instances $\ell+d$ and $\ell + (m - d)$ are bounded by fixed constants. This proves that the problem is \paraH\ for $\ell+d$ and $\ell + (m - d)$.

\smallskip
\noindent\textbf{Non-Existence of Polynomial Kernels.}
We prove that the problem does not admit polynomial kernels unless \NP$\subseteq $\coNPpoly\  with respect to $n + m + \alpha$, $n + m + (\ell - \alpha)$, and $m + \ell$ using three {\em cross-compositions} (defined in Section \ref{sec:preliminaries}) from \ThreeSat\ that rely on the aforementioned reduction to prove \paraH ness.
\section{Preliminaries}
\label{sec:preliminaries}
For any natural number $t$, we denote $[t] = \lbrace 1,\ldots,t \rbrace$. We use the $\OO^{*}$ and the $\Omega^{*}$ notations to suppress polynomial factors in the input size, that is, $\OO^{*}(f(k))=f(k) \cdot n^{\OO(1)}$ and $\Omega^{*}(f(k))=\Omega(f(k)) \cdot n^{\OO(1)}$.

\smallskip
\noindent\textbf{The \assign\ problem.}
An instance of the \assign\ problem is a triple $(A,I,P)$ where $A$ is a set of $n$ agents $\lbrace a_{1},\ldots,a_{n} \rbrace$, $I$ is a set of $m$ items $\lbrace b_{1},\ldots,b_{m} \rbrace$, and $P=(<_{a_{1}},\ldots,<_{a_{n}})$, called the {\em preference profile}, contains the preferences of the agents over the items, where each $<_{a_{i}}$ encodes the preferences of $a_{i}$ and is a linear order over a {\em subset} of $I$ (preferences are allowed to be incomplete). We refer to such linear orders as {\em preference lists}. If $b_{j} <_{a_{i}} b_{r}$, we say that agent $a_{i}$ {\em prefers} item $b_{r}$ over item $b_{j}$, and we write $b_{j} \leq_{a_{i}} b_{r}$ if $b_{j} <_{a_{i}} b_{r}$ or $b_{j} = b_{r}$. Item $b$ is {\em acceptable} by agent $a$ if $b$ appears in $a$'s preference list. An {\em assignment} is an allocation of items to agents such that each agent is allocated at most one item, and each item is allocated to at most one agent. Since the preferences of the agents may be incomplete, or the number of items may be smaller than the number of agents, some agents may not have available items to be assigned to. In order to deal with this case, we define a special item $b_{\emptyset}$, seen as the least preferred item of each agent, and will be used as a sign that an agent is not allocated an item. Throughout this paper, we assume that $b_{\emptyset}$ is not part of the input item set, and that it appears at the end of every preference list (we will not write $b_{\emptyset}$ explicitly in the preference lists). We formally define assignments as follows:

\begin{definition}\label{def:assignment}
Let $A=\lbrace a_{1},\ldots,a_{n} \rbrace$ be a set of $n$ agents and let $I = \lbrace b_{1},\ldots,b_{m} \rbrace$ be a set of $m$ items. A mapping $p:A \rightarrow I \cup \lbrace b_{\emptyset} \rbrace$ is called an {\em assignment} if for each $i \in [n]$, it satisfies one of the following conditions:
\begin{enumerate}
\item $p(a_{i})=b_{\emptyset}$.
\item Both $p(a_{i}) \in I$ and for each $j \in [n]\setminus \lbrace i \rbrace$, $p(a_{i}) \neq p(a_{j})$.
\end{enumerate}
\end{definition}

We refer to $p$ as {\em legal} if, for each $i \in [n]$, it satisfies $p(a_{i})=b_{\emptyset}$ or $p(a_{i}) \in I$ is acceptable by $a_{i}$. For brevity, we will omit the term ``legal'', referring to a legal assignment just as an assignment.\footnote{All the ``optimal'' assignments that we construct in this paper will be legal in a sufficient number of layers, where they are claimed to be pareto optimal.} Moreover, when we write a set in a preference list, we assume that its elements are ordered arbitrarily, unless stated otherwise. In the \assign\ problem, we are given such triple $(A,I,P)$, and we seek a pareto optimal assignmen.

\bigskip
\noindent\textbf{Pareto Optimality.}
There are different ways to define optimality of assignments, but the one that received the most attention in the literature is {\em pareto optimality}. Informally speaking, an assignment $p$ is {\em pareto optimal} if there does not exist another assignment $q$ that is ``at least as good'' as $p$ for all the agents, and is ``better'' for at least one agent. It is formally defined as follows.

\begin{definition}\label{def:poAssignment}
Let $A = \lbrace a_{1},\ldots,a_{n} \rbrace$ be a set of agents, and let $I$ be a set of items. An assignment $p:A \rightarrow I \cup \lbrace b_{\emptyset} \rbrace$ is {\em pareto optimal} if there does not exist another assignment $q : A \rightarrow I \cup \lbrace b_{\emptyset} \rbrace$ that satisfies:
\begin{enumerate}
\item $p(a_{i}) \leq_{a_{i}} q(a_{i})$ for every $i \in [n]$.
\item There exists $i \in [n]$ such that $p(a_{i}) <_{a_{i}} q(a_{i})$.
\end{enumerate}
If such an assignment $q$ exists, we say that $q$ {\em pareto dominates} $p$.
\end{definition}

\bigskip
\noindent\textbf{The \assign\ problem and the \textsc{Stable Marriage} problem.}
On first sight, it may seem that the \assign\ problem is a special case of the \textsc{Stable Marriage} problem with indifferences (ties) \cite{irving1994stable} (the preferences of the agents remain the same, and all the items rank all the agents equally). We stress that, in fact, it is a very different problem since pareto optimality is not equivalent to the stability notions defined by Irving \cite{irving1994stable} for the the \textsc{Stable Marriage} problem with indifference. Consider the following instance of the \assign\ problem with three agents $a_{1},a_{2},a_{3}$, and three items $b_{1},b_{2},b_{3}$:
\begin{framed}[.4\textwidth]
\vspace{-0.5em}
\begin{itemize}
\item $a_{1}\ $: $\ b_{1}>b_{2}>b_{3}$
\item $a_{2}\ $: $\ b_{3}> b_{2}>b_{1}$
\item $a_{3}\ $: $\ b_{3}>b_{1}>b_{2}$
\end{itemize}
\vspace{-0.2em}
\end{framed}

We transform this instance to an instance of \textsc{Stable Marriage} with indifference by making all the items to rank $a_{1},a_{2}$ and $a_{3}$ equally as follows:
\begin{framed}[.75\textwidth]
\begin{itemize}
\item $a_{1}\ $: $\ b_{1}>b_{2}>b_{3}\ \ \ \ \ \ \ \ \ \ \bullet\ b_{1} :\ a_{1}=a_{2}=a_{3}$
\item $a_{2}\ $: $\ b_{3}> b_{2}>b_{1}\ \ \ \ \ \ \ \ \ \ \bullet\ b_{2} :\ a_{1}=a_{2}=a_{3}$
\item $a_{3}\ $: $\ b_{3}>b_{1}>b_{2}\ \ \ \ \ \ \ \ \ \ \bullet\ b_{3} :\ a_{1}=a_{2}=a_{3}$
\end{itemize}
\end{framed}

In the context of \textsc{Stable Marriage} with indifference, a matching $M$ is called {\em weakly stable} if there is no pair of agents each of whom strictly prefers the other to its matched partner. In addition, $M$ is called {\em super stable} if there is no pair of agents each of whom either strictly prefers the other to its partner or is indifferent between them. Moreover, $M$ is called {\em strongly stable} if there is no pair of agents such that the first one strictly prefers the second to its partner, and the second one strictly prefers the first to its partner or is indifferent between them (see \cite{irving1994stable} for the formal definitions).

Consider the matching $p$ that satisfies $p(a_{i}) = b_{i}$ for every $i \in \lbrace 1,2,3 \rbrace$. Then, $p$ is pareto optimal for the \assign\ instance, but it is not strongly stable  for the constructed \textsc{Stable Marriage} instance since $\lbrace a_{2},b_{3} \rbrace$ is a {\em blocking pair}. Suppose the matching $q$ that satisfies $q(a_{1})=b_{2}$, $q(a_{2}) = b_{1}$ and $q(a_{3}) = b_{3}$. Then, $q$ is not pareto optimal for the \assign\ instance because it admits the trading cycle $(a_{1},b_{2},a_{2},b_{1})$, but it is weakly stable and even super stable for the \textsc{Stable Marriage} instance. Thus, we conclude that the problems are different from each other.

We first give some well-known characterizations of assignments, and after that we introduce new concepts of optimality and three new multi-layered assignment problems. 

Intuitively, an assignment admits a {\em trading cycle} if there exists a set of agents who all benefit by exchanging their allocated items among themselves. For example, a simple trading cycle among two agents $a$ and $b$ occurs when agent $a$ prefers agent $b$'s item over its own item, and agent $b$ prefers agent $a$'s item over its own item. Both $a$ and $b$ would benefit from exchanging their items. Formally, a trading cycle is defined as follows.

\begin{definition}
An assignment $p$ admits a {\em trading cycle} $(a_{i_{0}},b_{j_{0}},a_{i_{1}},b_{j_{1}},\ldots,a_{i_{k-1}},b_{j_{k-1}})$ if for each $r \in \lbrace 0,\ldots,k-1 \rbrace$, we have that $p(a_{i_r})=b_{j_r}$ and $b_{j_r} <_{a_{i_{r}}} b_{j_{r+1 \Mod{k}}}$.
\end{definition}

\begin{definition}
An assignment $p$ admits a {\em self loop} if there exist an agent $a_{i}$ and an item $b_j$ such that $b_{j}$ is not allocated to any agent by $p$, and $p(a_{i}) <_{a_{i}} b_{j}$.
\end{definition}

\begin{proposition}[Folklore; see, e.g., Aziz et al.~\cite{ijcai2017-12,soton425734}]
\label{prop:po-iff-no-tc-and-sl}
An assignment $p$ is pareto optimal if and only if it does not admit trading cycles and self loops.
\end{proposition}

By this proposition, the problem of checking whether an assignment admits trading cycles or self loops can be reduced to the problem of checking whether the directed graph defined next contains cycles. For an instance $(A,I,P)$ and an assignment $p$, the corresponding {\em trading graph} is the directed graph defined as follows. Its vertex set is $A \cup I$, and there are three types of edges:
\begin{itemize}
\item For each $a \in A$ such that $p(a) \neq b_{\emptyset}$, there is a directed edge from $p(a)$ to $a$. Namely, each allocated item points to its owner.
\item For each agent $a \in A$, there is an edge from $a$ to all the items it prefers over its assigned item $p(a)$ (if $p(a)=b_{\emptyset}$, $a$ points to all its acceptable items).
\item Each item with no owner points to all the agents that accept it.
\end{itemize} 

\begin{proposition}[Folklore; see, e.g., Aziz et al.~\cite{ijcai2017-12,soton425734}]\label{prop:po-iff-no-cycles-in-td}
An assignment $p$ is pareto optimal if and only if its corresponding trading graph does not contain cycles.
\end{proposition}

\bigskip
\noindent\textbf{Example.} Suppose that $A = \lbrace a_{1},a_{2},a_{3},a_{4},a_{5} \rbrace$ and $I = \lbrace b_{1},b_{2},b_{3},b_{4},b_{5} \rbrace$. Assume that the preferences of the agents are defined as follows.

\begin{framed}[.5\textwidth]
\begin{itemize}
\vspace{-0.5em}
  \item $a_{1}\ $: $\ b_{4}\ >\ b_{1}\ >\ b_{2}\ >\ b_{5}$
  \item $a_{2}\ $: $\ b_{1}\ >\ b_{4}\ >\ b_{5}$
  \item $a_{3}\ $: $\ b_{2}\ >\ b_{1}$
  \item $a_{4}\ $: $\ b_{3}\ >\ b_{5}$
  \item $a_{5}\ $: $\ b_{5}$
\end{itemize}
\vspace{-0.4em}
\end{framed}

Let $p: A \rightarrow I \cup \lbrace b_{\emptyset} \rbrace$ be an assignment such that $p(a_{1}) = b_{2}$, $p(a_{2}) = b_{4}$, $p(a_{3}) = b_{1}$, $p(a_{4}) = b_{5}$, and $p(a_{5}) = b_{\emptyset}$.
The trading graph of the preference profile with respect to $p$ is:

\begin{center}
\tikzset{every picture/.style={line width=0.75pt}} 

\begin{tikzpicture}[x=0.75pt,y=0.75pt,yscale=-1,xscale=1]

\draw   (182,45) .. controls (182,31.19) and (193.19,20) .. (207,20) .. controls (220.81,20) and (232,31.19) .. (232,45) .. controls (232,58.81) and (220.81,70) .. (207,70) .. controls (193.19,70) and (182,58.81) .. (182,45) -- cycle ;
\draw   (99,44) .. controls (99,30.19) and (110.19,19) .. (124,19) .. controls (137.81,19) and (149,30.19) .. (149,44) .. controls (149,57.81) and (137.81,69) .. (124,69) .. controls (110.19,69) and (99,57.81) .. (99,44) -- cycle ;
\draw   (61,129.2) -- (102.5,129.2) -- (102.5,160) -- (61,160) -- cycle ;
\draw    (153,130) -- (125.29,71.71) ;
\draw [shift={(124,69)}, rotate = 424.57] [fill={rgb, 255:red, 0; green, 0; blue, 0 }  ][line width=0.08]  [draw opacity=0] (8.93,-4.29) -- (0,0) -- (8.93,4.29) -- cycle    ;
\draw    (102.5,129.2) -- (276.65,71.28) ;
\draw [shift={(279.5,70.33)}, rotate = 521.6] [fill={rgb, 255:red, 0; green, 0; blue, 0 }  ][line width=0.08]  [draw opacity=0] (8.93,-4.29) -- (0,0) -- (8.93,4.29) -- cycle    ;
\draw    (306,130.2) -- (232.63,57.77) ;
\draw [shift={(230.5,55.67)}, rotate = 404.63] [fill={rgb, 255:red, 0; green, 0; blue, 0 }  ][line width=0.08]  [draw opacity=0] (8.93,-4.29) -- (0,0) -- (8.93,4.29) -- cycle    ;
\draw   (343,45) .. controls (343,31.19) and (354.19,20) .. (368,20) .. controls (381.81,20) and (393,31.19) .. (393,45) .. controls (393,58.81) and (381.81,70) .. (368,70) .. controls (354.19,70) and (343,58.81) .. (343,45) -- cycle ;
\draw   (263,46) .. controls (263,32.19) and (274.19,21) .. (288,21) .. controls (301.81,21) and (313,32.19) .. (313,46) .. controls (313,59.81) and (301.81,71) .. (288,71) .. controls (274.19,71) and (263,59.81) .. (263,46) -- cycle ;
\draw   (146,130.2) -- (187.5,130.2) -- (187.5,161) -- (146,161) -- cycle ;
\draw   (224,130.2) -- (265.5,130.2) -- (265.5,161) -- (224,161) -- cycle ;
\draw   (306,130.2) -- (347.5,130.2) -- (347.5,161) -- (306,161) -- cycle ;
\draw   (383,129.2) -- (424.5,129.2) -- (424.5,160) -- (383,160) -- cycle ;
\draw    (124,69) -- (83.24,126.22) ;
\draw [shift={(81.5,128.67)}, rotate = 305.46] [fill={rgb, 255:red, 0; green, 0; blue, 0 }  ][line width=0.08]  [draw opacity=0] (8.93,-4.29) -- (0,0) -- (8.93,4.29) -- cycle    ;
\draw    (148.5,52.33) -- (303.31,128.87) ;
\draw [shift={(306,130.2)}, rotate = 206.31] [fill={rgb, 255:red, 0; green, 0; blue, 0 }  ][line width=0.08]  [draw opacity=0] (8.93,-4.29) -- (0,0) -- (8.93,4.29) -- cycle    ;
\draw    (189.5,62.33) -- (104.88,127.37) ;
\draw [shift={(102.5,129.2)}, rotate = 322.45] [fill={rgb, 255:red, 0; green, 0; blue, 0 }  ][line width=0.08]  [draw opacity=0] (8.93,-4.29) -- (0,0) -- (8.93,4.29) -- cycle    ;
\draw    (279.5,70.33) -- (190.01,128.56) ;
\draw [shift={(187.5,130.2)}, rotate = 326.95] [fill={rgb, 255:red, 0; green, 0; blue, 0 }  ][line width=0.08]  [draw opacity=0] (8.93,-4.29) -- (0,0) -- (8.93,4.29) -- cycle    ;
\draw    (355.5,66.33) -- (267.95,128.46) ;
\draw [shift={(265.5,130.2)}, rotate = 324.64] [fill={rgb, 255:red, 0; green, 0; blue, 0 }  ][line width=0.08]  [draw opacity=0] (8.93,-4.29) -- (0,0) -- (8.93,4.29) -- cycle    ;
\draw    (245.5,129.67) -- (343.07,59.42) ;
\draw [shift={(345.5,57.67)}, rotate = 504.25] [fill={rgb, 255:red, 0; green, 0; blue, 0 }  ][line width=0.08]  [draw opacity=0] (8.93,-4.29) -- (0,0) -- (8.93,4.29) -- cycle    ;
\draw    (383,129.2) -- (375.87,71.31) ;
\draw [shift={(375.5,68.33)}, rotate = 442.98] [fill={rgb, 255:red, 0; green, 0; blue, 0 }  ][line width=0.08]  [draw opacity=0] (8.93,-4.29) -- (0,0) -- (8.93,4.29) -- cycle    ;
\draw   (418,45) .. controls (418,31.19) and (429.19,20) .. (443,20) .. controls (456.81,20) and (468,31.19) .. (468,45) .. controls (468,58.81) and (456.81,70) .. (443,70) .. controls (429.19,70) and (418,58.81) .. (418,45) -- cycle ;
\draw    (439.5,69.33) -- (425.23,126.29) ;
\draw [shift={(424.5,129.2)}, rotate = 284.07] [fill={rgb, 255:red, 0; green, 0; blue, 0 }  ][line width=0.08]  [draw opacity=0] (8.93,-4.29) -- (0,0) -- (8.93,4.29) -- cycle    ;

\draw (117,35) node [anchor=north west][inner sep=0.75pt]   [align=left] {$\displaystyle a_{1}$};
\draw (199,36) node [anchor=north west][inner sep=0.75pt]   [align=left] {$\displaystyle a_{2}$};
\draw (360,36) node [anchor=north west][inner sep=0.75pt]   [align=left] {$\displaystyle a_{4}$};
\draw (280,37) node [anchor=north west][inner sep=0.75pt]   [align=left] {$\displaystyle a_{3}$};
\draw (74,137) node [anchor=north west][inner sep=0.75pt]   [align=left] {$\displaystyle b_{1}$};
\draw (158,137) node [anchor=north west][inner sep=0.75pt]   [align=left] {$\displaystyle b_{2}$};
\draw (238,137) node [anchor=north west][inner sep=0.75pt]   [align=left] {$\displaystyle b_{3}$};
\draw (320,137) node [anchor=north west][inner sep=0.75pt]   [align=left] {$\displaystyle b_{4}$};
\draw (396,137) node [anchor=north west][inner sep=0.75pt]   [align=left] {$\displaystyle b_{5}$};
\draw (436,36) node [anchor=north west][inner sep=0.75pt]   [align=left] {$\displaystyle a_{5}$};

\end{tikzpicture}
\end{center}

Observe that agents $a_{1}$, $a_{2}$ and $a_{3}$ admit the trading cycle $(a_{1},b_{2},a_{2},b_{4},a_{3},b_{1})$, and that agent $a_{4}$ admits a self loop with $b_{3}$. By \Prcref{prop:po-iff-no-cycles-in-td}, $p$ is not pareto optimal. If $a_{1}$, $a_{2}$, and $a_{3}$ exchange their items, $a_{4}$ gets $b_{3}$, and $a_{5}$ gets $b_{5}$, we have a pareto optimal assignment $q$ in which $q(a_{1}) = b_{4}$, $q(a_{2}) = b_{1}$, $q(a_{3}) = b_{2}$, $q(a_{4}) = b_{3}$ and $q(a_{5}) = b_{5}$.

A simple assignment mechanism is the greedy {\em serial dictatorship} mechanism. For a given permutation over the agents, the mechanism takes agents in turns, one in each turn, according to the permutation. That is, the agent which is ordered first allocates its most preferred item, the second allocates its most preferred item among the remaining items, and so on. If at some point, an agent has no available item to allocate in its preference list, it allocates $b_{\emptyset}$. We say that an assignment $p$ is a {\em possible outcome} of serial dictatorship if there exists a permutation $\pi$ such that applying serial dictatorship with respect to $\pi$ results in $p$.

\begin{proposition}[Abdulkadiroglu and Sönmez \cite{RePEc:ecm:emetrp:v:66:y:1998:i:3:p:689-702}] \label{prop:po-iff-sd} 
An assignment $p$ is pareto optimal if and only if it is a possible outcome of serial dictatorship.
\end{proposition}

\begin{corollary}\label{corollary:number-of-po-assign}
A pareto optimal assignment always exists and can be found in polynomial time, and the number of pareto optimal assignments for an instance with $n$ agents is at most $n!$. 
\end{corollary}

\Prcref{prop:po-iff-sd} yields a surjective mapping from the set of permutations on the agents to the set of pareto optimal assignments. This implies an upper bound of $n!$ on the number of pareto optimal assignments. Observe that this bound is tight: consider an instance where there is an equal number of agents and items, and all the agents share the same complete preference list. Observe that each permutation gives us a unique assignment after applying serial dictatorship with respect to it. Thus, there exist exactly $n!$ pareto optimal assignments.

\smallskip
\noindent\textbf{Generalization of the Assignment Problem.} We introduce a generalized version of the \assign\ problem where there are $\ell$ layers of preferences. For each $j \in [\ell]$, we refer to $<_{a_{i}}^{(j)}$ as $a_{i}$'s preference list in layer $j$. The {\em preference profile in layer $j$} is the collection of all the agents' preference lists in layer $j$, namely, $P_{j}=(<_{a_{1}}^{(j)},\ldots,<_{a_{n}}^{(j)})$. We say that assignment $p$ is {\em pareto optimal in layer $j$} if it is pareto optimal in the single-layered instance $(A,I,P_{j})$. To adapt the notion of optimality to the new context, we introduce a natural generalization requiring an assignment to be optimal in a given number of layers.

\begin{definition} \label{def:alpha-globally-optimal}
An assignment $p$ is {\em $\alpha$-globally optimal} for an instance $(A,I,P_{1},\ldots,P_{\ell})$ if there exist $\alpha$ layers $i_{1},\ldots ,i_{\alpha} \in [\ell]$ such that $p$ is pareto optimal in layer $i_{j}$ for each $j \in [\alpha]$.
\end{definition}

Thus, the new problem is defined as follows.

\defproblem {\textsc{$\alpha$-Globally Optimal Assignment}}{$(A,I,P_{1},\ldots,P_{\ell},\alpha)$, where $A$ is a set of $n$ agents, $I$ is a set of $m$ items, $P_{i}$ is the preference profile in layer $i$ for each $i \in [\ell]$, and $\alpha \in [\ell]$.}{Does an $\alpha$-globally optimal assignment exist?}

{\it Example.}
Consider the following instance, where the agent set is $A= \lbrace a_{1},a_{2},a_{3},a_{4}\rbrace$, the item set is $I = \lbrace b_{1},b_{2},b_{3},b_{4} \rbrace$, and there are four layers, defined as follows.

Layer 1:
\begin{framed}[.35\textwidth]
\begin{itemize}
\item $a_{1}\ $: $\ b_{1}$
\item $a_{2}\ $: $\ b_{3}> b_{2}>b_{1}$
\item $a_{3}\ $: $\ b_{3}>b_{1}$
\item $a_{4}\ $: $\ b_{2}>b_{1}>b_{3}$
\end{itemize}
\end{framed} 
Layer 2:
\begin{framed}[.35\textwidth]
\begin{itemize}
\item $a_{1}\ $: $\ b_{2} > b_{1}$
\item $a_{2}\ $: $\ b_{2}> b_{3}$
\item $a_{3}\ $: $\ b_{1}>b_{2}>b_{3}$
\item $a_{4}\ $: $\ b_{3}$
\end{itemize}
\end{framed} 
Layer 3:
\begin{framed}[.35\textwidth]
\begin{itemize}
\item $a_{1}\ $: $\ b_{2} > b_{1}$
\item $a_{2}\ $: $\ b_{4}> b_{2}>b_{1}$
\item $a_{3}\ $: $\ b_{1} > b_{3}$
\item $a_{4}\ $: $\ b_{2}>b_{1}>b_{3}$

\end{itemize}
\end{framed}
Layer 4:
\begin{framed}[.35\textwidth]
\begin{itemize}
\item $a_{1}\ $: $\ b_{3}>b_{1}>b_{2}$
\item $a_{2}\ $: $\ b_{1}> b_{2}$
\item $a_{3}\ $: $\ b_{2}>b_{3}$
\item $a_{4}\ $: $\ \emptyset$
\end{itemize}
\end{framed}

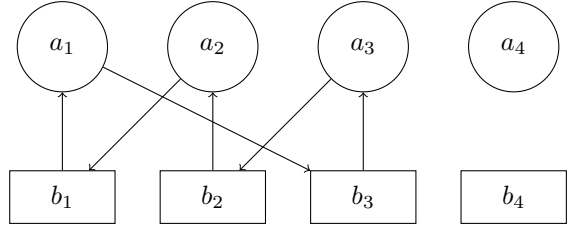
\begin{figure}[t]
\center
\begin{tikzpicture}
	\begin{pgfonlayer}{nodelayer}
		\node [style=vertex] (0) at (-9, 6) {$a_{1}$};
		\node [style=vertex] (1) at (-5, 6) {$a_{2}$};
		\node [style=vertex] (2) at (-1, 6) {$a_{3}$};
		\node [style=vertex] (3) at (3, 6) {$a_{4}$};
		\node [style=rec] (4) at (-9, 2) {$b_{1}$};
		\node [style=rec] (5) at (-5, 2) {$b_{2}$};
		\node [style=rec] (6) at (-1, 2) {$b_{3}$};
		\node [style=rec] (7) at (3, 2) {$b_{4}$};
	\end{pgfonlayer}
	\begin{pgfonlayer}{edgelayer}
		\draw [style=edge] (4) to (0);
		\draw [style=edge] (5) to (1);
		\draw [style=edge] (6) to (2);
		\draw [style=edge] (0) to (6);
		\draw [style=edge] (1) to (4);
		\draw [style=edge] (2) to (5);
	\end{pgfonlayer}
\end{tikzpicture}
\caption{The trading graph with respect to $p$ in the fourth layer, where $p$ admits the trading cycle $(a_{1},b_{1},a_{3},b_{3},a_{2},b_{2})$.}
\label{fig:example}
\end{figure}

Consider an assignment $p$ in which $a_{i}$ gets $b_{i}$ for every $i \in \lbrace 1,2,3 \rbrace$, and $a_{4}$ gets $b_{\emptyset}$. The assignment is $2$-globally optimal since it is pareto optimal in the first two layers. To see this, apply serial dictatorship in the first layer with respect to the permutation $(a_{1},a_{3},a_{2},a_{4})$, and in the second layer with respect to the permutation $(a_{2},a_{1},a_{3},a_{4})$, and verify that $p$ is the outcome of both runs. In contrast, $p$ is not pareto optimal in the third layer since it admits a self loop among $a_{2}$ and $b_{4}$ ($b_{4}$ is available and is preferred by $a_{2}$ over its assigned item $b_{2}$). Furthermore, assignment $p$ is not pareto optimal in the fourth layer because it admits a trading cycle $(a_{1},b_{1},a_{3},b_{3},a_{2},b_{2})$ (see Figure \ref{fig:example}): if $a_{1},a_{2}$ and $a_{3}$ trade their items, we get a new assignment $q$ in which $a_{1}$ gets $b_{3}$, $a_{2}$ gets $b_{1}$, and $a_{3}$ gets $b_{2}$; assignment $q$ pareto dominates $p$ in the fourth layer, and we can verify that it is also pareto optimal in this layer.      

\smallskip
\noindent\textbf{Parameterized Complexity.} Let $\Pi$ be an \NPH\ problem. In the framework of Parameterized Complexity, each instance of $\Pi$ is associated with a {\em parameter} $k$. Here, the goal is to confine the combinatorial explosion in the running time of an algorithm for $\Pi$ to depend only on $k$. Formally, we say that $\Pi$ is {\em fixed-parameter tractable (\FPT)} if any instance $(I, k)$ of $\Pi$ is solvable in time $f(k)\cdot |I|^{\OO(1)}$, where $f$ is an arbitrary computable function of $k$. A weaker request is that for every fixed $k$, the problem $\Pi$ would be solvable in polynomial time. Formally, we say that $\Pi$ is {\em slice-wise polynomial (\XP)} if any instance $(I, k)$ of $\Pi$ is solvable in time $f(k)\cdot |I|^{g(k)}$, where $f$ and $g$ are arbitrary computable functions of $k$. Nowadays, Parameterized Complexity supplies a rich toolkit to design \FPT\ and \XP\ algorithms.

Parameterized Complexity also provides methods to show that a problem is unlikely to be \FPT. The main technique is the one of parameterized reductions analogous to those employed in classical complexity. Here, the concept of \WO-hardness replaces the one of \NP-hardness, and for reductions we need not only construct an equivalent instance in \FPT\ time, but also ensure that the size of the parameter in the new instance depends only on the size of the parameter in the original one.

\begin{definition}\label{definition:parameterized-reduction}({\bf Parameterized Reduction})
Let $\Pi$ and $\Pi'$ be two parameterized problems. A {\em parameterized reduction} from $\Pi$ to $\Pi'$ is an algorithm that, given an instance $(I,k)$ of $\Pi$, outputs an instance $(I',k')$ of $\Pi'$ such that:
\begin{itemize}
\item $(I,k)$ is a \yes-instance\ of $\Pi$ if and only if $(I',k')$ is a \yes-instance\ of $\Pi'$.
\item $k' \leq g(k)$ for some computable function $g$.
\item The running time is $f(k) \cdot |\Pi|^{\OO(1)}$ for some computable function $f$.
\end{itemize}
\end{definition}

If there exists such a reduction transforming a problem known to be \WOH\ to another problem $\Pi$, then the problem $\Pi$ is \WO-hard as well. Central \WOH-problems include, for example, deciding whether a nondeterministic single-tape Turing machine accepts within $k$ steps, {\sc Clique} parameterized be solution size, and {\sc Independent Set} parameterized by solution size. To show that a problem $\Pi$ is not \XP\ unless \textsf{P}=\NP, it is sufficient to show that there exists a fixed $k$ such $\Pi$ is \NPH. Then, the problem is said to be \paraH.

A companion notion to that of fixed-parameter tractability is the one of a polynomial kernel. Formally, a parameterized problem $\Pi$ is said to admit a {\em polynomial compression} if there exists a (not necessarily parameterized) problem $\Pi'$ and a polynomial-time algorithm that given an instance $(I,k)$ of $\Pi$, outputs an equivalent instance $I'$ of $\Pi'$ (that is, $(I,k)$ is a \yes-instance of $\Pi$ if and only if $I'$ is a \yes-instance of $\Pi'$) such that $|I'|\leq p(k)$ where $p$ is some polynomial that depends only on $k$. In case $\Pi'=\Pi$, we further say that $\Pi$ admits a {\em polynomial kernel}. For more information on Parameterized Complexity, we refer the reader to recent books such as \cite{DBLP:series/txcs/DowneyF13,DBLP:books/sp/CyganFKLMPPS15,fomin2019kernelization}.

\smallskip
\noindent\textbf{Non-Existence of a Polynomial Compression.} Our proof of the ``unlikely existence'' of a polynomial kernel for \goassign\ relies on the well-known notion of cross-composition, defined as follows.

\begin{definition} [{\bf Cross-Composition}]\label{def:cross-comp} A (not parameterized) problem $\Pi$ {\em cross-composes} into a parameterized problem $\Pi'$ if there exists a polynomial-time algorithm, called a {\em cross-composition}, that given instances $I_1,I_2,\ldots,I_t$ of $\Pi$ for some $t \in \mathbb{N}$ that are of the same size $s$ for some $s\in\mathbb{N}$, outputs an instance $(I,k)$ of $\Pi'$ such that the following conditions are satisfied.
\begin{itemize}
\item $k\leq p(s + \log{t})$ for some polynomial $p$.
\item $(I,k)$ is a \yes-instance of $\Pi'$ if and only if at least one of the instances $I_1,I_2,\ldots,I_t$ is a \yes-instance of $\Pi$.
\end{itemize}
\end{definition}

\begin{proposition}[\cite{DBLP:journals/jcss/BodlaenderDFH09,DBLP:journals/siamdm/BodlaenderJK14}]\label{prop:noKern}
Let $\Pi$ be an \NP-hard (not parameterized) problem that cross-composes into a parameterized problem $\Pi'$. Then, $\Pi'$ does not admit a polynomial compression, unless \NP$\subseteq $\coNPpoly. 
\end{proposition}

To obtain (essentially) tight conditional lower bounds for the running times of algorithms, we rely on the well-known {\em Exponential-Time Hypothesis (\ETH)} \cite{DBLP:journals/jcss/ImpagliazzoP01,DBLP:journals/jcss/ImpagliazzoPZ01,DBLP:conf/iwpec/CalabroIP09}. To formalize the statement of \ETH, first recall that  given a formula $\varphi$ in conjuctive normal form (CNF) with $n$ variables and $m$ clauses, the task of {\sc CNF-SAT} is to decide whether there is a truth assignment to the variables that satisfies $\varphi$. In the {\sc $p$-CNF-SAT} problem, each clause is restricted to have at most $p$ literals. \ETH\ asserts that {\sc 3-CNF-SAT} cannot be solved in time $\OO(2^{o(n)})$.
\section{Fixed-Parameter Tractability and ETH Based Lower Bounds}
\label{sec:part1FPT}
We first prove that \goassign\ is \FPT\ with respect to the parameter $n=\nagents$ by presenting an algorithm with running time $\OO^{*}(n!)$. Afterwards, we prove that this is essentially the best possible running time for this parameter (and even for a larger parameter) under \ETH. Let us first prove the following lemma.
\begin{lemma} \label{lemma:verification}
Let $(A,I,P_{1},\ldots,P_{\ell},\alpha)$ be an instance of \goassign, and let $p$ be an assignment. Let $G_{i}$ denote the trading graph of $P_{i}$ with respect to $p$ for each $i \in [\ell]$. Then $p$ is $\alpha$-globally optimal for the instance if and only if there exist $\alpha$ trading graphs among $G_{1}\ldots,G_{\ell}$ that contain no cycles.
\end{lemma}
\begin{proof}
$(\Rightarrow)$ Assume that $p$ is $\alpha$-globally optimal. Then there exist $\alpha$ layers $i_{1},\ldots,i_{\alpha}$ such that for each $j \in [\alpha]$, $p$ is pareto optimal in $P_{i_{j}}$. By \Prcref{prop:po-iff-no-cycles-in-td}, for each $j \in [\alpha]$, $G_{i_{j}}$ does not contain cycles. 

$(\Leftarrow)$ Assume that there exist $\alpha$ layers $i_{1},\ldots,i_{\alpha}$, such that for each $j \in [\alpha]$, $G_{i_{j}}$ does not contain cycles. By \Prcref{prop:po-iff-no-cycles-in-td}, $p$ is pareto optimal in each $P_{i_{j}}$, implying that it is $\alpha$-globally optimal.\qed
\end{proof}

\begin{theorem}\label{theorem:goassign-is-fpt-num-of-agents}
There exists an $\OO^{*}(n!)$ algorithm for \goassign.
\end{theorem}
\begin{proof}
We present a brute-force algorithm (described formally in Algorithm \ref{algorithm:goassign}). The algorithm enumerates all possible pareto optimal assignments for each layer, using serial dictatorship. For each assignment, it constructs the corresponding trading graphs, and checks whether there exist $\alpha$ graphs with no cycles.

\begingroup
\removelatexerror
\begin{algorithm}[t]
\DontPrintSemicolon
\SetKwData{Left}{left}\SetKwData{This}{this}\SetKwData{Up}{up}
\SetKwFunction{Union}{Union}\SetKwFunction{FindCompress}{FindCompress}
\SetKwInOut{Input}{input}\SetKwInOut{Output}{output}
\Input{An instance $(A,I,P_{1},\ldots,P_{\ell},\alpha)$}
\Output{Does an $\alpha$-globally optimal assignment exist?}
\BlankLine
 \ForEach{$i \in [\ell]$}{
	\ForEach{permutation $\pi$ on $A$}{
		apply serial dictatorship on profile $P_{i}$ with respect to $\pi$, to obtain an assignment $p$\;
		$count \longleftarrow 1$\;
		\ForEach{$j \in [\ell] \setminus \lbrace i \rbrace$}{
			$G \longleftarrow$ $P_{j}$'s trading graph with respect to $p$\;
			\If {$G$ contains no cycles}{
				$count \longleftarrow count + 1$\;
			}		
		}
		\If{$count \geq \alpha$}{
			\KwRet{\yes}\;
		}		
	}
 }
 \KwRet{\no}\;
 \caption{Algorithm for \goassign.}
 \label{algorithm:goassign}
\end{algorithm}
\endgroup

The running time of the algorithm is $\OO^{*}(n!)$, since it iterates over $\ell n!$ assignments, and for each assignment, it takes polynomial time to construct $\ell-1$ trading graphs, and to count how many of them contain no cycles.

Let us now prove that the algorithm returns \yes\ if and only if the input is a \yes-instance.

$(\Rightarrow)$ Suppose that the algorithm returns \yes. This implies that there exist a layer $i \in [\ell]$ and a permutation $\pi$ on $A$, such that serial dictatorship on $P_{i}$ with respect to $\pi$ produces an assignment $p$ satisfying that there exist $\alpha - 1$ graphs among the trading graphs of $\lbrace P_{j} \mid j \in [\ell] , j \neq i \rbrace$ with respect to $p$ which do not contain cycles. By \Prcref{prop:po-iff-sd}, $p$ is pareto optimal in $P_{i}$, and by \Lecref{lemma:verification}, it is pareto optimal in $\alpha - 1$ preference profiles among $\lbrace P_{j} \mid j \in [\ell], j \neq i \rbrace$. This implies that $p$ is $\alpha$-globally optimal.

$(\Leftarrow)$ Suppose we are dealing with a \yesinstance; then, there exists an assignment $p$, and $\alpha$ layers $i_{1},\ldots,i_{\alpha}$ in which $p$ is pareto optimal. \Prcref{prop:po-iff-sd} implies that for each of these layers $i_{j}$, there exists a permutation $\pi_{j}$ such that applying serial dictatorship on $P_{i_{j}}$ with respect to $\pi_{j}$ results in $p$. Thus, when the algorithm reaches one of the layers $i_{1},\ldots,i_{\alpha}$ and its corresponding permutation for the first time, it generates $p$. By \Lecref{lemma:verification}, the algorithm verifies correctly that it is $\alpha$-globally optimal.\qed
\end{proof}

We first give a simple lemma that will help us to design a polynomial kernel for \goassign\ with respect to the parameter $n + \ell$. 
\begin{lemma}\label{lemma:first-k}
Let $(A,I,P)$ be an instance of the \assign\ problem where $|A|= n$. Then, for any agent $a \in A$ and pareto optimal assignment, $a$ is assigned to $b_{\emptyset}$ or to one of the $n$ most preferred items in its preference list.
\end{lemma}
\begin{proof}
Let $p$ be a pareto optimal assignment for $(A,I,P)$. By \Prcref{prop:po-iff-sd}, there exists a permutation $\pi = (a_{1},\ldots,a_{n})$ on $A$ such that applying serial dictatorship with respect to $\pi$ results in $p$. Let $i \in [n]$. When the mechanism is in the $i$-th step, it has already allocated at most $i-1$ items from $I$ to agents $a_{1},\ldots,a_{i-1}$. Hence, $p(a_{i}) = b_{\emptyset}$, or $p(a_{i})$ is the most preferred item of $a_{i}$ among the remaining items in its preference list, and is ranked at position $j \leq i \leq n$.\qed
\end{proof}

\begin{figure}[t]
\center
\includegraphics[width=0.8\textwidth, angle=0]{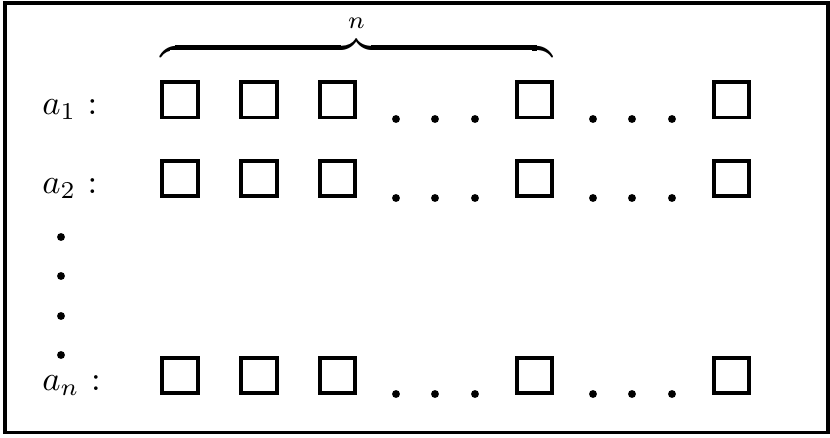}  
\caption{The kernel keeps only the $n$ most preferred items in each preference list; thus the size of the reduced instance is $\OO(\ell n ^{2})$.}
\label{fig:kernelagents}
\end{figure}
\begin{theorem}\label{theorem:poly-kernel}
\goassign\ admits a kernel of size $\OO(\ell n^{2})$.
\end{theorem}
\begin{proof}
Given an instance of \goassign\ $I_{1} = (A,I,P_{1},\ldots,P_{\ell},\alpha)$, the kernel reduces each preference profile $P_{i}$ to a preference profile $P_{i}'$ by keeping only the $n$ first-ranked items in each preference list (shown in Figure \ref{fig:kernelagents}). Let $I'$ be a set containing the items ranked in the first $n$ positions in some preference list in $I_{1}$; then, $|I'|\leq \ell n^{2}$. The resulting instance is $I_{2}=(A,I',P_{1}',\ldots,P_{\ell}',\alpha)$, and it satisfies $|I_{2}| = \OO(\ell n^{2})$. We claim that an assignment $p$ is $\alpha$-globally optimal for the instance $I_{1}$ if and only if it is $\alpha$-globally optimal for the instance $I_{2}$.

($\Rightarrow$) Let $p$ be an $\alpha$-globally optimal assignment for $I_{1}$; then, there exist $\alpha$ layers $i_{1},\ldots,i_{\alpha}$, such that for each $j \in [\alpha]$, $p$ is pareto optimal in $P_{i_{j}}$. By \Lecref{lemma:first-k}, each agent $a \in A$ is assigned an item that appears in the first $n$ items in its preference list in each $P_{i_{j}}$. Hence, $p$ assigns an acceptable item to $a$ in each $P_{i_{j}}'$. Moreover, it is pareto optimal in each $P_{i_{j}}'$, as otherwise, it would contradict $p$'s optimality in each $P_{i_{j}}$.

($\Leftarrow$) Let $p$ be an $\alpha$-globally optimal assignment for $I_{2}$. Then by \Prcref{prop:po-iff-sd}, there exist $\alpha$ profile-permutation pairs $(P_{i_{1}}',\pi_{1}),\ldots,(P_{i_{\alpha}}',\pi_{\alpha})$, such that for each $j \in [\alpha]$, applying serial dictatorship on profile $P_{i_{j}}'$ with respect to $\pi_{j}$ results in $p$. Observe that applying serial dictatorship on each $P_{i_{j}}$ with respect to the same permutation $\pi_{j}$ results also in $p$. Therefore, by \Prcref{prop:po-iff-sd}, $p$ is also $\alpha$-globally optimal for $I_{1}$.\qed
\end{proof}

\begin{corollary}
\goassign\ admits a polynomial kernel with respect to the parameter $k=n + \ell$.
\end{corollary}

Before we present an exponential kernel for \goassign\ with respect to the parameter $k = m + \ell$, let us define the following.

\begin{definition}
Let $Q = (A,I,P_{1},\ldots,P_{\ell},\alpha)$ be an instance of \goassign\ and $u \in A$. The {\em agent class} of $u$ in $Q$, $\CC(u,Q)$, is the tuple that contains the preference lists of $u$ in all the layers, namely, $\CC(u,Q)=(<_{u}^{1},\ldots,<_{u}^{\ell})$. Define $D = \lbrace B \subseteq I \times I | B\text{ is a linear order} \rbrace$. For a given tuple of length $\ell$ consisting of linear orderings on subsets of $I$, $C \subseteq D^{\ell}$, define $A(C,Q) = \lbrace a \in A \mid \CC(a,Q)=C \rbrace$.
\end{definition}

\begin{theorem}\label{theorem:itemsPlusEllkernel}
\goassign\ admits a kernel of size $\OO((m!)^{\ell+1})$. Thus, it is \FPT\ with respect to $m + \ell$.
\end{theorem}

\begin{proof}
Given an instance of \goassign\ $Q = (A,I,P_{1},\ldots,P_{\ell},\alpha)$, the kernelization algorithm works as follows (formally described in Algorithm \ref{algorithm:kernelForItemsPlusEll}): It removes from $A$ agents which share the same agent class together with all their preference lists, such that in the resulting instance there will be at most $m+1$ agents in the set $A(\CC(a,Q),Q)$, for each $a \in A$. Intuitively, the idea behind the correctness is that since there are $m$ items, at most $m$ agents in $A(\CC(a,Q),Q)$ will be assigned to items; we keep at most $m+1$ agents (rather than $m$) in each agent class to cover the case where an agent is assigned to $b_{\emptyset}$ and admits a self-loop. The kernelization algorithm clearly runs in a polynomial time. 
\begingroup
\removelatexerror
\begin{algorithm}[h]
\setstretch{1.2}
\DontPrintSemicolon
\BlankLine
 \ForEach{$a \in A$}{
 	construct $A(\CC(a,Q),Q)$\;
 	\If {$|A(\CC(a,Q),Q)|>m+1$}{
		remove $|A(\CC(a,Q),Q)|-m-1$ arbitrary agents from $A(\CC(a,Q),Q)$ together with all their preference lists\;
	}	
	}
	\KwRet{the reduced instance}\;
 \caption{Kernel for \goassign\ with respect to the parameter $k = m + \ell$.}
 \label{algorithm:kernelForItemsPlusEll}
\end{algorithm}
\endgroup

Assume that we run the kernel on $I_{1} = (A_{1},I,P_{1},\ldots,P_{\ell},\alpha)$ to obtain an instance $I_{2} = (A_{2},I,Q_{1},\ldots,Q_{\ell},\alpha)$. We first observe the following:

\begin{myclaim}
$|I_{2}| = \OO((m!)^{\ell+1})$.
\end{myclaim}
\begin{proof}
Note that there exist $\sum_{j=0}^{m}{{\binom{m}{j}} \cdot j!} = \sum_{j=0}^{m}{\frac{m!}{j!(m-j)!}j!} = m! \sum_{j=0}^{m}{\frac{1}{j!}} \leq e \cdot m! = \OO(m!)$ possible linear orderings of subsets of $I$. Then, there exist $\OO((m!)^{\ell})$ different combinations of such $\ell$ orderings, implying that there exist $\OO((m!)^{\ell})$ possible agent classes defined over the item set $I$. Since for each agent class $C$, $|A_{2}(C,I_{2})| \leq m+1$, we have that  $|A_{2}| = \Sigma_{\text{agent class C}}{|A_{2}(C,I_{2})|} \leq (m!)^{\ell} \cdot (m+1)$. Thus, $|I_{2}| = \OO((m!)^{\ell} \cdot (m+1)) = \OO((m!)^{\ell+1})$.\qed
\end{proof}

We now prove that $I_{1}$ is a \yes-instance\ if and only if $I_{2}$ is a \yes-instance.

($\Rightarrow$): Assume that there exists an $\alpha$-globally optimal assignment $p$ for $I_{1}$. Then, there exist $\alpha$ layers $i_{1},\ldots,i_{\alpha}$ of $I_{1}$ in which $p$ is pareto optimal. We create an assignment $q:A_{2} \rightarrow I \cup \lbrace b_{\emptyset} \rbrace$ for the reduced instance as follows: For each $a \in A_{2}$, let $p(A_{1}(\CC(a,I_{1}),I_{1}))$ denote the set of items allocated to the agents from $A_{1}(\CC(a,I_{1}),I_{1})$ by $p$. We allocate the items in $p(A_{1}(\CC(a,I_{1}),I_{1}))$ to agents in $A_{2}(\CC(a,I_{2}),I_{2})$ arbitrarily (observe that $\CC(a,I_{1}) = \CC(a,I_{2})$). Agents that do not have available items are assigned to $b_{\emptyset}$. First, observe that $q$ allocates all the items which are allocated by $p$ since there are at most $m$ items, and the algorithm keeps all or exactly $m+1$ agents from each set $A_{1}(\CC(a,I_{1}),I_{1})$. As a result, $q$ cannot admit self loops in layers $i_{1},\ldots,i_{\alpha}$ of $I_{2}$. Formally, the sets $A_{1}(\CC(a,I_{1}),I_{1})$ and $A_{2}(\CC(a,I_{2}),I_{2})$ satisfy $|A_{2}(\CC(a,I_{2}),I_{2})| \leq |A_{1}(\CC(a,I_{1}),I_{1})|$. Since the agents in these sets are allocated the same number of items by $p$ and $q$, if there exists an agent in $A_{2}(\CC(a,I_{2}),I_{2})$ that admits a self loop in $I_{2}$, there must exist an agent in $A_{1}(\CC(a,I_{1}),I_{1})$ that admits a self loop in $I_{1}$. Second, we claim that $q$ does not admit trading cycles in these layers. For the sake of contradiction, suppose there exists a layer $i_{j}$ in $I_{2}$, and $t$ agents $a_{1}',\ldots,a_{t}' \in A_{2}$ that admit a trading cycle $(a_{1}',q(a_{1}'),\ldots,a_{t}',q(a_{t}'))$ in $Q_{i_{j}}$. By the construction of $q$, notice that there exist $t$ agents $a_{1},\ldots,a_{t} \in A_{1}$, such that for each $i \in [t]$, $\CC(a_{i},I_{1})$ = $\CC(a_{i}',I_{2})$, and $q(a_{i}')=p(a_{i})$. Then, $p$ admits the trading cycle  $(a_{1},p(a_{1}),\ldots,a_{t},p(a_{t}))$ in $P_{i_{j}}$. This gives a contradiction to \Prcref{prop:po-iff-no-tc-and-sl}.

($\Leftarrow$): Assume that there exists an $\alpha$-globally optimal assignment $q$ for $I_{2}$. Then there exist $\alpha$ layers $i_{1},\ldots,i_{\alpha}$ in $I_{2}$ in which $q$ is pareto optimal. We denote an assignment $p$ for $I_{1}$ by $
  p(a) =
  \begin{cases}
    q(a) \ \ a \in A_{2}\\
    b_{\emptyset} \text{   otherwise}
  \end{cases}
$, 
and we claim that $p$ is pareto optimal in layers $i_{1},\ldots,i_{\alpha}$ in $I_{1}$. By the construction of $p$, for each $a_{1} \in A_{1} \setminus A_{2}$, there exists an agent $a_{2} \in A_{2}$ such that $\CC(a_{1},I_{1}) = \CC(a_{2},I_{2})$ and $p(a_{1}) = q(a_{2})$. Namely, there exists a mapping $f$ from agents in $A_{1}$ to agents in $A_{2}$ such that for each $a_{1} \in A_{1}$, $\CC(a_{1},I_{1}) = \CC(f(a_{1}),I_{2})$ and $p(a_{1}) = q(f(a_{1}))$. If $p$ admits a trading cycle $(a_{1},p(a_{1}), \ldots , a_{r},p(a_{r}))$ in some layer $i_{j}$ of $I_{1}$, then $q$ admits the trading cycle $(f(a_{1}),q(f(a_{1})), \ldots , f(a_{r}),q(f(a_{r})))$ in layer $i_{j}$ of $I_{2}$. If $p$ admits a self loop in layer $i_{j}$ of $I_{1}$ with agent $a_{1} \in A_{1}$, then $q$ admits a self loop with agent $f(a_{1})$ in layer $i_{j}$ of $I_{2}$. Thus by \Prcref{prop:po-iff-no-tc-and-sl}, we conclude that $p$ is $\alpha$-globally optimal in $I_{1}$.\qed
\end{proof}

\begin{corollary}\label{corollary:FptItemsPlusEll}
\goassign\ is \FPT\ with respect to the parameter $k = m + \ell$.
\end{corollary}

By Theorems \ref{theorem:itemsPlusEllkernel} and \ref{theorem:goassign-is-fpt-num-of-agents}, we conclude the following. 
\begin{corollary}\label{corollary:items-plus-ell-algorithm}
\goassign\ is solvable in time $\OO^{*}(((m!)^{\ell+1})!)$.
\end{corollary}

\begin{theorem} \label{xp-for-items}
\goassign\ is solvable in time $(nm)^{\OO(m)}$.
\end{theorem}
\begin{proof}
We present a simple brute-force algorithm. The algorithm simply iterates over all subsets of items $I' \subseteq I$. For each subset, it iterates over all subsets $A' \subseteq A$ such that $|A'|=|I'|$. For each $a \notin A'$, the algorithm allocates $b_\emptyset$, and it tries all possible $|I'|!$ different ways to allocate the items in $I'$ to the agents in $A'$ (it skips allocations that allocate items that are not acceptable by their owners in more than $\ell-\alpha + 1$ layers). The algorithm constructs the corresponding trading graphs, and verifies in polynomial time whether the current assignment is $\alpha$-globally optimal.
Hence, the running time of the algorithm is $\sum_{t=0}^{m} \binom{m}{t} \cdot \binom{n}{t} \cdot t! \cdot (n+m)^{\OO(1)} \leq m \cdot 2^{m}\cdot n^{\frac{m}{2}} \cdot m! \cdot (n+m)^{\OO(1)} = (nm)^{\OO(m)}$.\qed
\end{proof}

Before we continue with our next results, let us discuss a simple property that will help in many of our proofs.

\begin{definition}
Let $(A,I,P)$ be an instance of the \assign\ problem and suppose that $P = \lbrace <_{a} \mid a \in A \rbrace$. We say that agents $a_{1}, a_{2} \in A$ {\em respect each other} if there exists a linear order on a subset of $I$, $\biglefttriangle \subseteq I \times I$, such that both $<_{a_{1}} \subseteq \biglefttriangle$ and $<_{a_{2}} \subseteq \biglefttriangle$.
\end{definition}

\begin{lemma}\label{lemma:respectedAgents}
Let $(A,I,P)$ be an instance of the \assign\ problem, such that there exist agents $a_{1},\ldots,a_{r} \in A$ where for each $i,j \in [r]$, $a_{i}$ and $a_{j}$ respect each other. Then, for every assignment $p : A \rightarrow I \cup \lbrace b_{\emptyset} \rbrace$, $p$ does not admit a trading cycle among the agents $a_{1},\ldots,a_{r}$.
\end{lemma}
\begin{proof}
Towards a contradiction, suppose there exist an assignment $p$ which admits a trading cycle $(a_{1},p(a_{1}),\ldots,a_{r},p(a_{r}))$ (notice that $r \geq 2$). Since $a_{1},\ldots,a_{r}$ pairwise respect each other, there exists a linear order $\biglefttriangle \subseteq I \times I$, such that for each $i \in [r]$, $<_{a_{i}} \subseteq \biglefttriangle$. This implies that $p(a_{1}) \biglefttriangle p(a_{2})\biglefttriangle\ldots\biglefttriangle p(a_{r})$. Since $p(a_{r}) <_{a_{r}} p(a_{1})$, we have that $p(a_{r}) \biglefttriangle p(a_{1})$, a contradiction to $\biglefttriangle$ being a linear order.\qed
\end{proof}

We now prove that $\Omega^{*}(k!)$ is a (tight) lower bound on the running time of any algorithm for \goassign\ under the \ETH, even for larger parameters than $n$ such as $k= n + m + \alpha$ and $k = n + m + (\ell - \alpha)$. So, the algorithm in \Thcref{theorem:goassign-is-fpt-num-of-agents} is optimal (in terms of running time).

\begin{proposition}[Cygan et al.~\cite{DBLP:books/sp/CyganFKLMPPS15}] \label{proposition:eth-reduction}
Suppose that there is a polynomial-time parameterized reduction from problem $A$ to problem $B$ such that if the parameter of an
instance of $A$ is $k$, then the parameter of the constructed instance of $B$ is at most $g(k)$ for some nondecreasing function $g$. Then an $\OO^{*}(2^{o(f(k))})$-time algorithm for $B$ for some nondecreasing function $f$ implies an $\OO^{*}(2^{o(f(g(k)))})$-time algorithm for $A$.
\end{proposition}

\begin{theorem}\label{theorem:optimal-agents-items-alpha}
Unless \ETH\ fails, there does not exist an algorithm for \goassign\ with running time $\OO^{*}(2^{o(k \log{k})})$ where $k=n+m+\alpha$.
\end{theorem}
\begin{proof}
We provide a linear parameter reduction from \pc. In the \pc\ problem, we are given a graph $G$ where the vertices are elements of a $k \times k$ table, namely, $V(G)=[k] \times [k]$. Then, a {\em $k \times k$-permutation clique} is a clique of size $k$ in $G$ that contains exactly one vertex from each row and exactly one vertex from each column. In other words, there exists a permutation $\pi$ on $[k]$ such that the vertices of the clique are $(1,\pi(1)),\ldots,(k,\pi(k))$. The task is to decide whether there exists a $k \times k$-permutation clique in $G$. Lokshtanov et al.~\cite{doi:10.1137/16M1104834} proved that there is no $\OO^{*}(2^{o(k \log{k})})$-time algorithm for \pc, unless \ETH\ fails. 

Let $(G,k)$ be an instance of \pc. We create an agent $a_{i}$ for each row $i \in [k]$, and an item $b_{j}$ for each column $j \in [k]$. We construct an instance of \goassign\ consisting of $k^{2}$ layers, each corresponds to a row-column pair $(i,j)$, containing the preference profile $P_{(i,j)}$ defined as follows.
\begin{framed}
\begin{itemize}
\item $a_{i}\ $: $\ b_j$
\item $a_{r}\ $:$\ \lbrace b_{q} \mid \lbrace (i,j),(r,q) \rbrace \in E(G),q \neq j \rbrace$ (sorted in ascending order by $q$) $\ \ \forall r \in [k]\setminus \lbrace i \rbrace$.
\end{itemize}
\end{framed}

We finally set $\alpha = k$. We now prove that there exists a $k \times k$-permutation clique in $G$ if and only if there exists a $k$-globally optimal assignment for the constructed instance.

\begin{figure}[t]
\centering
\scalebox{0.7}{
\input{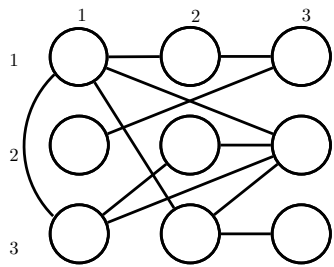}
}
\caption{For the presented graph, the constructed instance of \goassign\ consists of 9 layers, one for each row-column pair; $\lbrace (1,1),(2,3),(3,2) \rbrace$ is a $3 \times 3$-permutation clique. Thus, the assignment $p$ which satisfies $p(a_{1}) = b_{1} , p(a_{2}) = b_{3}$, and $p(a_{3}) = b_{2}$  is $3$-globally optimal for the constructed instance, and is pareto optimal in the profiles $P_{(1,1)},P_{(2,3)}$, and $P_{(3,2)}$.}
\label{fig:pc1}
\end{figure}

$(\Rightarrow)$ Suppose there exists a permutation $\pi$ for $[k]$ such that $(1,\pi(1)),\ldots,(k,\pi(k))$ form a clique in $G$. We define an assignment $p$ by $p(a_{i})=b_{\pi(i)}$ for each $i \in [k]$, and we claim that $p$ is pareto optimal in each $P_{(i,\pi(i))}$ (see Figure \ref{fig:pc1}).
Observe that for each $i \in [k]$, $b_{\pi(i)}$ is acceptable by $a_{i}$ in $P_{(i,\pi(i))}$ and in all profiles $P_{(j,\pi(j))}$ such that $j \in [k] \setminus \lbrace i \rbrace$ since there is an edge between $(i,\pi(i))$ and each $(j,\pi(j))$. Moreover, each $P_{(j,\pi(j))}$ contains no self loops due to the fact that all the items are allocated. Since we sorted each preference list in an ascending order by the item indices, all the agents respect each other in each preference profile and by \Lecref{lemma:respectedAgents}, $p$ does not admit trading cycles in any layer.

$(\Leftarrow)$ Suppose there exists a $k$-globally optimal assignment $p$ for the constructed instance. Note that if $p$ is pareto optimal in some profile $P_{(i,j)}$, it must satisfy $p(a_{i})=b_{j}$, as otherwise, $a_{i}$ would admit a self loop ($b_{j}$ is not acceptable by any agent but $a_{i}$ in $P_{(i,j)}$). Hence, we have that for each $i \in [k]$, $p$ is pareto optimal in at most one profile among $P_{(i,1)},\ldots,P_{(i,k)}$ and in at most one profile among $P_{(1,i)},\ldots,P_{(k,i)}$. Since we set $\alpha = k$, we have that for each $i \in [k]$, $p$ is pareto optimal in exactly one profile among $P_{(i,1)},\ldots,P_{(i,k)}$ and in exactly one profile among $P_{(1,i)},\ldots,P_{(k,i)}$. This implies that there exists a permutation $\pi$ on $[k]$ such that $p$ is pareto optimal in $P_{(i,\pi(i))}$ for each $i \in [k]$. We claim that $\lbrace(i,\pi(i)) \mid i \in [k]\rbrace$ is the vertex set of a clique in $G$. Towards a contradiction, suppose that there exist two different rows $i_{1}$ and $i_{2}$ such that $\lbrace (i_{1},\pi(i_{1})), (i_{2},\pi(i_{2})) \rbrace \notin E(G)$. By the construction of the preference lists, observe that $b_{\pi(i_{2})}$ is not acceptable by $a_{i_{2}}$ in $P_{(i_{1},\pi(i_{1}))}$. Therefore, $p$ is not a legal assignment for $P_{(i_{1},\pi(i_{1}))}$, a contradiction to its optimality.

It holds that $n + m+ \alpha= \OO(k)$. Thus, by \Prcref{proposition:eth-reduction}, we conclude that there is no $\OO^{*}(2^{o(k \log{k})})$-time algorithm for \goassign, unless \ETH\ fails.\qed
\end{proof}

We now prove a that the same result holds also for the parameter $n + m + (\ell - \alpha)$.

\begin{theorem}\label{theorem:optimal-agents-items-ell-alpha}
Unless \ETH\ fails, there does not exists an algorithm for \goassign\ with running time $\OO^{*}(2^{o(k \log{k})})$, where $k=n+m+(\ell -\alpha)$.
\end{theorem}
\begin{proof}
We provide a different linear parameter reduction from \pc\ for the parameter $n + m + (\ell -\alpha)$.
Let $(G=(V,E),k)$ be an instance of \pc\ (recall that $V = [k] \times [k]$).
We create $2k$ agents $a_{1},\ldots,a_{k}$, $c_{1},\ldots,c_{k}$, and $2k$ items $b_{1},\ldots,b_{k}$, $d_{1},\ldots,d_{k}$. Then, we construct an instance of \goassign\ consisting of $\ell = 1+k+{\binom{k^{2}}{2}}-|E|$ layers as follows. Intuitively, the first layer requires the agents $a_{1},\ldots,a_{k}$ to accept only the items $b_{1},\ldots,b_{k}$, and the agents $c_{1},\ldots,c_{k}$ to accept only the items $d_{1},\ldots,d_{k}$. Formally, it is defined as follows:

\begin{framed}[.5\textwidth]
\begin{itemize}
\item $a_{i}\ $: $\ b_{1}>\ldots>b_{k}\ \ \forall i \in [k]$
\item $c_{i}\ $: $\ d_{1}>\ldots>d_{k}\ \ \forall i \in [k]$
\end{itemize}
\end{framed}

The next $k$ layers require $a_{i}$ and $c_{i}$ to ``agree'' with each other. Namely, $p(a_{i}) = b_{j}$ if and only if $p(c_{i})=d_{j}$. To this end, for each $j \in [k]$, we construct the following preference profile:

\begin{framed}
\begin{itemize}
\item $a_{i}\ $: $\ \lbrace b_{q} \mid q \in [k] \rbrace \setminus \lbrace b_{j} \rbrace$ (sorted in ascending order by $q$) $ >  \lbrace d_{q} \mid q \in [k] \rbrace > b_{j}\ \ \forall i \in [k]$ 
\item $c_{i}\ $: $\ d_{j}  > b_{j} > \lbrace d_{q} \mid q \in [k] \rbrace \setminus \lbrace d_{j} \rbrace$ (sorted in ascending order by $q$) $\forall i \in [k]$
\end{itemize}
\end{framed}

We construct additional ${\binom{k^{2}}{2}}-|E|$ layers (notice that there are exactly ${\binom{k^{2}}{2}}-|E|$ pairs of vertices $\lbrace (i_{1},j_{1}),(i_{2},j_{2}) \rbrace$ that are not adjacent in $G$). Each layer is defined with respect to a vertex pair $(i_{1},j_{1}) \neq (i_{2},j_{2})$ such that $\lbrace (i_{1},j_{1}),(i_{2},j_{2}) \rbrace \notin E$ as follows:

\begin{framed}
\begin{itemize}
\item $a_{i_{1}}\ $: $\ \lbrace b_{q} \mid q \in [k] \rbrace \setminus \lbrace b_{j_{1}} \rbrace$ (sorted in ascending order by $q$) $ >  d_{j_{2}} > b_{j_{1}}$
\item $c_{i_{2}}\ $: $\ \lbrace d_{q} \mid q \in [k] \rbrace \setminus \lbrace d_{j_{2}} \rbrace$ (sorted in ascending order by $q$) $ >  b_{j_{1}} > d_{j_{2}}$
\item $a_{i}\ $: $\ \lbrace b_{q} \mid q \in [k] \rbrace \setminus \lbrace b_{j_{1}} \rbrace$ (sorted in ascending order by $q$) $ > b_{j_{1}}\ \ \forall i \in [k]\setminus \lbrace i_{1} \rbrace$
\item $c_{i}\ $: $\ \lbrace d_{q} \mid q \in [k] \rbrace \setminus \lbrace d_{j_{2}} \rbrace$ (sorted in ascending order by $q$) $ > d_{j_{2}}\ \ \forall i \in [k] \setminus \lbrace i_{2} \rbrace$
\end{itemize}
\end{framed}

Finally, we set $\alpha = \ell$. The reduction can be done in a polynomial time using simple polynomial-time operations on graphs. Before we prove the correctness of the construction, let us prove the following:

\begin{myclaim} \label{claim:claim1}
Let $p$ be an $\ell$-globally optimal assignment for the constructed instance. Then, for each $i \in [k]$, $p(a_{i}) \in \lbrace b_{1} \ldots b_{k} \rbrace$, and $p(c_{i}) \in \lbrace d_{1},\ldots,d_{k} \rbrace$.
\end{myclaim}
\begin{proof}
Observe that in the first layer, $b_{1},\ldots,b_{k}$ are all only acceptable by $a_{1},\ldots,a_{k}$ and $d_{1},\ldots,d_{k}$ are all only acceptable by $c_{1},\ldots,c_{k}$. Hence, $p$ must allocate all the items. Towards a contradiction, suppose that there exists $b_{j}$ that is not assigned to any agent (the case when there exists $d_{j}$ that is not assigned to any agent is symmetric). This implies that there exists an agent $a_{i}$ that satisfies $p(a_{i})=b_{\emptyset}$ and therefore admits a self loop with $b_{j}$, a contradiction to the optimality of $p$ in the first layer.\qed 
\end{proof}

\begin{myclaim} \label{claim:claim2}
Let $p$ be an $\ell$-globally optimal assignment for the constructed instance. Then, for every $i,j \in [k]$, $p(a_{i})=b_{j}$ if and only if $p(c_{i})=d_{j}$.
\end{myclaim}
\begin{proof}
Towards a contradiction, suppose there exist $i,j \in [k]$ such that $p(a_{i})=b_{j}$ and $p(c_{i}) \neq d_{j}$ (the second case where $p(c_{i})=d_{j}$ and $p(a_{i}) \neq b_{j}$ is symmetric). By \Clcref{claim:claim1}, $p(c_{i}) = d_{q}$ such that $q \neq j$. Note that $p$ admits the trading cycle $(a_{i},b_{j},c_{i},d_{q})$ in layer $j$, a contradiction to the $\ell$-global optimality of $p$.\qed
\end{proof}

\begin{myclaim} \label{claim:claim3}
Let $p$ be an $\ell$-globally optimal assignment for the constructed instance. Let $\pi: [k] \rightarrow [k]$ be defined as follows: $\pi(i)= j$ such that $p(a_{i})=b_{j}$ for each $i \in [k]$. Let $V' = \lbrace (i,\pi(i)) \mid i \in [k] \rbrace$. Then $V'$ is the vertex set of a $k \times k$-permutation clique in $G$.
\end{myclaim}
\begin{proof}
Observe that $\pi$ is a well-defined permutation on $[k]$ since $p$ allocates all the items in $\lbrace b_{1},\ldots,b_{k} \rbrace$. Towards a contradiction, suppose there exist different $i_{1},i_{2} \in [k]$ such that $\lbrace (i_{1},\pi(i_{1})),(i_{2},\pi(i_{2})) \rbrace \notin E$. First, by Claims \ref{claim:claim1} and \ref{claim:claim2}, $p(c_{i_{1}}) = d_{\pi(i_{1})}$ and $p(c_{i_{2}}) = d_{\pi(i_{2})}$. Second, by the construction of the instance, there exists a layer among the ${\binom{k^{2}}{2}} - |E|$ last layers which corresponds to the non-adjacent pair $\lbrace (i_{1},\pi(i_{1})),(i_{2},\pi(i_{2}) \rbrace$. Note that $p$ admits the trading cycle $(a_{i_{1}},b_{\pi(i_{1})},c_{i_{2}},d_{\pi(i_{2})})$ in this layer, a contradiction.\qed
\end{proof}

We now prove that there exists a $k \times k$-permutation clique in $G$ if and only if there exists an $\ell$-globally optimal assignment for the constructed instance.

($\Rightarrow$): Let $\pi$ be a permutation on $[k]$ such that $V' = \lbrace (1,\pi(1)),\ldots,(k,\pi(k)) \rbrace$ is the vertex set of a $k \times k$-permutation clique in $G$. Let us define an assignment $p$ by $p(a_{i}) = b_{\pi(i)}$, and $p(c_{i}) = d_{\pi(i)}$ for each $i \in [k]$. Observe that both $\lbrace a_{1},\ldots,a_{k} \rbrace$ and $\lbrace c_{1},\ldots,c_{k} \rbrace$ contain agents that pairwise respect each other. Since each $d_{j}$ is not acceptable by any $a_{i}$, and each $b_{j}$ is not acceptable by any $c_{i}$, we have by \Lecref{lemma:respectedAgents} that $p$ does not admit trading cycles in the first layer. Moreover, since all the items are allocated, it does not admit self loops in the first layer, implying that $p$ is pareto optimal in the this layer. Furthermore, $p$ is also pareto optimal in the next $k$ layers since $a_{i}$ and $c_{i}$ ``agree'' with each other for each $i \in [k]$. Formally, for every $j \in [k]$, both $\lbrace a_{1},\ldots,a_{k} \rbrace$ and $\lbrace c_{1},\ldots,c_{k} \rbrace$ contain agents that pairwise respect each other in layer $1+j$. By the construction of layer $1+j$, and by \Lecref{lemma:respectedAgents}, the only possible trading cycle in this layer occurs when $p(a_{i})=b_{j}$, and $p(c_{i}) = d_{q}$ such that $q \neq j$. We claim that $p$ is also pareto optimal in the last ${\binom{k^{2}}{2}} - |E|$ layers. Suppose, for the sake of contradiction, that there exists a pair $\lbrace (q_{1},r_{1}),(q_{2},r_{2}) \rbrace \notin E$, such that $p$ is not pareto optimal in its corresponding layer. By \Lecref{lemma:respectedAgents}, the only possible trading cycle in this layer is $(a_{q_{1}},b_{r_{1}},c_{q_{2}},d_{r_{2}})$. Thus, we have that $p(a_{q_{1}})=b_{r_{1}}$ and $p(c_{q_{2}})=d_{r_{2}}$. By the construction of $p$, both $(q_{1},r_{1})$ and $(q_{2},r_{2})$ are vertices in $V'$, a contradiction. 

($\Leftarrow$): Let $p$ be an $\ell$-globally optimal assignment for the constructed instance. By \Clcref{claim:claim3}, there exists a $k \times k$-permutation clique in $G$.\qed 
\end{proof}

\begin{theorem}\label{w-one-hardness-items-alpha}
\goassign\ is \WOH\ for the parameter $k=m+\alpha$.
\end{theorem}
\begin{proof}

We provide a parameterized (and also polynomial time) reduction from the \WOH\ problem \MCIS\ to \goassign. The input of \MCIS\ consists of an undirected graph $G=(V,E)$, and a coloring $c:\ V \rightarrow [k]$ that colors the vertices in $V$ with $k$ colors. The task is to decide whether $G$ admits a {\em multicolored} independent set of size $k$, which is an independent set (i.e.~a vertex subset with pair-wise non-adjacent vertices) $V' \subseteq V$ that satisfies $\lbrace c(v') \mid v' \in V' \rbrace = [k]$ and $|V'| = k$. 

Given an instance $(G=(V,E),c)$ of \MCIS, denote $V= \lbrace v_{1},\ldots,v_{n} \rbrace$. We construct an instance of \goassign\ consisting of $n$ layers, with agent set $A = \lbrace a_{1},\ldots,a_{n} \rbrace$, and item set $I= \lbrace {b_{1},\ldots,b_{k}} \rbrace$. Intuitively, each agent $a_{i}$ corresponds to the vertex $v_{i}$, and each item $b_{i}$ corresponds to the color $i \in [k]$. We construct the instance such that the items are allocated to agents whose corresponding vertices form a multicolored independent set. For each $i \in [n]$, we construct the preference profile in layer $i$, denoted by $P_{i}$, as follows:

\begin{framed}
\begin{itemize}
  \item $a_{i}\ $: $b_{c(v_{i})}$
  \item $a_{j}\ :\ \emptyset\ \ \forall j \in [n] \setminus \lbrace i \rbrace$ and $c(v_{j})=c(v_{i})$
  \item $a_{j}\ :\ \emptyset\ \ \forall j \in [n] \setminus \lbrace i \rbrace$ such that $\lbrace v_{i},v_{j} \rbrace \in E$ and $c(v_{j}) \neq c(v_{i})$
  \item $a_{j}\ :\  b_{c(v_{j})}\ \ \forall j \in [n] \setminus \lbrace i \rbrace$ such that $\lbrace v_{i} , v_{j} \rbrace \notin E$ and $c(v_{j}) \neq c(v_{i})$
\end{itemize}
\end{framed}
Finally, we set $\alpha = k$. The construction clearly can be done in a polynomial time. Before we prove the correctness of the reduction, let us prove the following claim.
\begin{myclaim}\label{claim:reduction-multicolored-claim}
Let $p: A \rightarrow I \cup \lbrace b_{\emptyset} \rbrace$ be an assignment. Then for any $i \in [n]$, if $p$ is pareto optimal in layer $i$, then it satisfies that $p(a_{i})=b_{c(v_{i})}$, and for each $j \in [n] \setminus \lbrace i \rbrace$ such that $\lbrace v_{i},v_{j} \rbrace \in E$,  $p(a_{j})=b_{\emptyset}$. 
\end{myclaim}
\begin{proof}
Let $i \in [n]$. Assume that $p$ is pareto optimal in layer $i$. Note that $a_{i}$ is the only agent that accepts $b_{c(v_{i})}$ in $P_{i}$ and this is the only item it accepts. Hence, $p$ must allocate $b_{i}$ to $a_{i}$, as otherwise, $a_{i}$ would be a part of a self loop. Moreover, for each $j \in [n]$ such that $\lbrace v_{i},v_{j} \rbrace \in E$, the preference list of $a_{j}$ in $P_{i}$ is empty, implying that $p(a_{j}) = b_{\emptyset}$.\qed 
\end{proof}

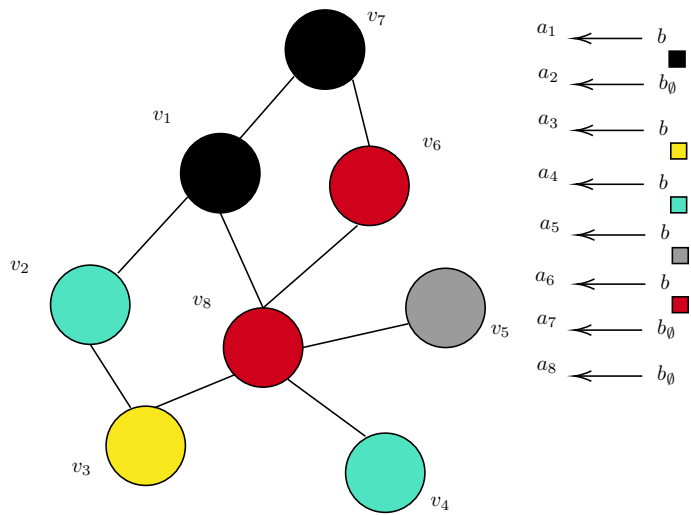
\begin{figure}[t]
\centering
\scalebox{0.8}{
\tikzset{every picture/.style={line width=0.75pt}} 

\begin{tikzpicture}[x=0.75pt,y=0.75pt,yscale=-1,xscale=1]

\draw  [fill={rgb, 255:red, 0; green, 0; blue, 0 }  ,fill opacity=1 ] (256,41) .. controls (256,27.19) and (267.19,16) .. (281,16) .. controls (294.81,16) and (306,27.19) .. (306,41) .. controls (306,54.81) and (294.81,66) .. (281,66) .. controls (267.19,66) and (256,54.81) .. (256,41) -- cycle ;
\draw  [fill={rgb, 255:red, 0; green, 0; blue, 0 }  ,fill opacity=1 ] (190,119) .. controls (190,105.19) and (201.19,94) .. (215,94) .. controls (228.81,94) and (240,105.19) .. (240,119) .. controls (240,132.81) and (228.81,144) .. (215,144) .. controls (201.19,144) and (190,132.81) .. (190,119) -- cycle ;
\draw  [fill={rgb, 255:red, 208; green, 2; blue, 27 }  ,fill opacity=1 ] (284,127) .. controls (284,113.19) and (295.19,102) .. (309,102) .. controls (322.81,102) and (334,113.19) .. (334,127) .. controls (334,140.81) and (322.81,152) .. (309,152) .. controls (295.19,152) and (284,140.81) .. (284,127) -- cycle ;
\draw  [fill={rgb, 255:red, 80; green, 227; blue, 194 }  ,fill opacity=1 ] (108,202) .. controls (108,188.19) and (119.19,177) .. (133,177) .. controls (146.81,177) and (158,188.19) .. (158,202) .. controls (158,215.81) and (146.81,227) .. (133,227) .. controls (119.19,227) and (108,215.81) .. (108,202) -- cycle ;
\draw  [fill={rgb, 255:red, 208; green, 2; blue, 27 }  ,fill opacity=1 ] (217,229) .. controls (217,215.19) and (228.19,204) .. (242,204) .. controls (255.81,204) and (267,215.19) .. (267,229) .. controls (267,242.81) and (255.81,254) .. (242,254) .. controls (228.19,254) and (217,242.81) .. (217,229) -- cycle ;
\draw  [fill={rgb, 255:red, 155; green, 155; blue, 155 }  ,fill opacity=1 ] (332,204) .. controls (332,190.19) and (343.19,179) .. (357,179) .. controls (370.81,179) and (382,190.19) .. (382,204) .. controls (382,217.81) and (370.81,229) .. (357,229) .. controls (343.19,229) and (332,217.81) .. (332,204) -- cycle ;
\draw    (298.5,60) -- (309,102) ;
\draw    (261.5,58) -- (227.5,97) ;
\draw    (194.5,134) -- (150.5,182) ;
\draw    (215,144) -- (242,204) ;
\draw    (301.5,152) -- (242,204) ;
\draw    (267,229) -- (333.5,214) ;
\draw  [fill={rgb, 255:red, 248; green, 231; blue, 28 }  ,fill opacity=1 ] (143,291) .. controls (143,277.19) and (154.19,266) .. (168,266) .. controls (181.81,266) and (193,277.19) .. (193,291) .. controls (193,304.81) and (181.81,316) .. (168,316) .. controls (154.19,316) and (143,304.81) .. (143,291) -- cycle ;
\draw    (133,227) -- (158.5,267) ;
\draw  [fill={rgb, 255:red, 80; green, 227; blue, 194 }  ,fill opacity=1 ] (294,308) .. controls (294,294.19) and (305.19,283) .. (319,283) .. controls (332.81,283) and (344,294.19) .. (344,308) .. controls (344,321.81) and (332.81,333) .. (319,333) .. controls (305.19,333) and (294,321.81) .. (294,308) -- cycle ;
\draw    (257.5,249) -- (306.5,285) ;
\draw    (224.5,246) -- (173.5,267) ;
\draw    (482,92) -- (441,92) ;
\draw [shift={(439,92)}, rotate = 360] [color={rgb, 255:red, 0; green, 0; blue, 0 }  ][line width=0.75]    (10.93,-3.29) .. controls (6.95,-1.4) and (3.31,-0.3) .. (0,0) .. controls (3.31,0.3) and (6.95,1.4) .. (10.93,3.29)   ;
\draw  [fill={rgb, 255:red, 248; green, 231; blue, 28 }  ,fill opacity=1 ] (499,100) -- (509,100) -- (509,110) -- (499,110) -- cycle ;
\draw    (482,126) -- (441,126) ;
\draw [shift={(439,126)}, rotate = 360] [color={rgb, 255:red, 0; green, 0; blue, 0 }  ][line width=0.75]    (10.93,-3.29) .. controls (6.95,-1.4) and (3.31,-0.3) .. (0,0) .. controls (3.31,0.3) and (6.95,1.4) .. (10.93,3.29)   ;
\draw  [fill={rgb, 255:red, 80; green, 227; blue, 194 }  ,fill opacity=1 ] (499,134) -- (509,134) -- (509,144) -- (499,144) -- cycle ;
\draw    (483,158) -- (442,158) ;
\draw [shift={(440,158)}, rotate = 360] [color={rgb, 255:red, 0; green, 0; blue, 0 }  ][line width=0.75]    (10.93,-3.29) .. controls (6.95,-1.4) and (3.31,-0.3) .. (0,0) .. controls (3.31,0.3) and (6.95,1.4) .. (10.93,3.29)   ;
\draw  [fill={rgb, 255:red, 155; green, 155; blue, 155 }  ,fill opacity=1 ] (500,166) -- (510,166) -- (510,176) -- (500,176) -- cycle ;
\draw    (483,189) -- (442,189) ;
\draw [shift={(440,189)}, rotate = 360] [color={rgb, 255:red, 0; green, 0; blue, 0 }  ][line width=0.75]    (10.93,-3.29) .. controls (6.95,-1.4) and (3.31,-0.3) .. (0,0) .. controls (3.31,0.3) and (6.95,1.4) .. (10.93,3.29)   ;
\draw  [fill={rgb, 255:red, 208; green, 2; blue, 27 }  ,fill opacity=1 ] (500,197) -- (510,197) -- (510,207) -- (500,207) -- cycle ;
\draw    (482,63) -- (441,63) ;
\draw [shift={(439,63)}, rotate = 360] [color={rgb, 255:red, 0; green, 0; blue, 0 }  ][line width=0.75]    (10.93,-3.29) .. controls (6.95,-1.4) and (3.31,-0.3) .. (0,0) .. controls (3.31,0.3) and (6.95,1.4) .. (10.93,3.29)   ;
\draw    (481,34) -- (440,34) ;
\draw [shift={(438,34)}, rotate = 360] [color={rgb, 255:red, 0; green, 0; blue, 0 }  ][line width=0.75]    (10.93,-3.29) .. controls (6.95,-1.4) and (3.31,-0.3) .. (0,0) .. controls (3.31,0.3) and (6.95,1.4) .. (10.93,3.29)   ;
\draw  [fill={rgb, 255:red, 0; green, 0; blue, 0 }  ,fill opacity=1 ] (498,42) -- (508,42) -- (508,52) -- (498,52) -- cycle ;
\draw    (481,218) -- (440,218) ;
\draw [shift={(438,218)}, rotate = 360] [color={rgb, 255:red, 0; green, 0; blue, 0 }  ][line width=0.75]    (10.93,-3.29) .. controls (6.95,-1.4) and (3.31,-0.3) .. (0,0) .. controls (3.31,0.3) and (6.95,1.4) .. (10.93,3.29)   ;
\draw    (481,247) -- (440,247) ;
\draw [shift={(438,247)}, rotate = 360] [color={rgb, 255:red, 0; green, 0; blue, 0 }  ][line width=0.75]    (10.93,-3.29) .. controls (6.95,-1.4) and (3.31,-0.3) .. (0,0) .. controls (3.31,0.3) and (6.95,1.4) .. (10.93,3.29)   ;

\draw (171,78) node [anchor=north west][inner sep=0.75pt]   [align=left] {$\displaystyle v_{1}$$ $};
\draw (81,173) node [anchor=north west][inner sep=0.75pt]   [align=left] {$\displaystyle v_{2}$$ $};
\draw (120,300) node [anchor=north west][inner sep=0.75pt]   [align=left] {$\displaystyle v_{3}$$ $};
\draw (346,322) node [anchor=north west][inner sep=0.75pt]   [align=left] {$\displaystyle v_{4}$$ $};
\draw (384,213) node [anchor=north west][inner sep=0.75pt]   [align=left] {$\displaystyle v_{5}$$ $};
\draw (341,95) node [anchor=north west][inner sep=0.75pt]   [align=left] {$\displaystyle v_{6}$$ $};
\draw (306,16) node [anchor=north west][inner sep=0.75pt]   [align=left] {$\displaystyle v_{7}$$ $};
\draw (196,195) node [anchor=north west][inner sep=0.75pt]   [align=left] {$\displaystyle v_{8}$$ $};
\draw (414,81) node [anchor=north west][inner sep=0.75pt]   [align=left] {$\displaystyle a_{3}$};
\draw (490,84) node [anchor=north west][inner sep=0.75pt]   [align=left] {$\displaystyle b$};
\draw (414,115) node [anchor=north west][inner sep=0.75pt]   [align=left] {$\displaystyle a_{4}$};
\draw (490,118) node [anchor=north west][inner sep=0.75pt]   [align=left] {$\displaystyle b$};
\draw (415,147) node [anchor=north west][inner sep=0.75pt]   [align=left] {$\displaystyle a_{5}$};
\draw (491,150) node [anchor=north west][inner sep=0.75pt]   [align=left] {$\displaystyle b$};
\draw (412,180) node [anchor=north west][inner sep=0.75pt]   [align=left] {$\displaystyle a_{6}$};
\draw (491,181) node [anchor=north west][inner sep=0.75pt]   [align=left] {$\displaystyle b$};
\draw (414,52) node [anchor=north west][inner sep=0.75pt]   [align=left] {$\displaystyle a_{2}$};
\draw (490,55) node [anchor=north west][inner sep=0.75pt]   [align=left] {$\displaystyle b_{\emptyset }$};
\draw (413,23) node [anchor=north west][inner sep=0.75pt]   [align=left] {$\displaystyle a_{1}$};
\draw (489,26) node [anchor=north west][inner sep=0.75pt]   [align=left] {$\displaystyle b$};
\draw (413,207) node [anchor=north west][inner sep=0.75pt]   [align=left] {$\displaystyle a_{7}$};
\draw (489,210) node [anchor=north west][inner sep=0.75pt]   [align=left] {$\displaystyle b_{\emptyset }$};
\draw (413,236) node [anchor=north west][inner sep=0.75pt]   [align=left] {$\displaystyle a_{8}$};
\draw (489,239) node [anchor=north west][inner sep=0.75pt]   [align=left] {$\displaystyle b_{\emptyset }$};

\end{tikzpicture}
}
\caption{In the presented graph colored using 5 colors, $\lbrace v_{1},v_{3},v_{4},v_{5},v_{6} \rbrace$ is a multicolored independent set of size 5. The contsructed instance contains $8$ layers, one for each vertex; and the presented assignment is 5-globally optimal since it is pareto optimal in layers $1,3,4,5,$ and 6.}
\label{fig:W1ItemsPlusAlpha}
\end{figure}

We now prove that $G$ admits a multicolored independent set of size $k$ if and only if the constructed instance admits a $k$-globally optimal assignment.

($\Rightarrow$) Assume that $G$ admits a multicolored independent set $V' = \lbrace v_{i_{1}},\ldots,v_{i_{k}} \rbrace$. We define an assignment $p$ by $p(a_{i_{j}}) = b_{c(v_{j})}$ for each $j \in [k]$ (see Figure \ref{fig:W1ItemsPlusAlpha}), and $p(a_{i}) = b_{\emptyset}$ for each $i \notin \lbrace i_{1},\ldots,i_{k} \rbrace$. Note that the assignment does not assign the same item to two or more agents since the colors of the vertices in $V'$ are distinct. We claim that $p$ is $k$-globally optimal for the constructed instance. First, the agents $a_{i_{1}},\ldots,a_{i_{k}}$ contain non-empty preference lists in the preference profiles $P_{i_{1}},\ldots,P_{i_{k}}$ since their corresponding vertices are pair-wise non-adjacent and colored with distinct colors. Hence, each item is acceptable by its owner in each $P_{i_{j}}$. Second, $P_{i_{1}},\ldots,P_{i_{k}}$ cannot include trading cycles since each preference list is of length at most one.

($\Leftarrow$) Suppose there exists a $k$-globally optimal assignment $p$ for the constructed instance. Then, there exist $k$ layers $i_{1},\ldots,i_{k}$ such that $p$ is pareto optimal in each $P_{i_{j}}$. \Prcref{claim:reduction-multicolored-claim} implies that for each $j \in [k]$, $p(a_{i_{j}})=b_{c(v_{i_{j}})}$, and for each $q \in [n]$ such that $\lbrace v_{q},v_{j} \rbrace \in E$, $p(a_{q}) = b_{\emptyset}$. We claim that $V' = \lbrace v_{i_{1}},\ldots,v_{i_{k}} \rbrace$ is a multicolored independent set. First, we show that it forms an independent set. Towards a contradiction, suppose there exist $j_{1},j_{2} \in [k]$ such that $\lbrace v_{i_{j_{1}}} , v_{i_{j_{2}}} \rbrace \in E$. Then, the preference list of $a_{i_{j_{2}}}$ in $P_{i_{j_{1}}}$ is empty, and therefore $p$ is not legal in $P_{i_{j_{1}}}$, a contradiction to its pareto optimality in $P_{i_{j_{1}}}$. Second, we show that it is multicolored. Towards a contradiction, suppose there exist $j_{1},j_{2} \in [k]$ such that $c(v_{i_{j_{1}}}) = c(v_{i_{j_{2}}})$. By \Clcref{claim:reduction-multicolored-claim}, since $p$ is pareto optimal in layers $i_{j_{1}}$ and $i_{j_{2}}$, it satisfies $p(a_{i_{j_{1}}}) = b_{c({v_{i_{j_{1}}}})}$ and $p(a_{i_{j_{2}}}) = b_{c(v_{i_{j_{1}}})}$. We have that $p$ assigns $b_{c(v_{i_{j_{1}}})}$ to two agents, a contradiction.\qed
\end{proof}

We now provide a similar hardness result for the parameter $m + (\ell - \alpha)$.

\begin{theorem}\label{theorem:w-one-hardness-items-l-alpha}
\goassign\ is \WOH\ for the parameter $k=m+(\ell - \alpha)$.
\end{theorem}
\begin{proof}
We provide a different parameterized (and polynomial time) reduction from \MCIS\ to \goassign.

Given an instance $(G=(V,E),c)$ of \MCIS, assume that $V= \lbrace v_{1},\ldots,v_{n} \rbrace$. We construct an instance of \goassign\ with agent set $A = \lbrace a_{1},\ldots,a_{n} \rbrace$ and item set $I = \lbrace b_{1},\ldots,b_{k} \rbrace$, consisting of $\ell = n+1$ layers. Informally speaking, each agent $a_{i}$ corresponds to a vertex $v_{i}$, and each item $b_{i}$ corresponds to a color $i$. We construct an instance such that the agents that allocate items from $I$ in an $\ell$-globally optimal assignment correspond to vertices which form a multicolored independent set in $G$. Moreover, we require each agent to allocate either the item that corresponds to its color with respect to $c$, or $b_{\emptyset}$. The first layer is defined as follows:
 
\begin{framed}[.4\textwidth]
\begin{itemize}
  \item $a_{i}\ $: $\ b_{c(i)}\ \forall i \in [k]$
\end{itemize}
\end{framed}

We construct $n$ additional layers. For each $i \in [n]$, layer $1+i$ is defined as follows:

\begin{framed}
\begin{itemize}
  \item $a_{i}\ $:$\ \lbrace b_{j} \mid j \in [k] \rbrace \setminus \lbrace b_{c(v_{i})} \rbrace$ (ordered arbitrarily)$\ >b_{c(v_{i})}$
   \item $a_{j}\ $:$\ b_{c(v_{j})}\ \ \forall j \in [n] \setminus \lbrace i \rbrace$ such that (1) $\lbrace v_{i},v_{j} \rbrace \in E$ and $c(v_{j}) = c(v_{i})$ or (2) $\lbrace v_{i},v_{j} \rbrace \notin E$
   \item $a_{j}\ $:$\ b_{c(v_{i})} > b_{c(v_{j})}\ \ \forall j \in [n] \setminus \lbrace i \rbrace$ such that $\lbrace v_{i},v_{j} \rbrace \in E$ and $c(v_{j}) \neq c(v_{i})$
\end{itemize}
\end{framed}
We finally set $\alpha = \ell$. We claim that $G$ admits a multicolored independent set of size $k$ if and only if there exists an $\ell$-globally optimal assignment for the constructed instance.

($\Rightarrow$): Suppose $G$ admits a multicolored independent set of size $k$, $V' = \lbrace v_{i_{1}},\ldots,v_{i_{k}} \rbrace$. We define an assignment $p$ by $p(a_{i_{j}}) = b_{c(v_{i_{j}})}$ for each $j \in [k]$, and $p(a_{i})=b_{\emptyset}$ for each $i \notin \lbrace i_{1},\ldots,i_{k} \rbrace$. We claim that $p$ is pareto optimal in all the layers. First, observe that for each $i \in [k]$, $p(a_{i})$ is acceptable by $a_{i}$ in each layer, and each layer cannot admit self loops since all the items are allocated. Second, $p$ is pareto optimal in the first layer: since each agent accepts only a single item, no subset of agents can admit a trading cycle. Third, we claim that $p$ is pareto-optimal in layer $1+i$ for each $i \in [n]$. Towards a contradiction, suppose that there exists $i \in [n]$ such that $p$ is not pareto optimal in layer $1+i$. By \Prcref{prop:po-iff-no-tc-and-sl} there exists a trading cycle in layer $1+i$. Observe that a trading cycle in layer $1+i$ cannot contain agents with a single item on their preference list. Moreover, it cannot contain only agents from $\lbrace a_{j} \mid \lbrace v_{i},v_{j} \rbrace \in E \text{ and } c(v_{j}) \neq c(v_{i}) \rbrace$ since each such agent ranks $b_{c(v_{i})}$ first in its preference list. Therefore, we have that the only possible trading cycle in layer $1+i$ consists of the agent $a_{i}$ and an agent $a_{r}$ such that $p(a_{i}) = b_{c(v_{i})}$, $c(v_{i}) \neq c(v_{r})$, $\lbrace v_{i},v_{r} \rbrace \in E$, and $p(a_{r}) = b_{c(v_{r})}$. Then, by the construction of $p$, we have that $v_{i},v_{r} \in V'$, a contradiction.

($\Leftarrow$): Suppose there exists an $\alpha$-globally optimal assignment $p$ for the constructed instance. We first show that $p$ must allocate all the items. Towards a contradiction, assume that there exists $j \in [k]$ such that $b_{j}$ is not assigned to any agent. Since $n \geq k$, there exists an agent $a_{t}$ that is not assigned to an item. Notice that $a_{t}$ accepts all the items in layer $1+t$, thus, $a_{t}$ and $b_{j}$ admit a self loop, a contradiction. Let $a_{i_{1}},\ldots,a_{i_{k}}$ denote the agents that are allocated an item from $I$. We claim that $p(a_{i_{j}}) = b_{c(v_{i_{j}})}$ for each $j \in [k]$. For the sake of contradiction, assume that there exists $j \in [k]$ such that $p(a_{i_{j}}) \neq b_{c(v_{i_{j}})}$. We consider two cases: (1) $p(a_{i_{j}})=b_{\emptyset}$. In this case, $p$ admits a self loop among $a_{i_{j}}$ and $b_{c(v_{i_{j}})}$ in the first layer, a contradiction. (2) $p(a_{i_{j}})=b_{r} \neq b_{c(v_{i_{j}})}$. In this case, $p$ is not legal in the first layer, a contradiction. This implies that the vertices in the set $V' = \lbrace v_{i_{1}},\ldots,v_{i_{k}} \rbrace$ are colored with all the colors. We claim that $V'$ is a also an independent set. Towards a contradiction, assume there exist $j,r \in [k]$ such that $\lbrace v_{i_{j}},v_{i_{r}} \rbrace \in E$. Then we have that $a_{i_{j}}$ and $a_{i_{r}}$ form the trading cycle $(a_{i_{j}},b_{c(v_{i_{j}})},a_{i_{r}},b_{c(v_{i_{r}})})$ in layer $1+i_{r}$, a contradiction to \Prcref{prop:po-iff-no-tc-and-sl}.\qed
\end{proof}
\section{\NP-Hardness}
\label{sec:NPHardness}
In this section, we prove that for any $2 \leq \alpha \leq \ell$, \goassign\ is \NPH. After that, we provide three constructions of preference profiles given an instance of \ThreeSat. Then, we will rely on these constructions to prove that \goassign\ is \paraH\ with respect to the parameter $\ell + d$.

\begin{theorem} \label{theorem:np-hardness-for-larger-ell} 
For any $2 \leq \alpha \leq \ell$, \goassign\ is \NPH.
\end{theorem}

\begin{proof}
We rely on a reduction made by Aziz et al.~\cite{ijcai2017-12} from the \serDictFeas\ problem, which was proved to be \NPH\ by Saban and Sethuraman \cite{RePEc:inm:ormoor:v:40:y:2015:i:4:p:1005-1014}. In the \serDictFeas\ problem, the input is a tuple $(A,I,P,a,b)$ where $A$ is a set of $n$ agents, $I$ is a set of $n$ items, $P$ is the preference profile in which each agent has a complete linear order on the items, $a \in A$, and $b \in I$. The task is to decide whether there exists a permutation for which serial dictatorship (defined in Section \ref{sec:preliminaries}) allocates item $b$ to agent $a$.  Given such $(A,I,P,a,b)$, Aziz et al.~\cite{ijcai2017-12} constructed two preference profiles, $P_{1}$ and $P_{2}$, such that $(A,I,P,a,b)$ is a \yesinstance\ if and only if there exists an assignment that is pareto optimal in both $P_{1}$ and $P_{2}$.

Let $(A,I,P,a,b)$ be an instance of \serDictFeas. We construct the aforementioned preference profiles $P_{1}$ and $P_{2}$. We also add $\ell-\alpha$ additional new items $c_{1},\ldots,c_{\ell-\alpha}$, and define $I'=I \cup \lbrace c_{1},\ldots,c_{\ell-\alpha} \rbrace$. We construct $\ell - \alpha$ preference profiles $P_{1}',\ldots,P_{\ell - \alpha}'$, where for each $i \in [\ell - \alpha]$, $P_{i}'$ is defined as follows:
\begin{framed}[.38\textwidth]
\begin{itemize}
  \item $a\ $: $\ c_{i}$
  \item $a'\ :\ \emptyset\ \ \forall a' \in A \setminus \lbrace a \rbrace$
\end{itemize}
\end{framed}
In other words, the only item that agent $a$ accepts in $P_{i}'$ is $c_{i}$, and for each $a' \in A \setminus \lbrace a \rbrace$, $a'$ accepts no items. 
We construct an instance of \goassign\ with the agent set $A$, and the item set $I'$, consisting of $\ell$ layers. The first two layers contain the preference profiles $P_{1}$ and $P_{2}$, the next $\alpha - 2$ layers contain copies of $P_{1}$, and the next $\ell-\alpha$ layers contain $P_{1}',\ldots,P_{\ell - \alpha}'$. Let us first prove the following.
\begin{myclaim} \label{claim:np-hardness-claim}
Let $p$ be an $\alpha$-globally optimal assignment for the constructed instance. Then $p$ is pareto optimal in $P_{1}$ and $P_{2}$.
\end{myclaim}
\begin{proof}
Note that the only pareto optimal assignment for $P_{i}'$ is the assignment that allocates $c_{i}$ to $a$, and $b_{\emptyset}$ to each $a' \in A \setminus \lbrace a \rbrace$. Hence, there does not exist an assignment that is pareto optimal in both $P_{i}'$ and $P_{j}'$ when $i \neq j$. Moreover, there does not exist an assignment that is pareto optimal in both $P_{i}'$ and $P_{1}$ or in both $P_{i}'$ and $P_{2}$ since $c_{i}$ is not acceptable by $a$ in $P_{1}$ as well as $P_{2}$. The only possible option is that $p$ is pareto optimal in the first $\alpha$ layers, then we have that $p$ is pareto optimal in both $P_{1}$ and $P_{2}$.
\end{proof}

We claim that there exists an $\alpha$-globally optimal assignment for the constructed instance if and only if $(A,I,P,a,b)$ is a \yesinstance.

($\Rightarrow$): Suppose there exists an $\alpha$-globally optimal assignment $p$ for the constructed instance. By \Clcref{claim:np-hardness-claim}, $p$ is pareto optimal in both $P_{1}$ and $P_{2}$. By Aziz et al.~\cite{ijcai2017-12}, this implies that $(A,I,P,a,b)$ is a \yesinstance.

($\Leftarrow$): Assume that $(A,I,P,a,b)$ is a \yesinstance. By Aziz et al.~\cite{ijcai2017-12}, there exists an assignment $p$ that is pareto optimal in both $P_{1}$ and $P_{2}$. By the construction of the instance, $p$ is pareto optimal in the first $\alpha$ layers, hence $p$ is $\alpha$-globally optimal for the constructed instance.\qed
\end{proof}

\begin{corollary} \label{corollary:np-hardness-for-ell-2} 
\goassign\ is \paraH\ for the parameter $\ell + \alpha$.
\end{corollary}

In contrast to \Thcref{theorem:np-hardness-for-larger-ell}, we mention that it is impossible to obtain the same hardness result when $\alpha = 1$. 

\begin{observation} \label{observation:alpha-one-is-poly} 
\goassign\ is solvable in a polynomial time when $\alpha=1$.
\end{observation}
\begin{proof}
We can simply find a pareto optimal assignment for the first layer using \Cocref{corollary:number-of-po-assign}.\qed
\end{proof}

We define three constructions of preference profiles given an instance of \ThreeSat, and we consider their connections to the satisfiability of the formula. We will rely on these connections to design a polynomial reduction from \ThreeSat\ to \goassign\ that shows that \goassign\ is \paraH\ with respect to the parameter $\ell + \alpha +d$ (where $d$ is the maximal length of a preference list). We will also rely on these results in the design of our cross-compositions, which prove that \goassign\ does not admit a polynomial kernel with respect to the parameters $n + m+\alpha$, $n + m+(\ell - \alpha)$, and $m+\ell$ unless \NP$\subseteq $\coNPpoly. 

Let $n,m \in \mathbb{N}$ be positive integers. We use the notation $A(m,n)$ to refer to the agent set $A(m,n)= \lbrace a_{i,j} \mid i\in [m], j \in [n] \rbrace \cup \lbrace \overline{a_{i,j}} \mid i\in [m], j \in [n] \rbrace$, and we use the notation $I(m,n)$ to refer to the the item set $I(m,n)= \lbrace b_{i,j} \mid i\in [m], j \in [n] \rbrace \cup \lbrace \overline{b_{i,j}} \mid i\in [m], j \in [n] \rbrace$. We will define two preference profiles over $A(m,n)$ and $I(m,n)$: $P_{1}(m,n)$ and $P_{2}(m,n)$. Intuitively, given a \ThreeSat\ instance with $n$ variables and $m$ clauses, we will construct the sets $A(m,n)$ and $I(m,n)$, which contain two agents and two items for each clause-variable pair. The way that these agents and items are assigned to each other in an assignment that is pareto optimal in both $P_{1}(m,n)$ and $P_{2}(m,n)$ will encode a boolean assignment for the variable set of the instance. We define the preference profile $P_{1}(m,n)$ as follows.

\begin{framed}[.5\textwidth]
\begin{itemize}
\item $a_{i,j}$ : $b_{i,j}>\overline{b_{i,j}}\ \ \forall i \in [m], j \in [n]$
\item $\overline{a_{i,j}}$ : $b_{i,j}>\overline{b_{i,j}}\ \ \forall i \in [m], j \in [n]$
\end{itemize}
\end{framed}

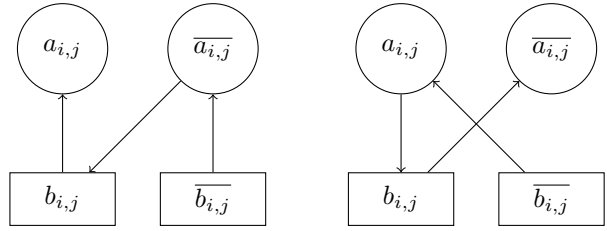
\begin{figure}[t]
\center
\begin{tikzpicture}
	\begin{pgfonlayer}{nodelayer}
		\node [style=vertex] (0) at (-11, 4) {$a_{i,j}$};
		\node [style=vertex] (1) at (-7, 4) {$\overline{a_{i,j}}$};
		\node [style=rec] (2) at (-11, 0) {$b_{i,j}$};
		\node [style=rec] (3) at (-7, 0) {$\overline{b_{i,j}}$};
		\node [style=vertex] (4) at (-2, 4) {$a_{i,j}$};
		\node [style=vertex] (5) at (2, 4) {$\overline{a_{i,j}}$};
		\node [style=rec] (6) at (-2, 0) {$b_{i,j}$};
		\node [style=rec] (7) at (2, 0) {$\overline{b_{i,j}}$};
	\end{pgfonlayer}
	\begin{pgfonlayer}{edgelayer}
		\draw [style=edge] (2) to (0);
		\draw [style=edge] (3) to (1);
		\draw [style=edge] (1) to (2);
		\draw [style=edge] (7) to (4);
		\draw [style=edge] (4) to (6);
		\draw [style=edge] (6) to (5);
	\end{pgfonlayer}
\end{tikzpicture}
\caption{If $p$ is pareto optimal in $P_{1}(m,n)$, then for each $i \in [m],j \in [n]$, the trading graph of $P_{1}(m,n)$ with respect to $p$ contains one of these two sub-graphs.}
\label{fig:tgForTheFirstClaim3Sat}
\end{figure}

\begin{myclaim}\label{claim:threesatclaim1}
An assignment $p : A(m,n) \rightarrow I(m,n) \cup \lbrace b_{\emptyset} \rbrace$ is pareto optimal in $P_{1}(m,n)$ if and only if $\lbrace p(a_{i,j}),p(\overline{a_{i,j}}) \rbrace = \lbrace b_{i,j},\overline{b_{i,j}} \rbrace$ for each $i \in [m]$ and $j \in [n]$.
\end{myclaim}
\begin{proof}
($\Rightarrow$): Assume that $p$ is pareto optimal in $P_{1}(m,n)$. Observe that $b_{i,j}$ and $\overline{b_{i,j}}$ are only acceptable by $a_{i,j}$ and $\overline{a_{i,j}}$. If $|\lbrace p(a_{i,j}),p(\overline{a_{i,j}}) \rbrace \cap \lbrace b_{i,j},\overline{b_{i,j}} \rbrace |< 2$, then at least one of the items $b_{i,j},\overline{b_{i,j}}$ does not have an owner; this implies that at least one of the agents $a_{i,j}$, $\overline{a_{i,j}}$ admits a self loop, a contradiction to \Prcref{prop:po-iff-no-tc-and-sl} (see Figure \ref{fig:tgForTheFirstClaim3Sat}).

($\Leftarrow$): Assume that $\lbrace p(a_{i,j}),p(\overline{a_{i,j}}) \rbrace = \lbrace b_{i,j},\overline{b_{i,j}} \rbrace$ for each $i \in [m]$ and $j \in [n]$. Observe that the agents in $A(m,n)$ pairwise respect each other in $P_{1}(m,n)$; then, by \Lecref{lemma:respectedAgents}, $p$ does not admit trading cycles in $P_{1}(m,n)$. Since $p$ allocates all the items, it also does not admit self loops. Thus, by \Prcref{prop:po-iff-no-tc-and-sl}, $p$ is pareto optimal in $P_{1}(m,n)$.\qed
\end{proof}

We now define the second preference profile $P_{2}(m,n)$ over $A(m,n)$ and $I(m,n)$. Intuitively, being pareto optimal in both $P_{1}(m,n)$ and $P_{2}(m,n)$ requires the agents to ``agree'' with each other. That is, $a_{i_{1},j}$ gets $b_{i_{1},j}$ if and only if $a_{i_{2},j}$ gets $b_{i_{2},j}$ for each $i_{1} \neq i_{2}$. For $j \in [n]$, we denote the sets $P_{j}^{\mathrm{true}}$, and $P_{j}^{\mathrm{false}}$ by $P_{j}^{\mathrm{true}}=\lbrace (a_{i,j},b_{i,j}) , (\overline{a_{i,j}},\overline{b_{i,j}}) \mid i \in [m] \rbrace$, and $P_{j}^{\mathrm{false}}=\lbrace (a_{i,j},\overline{b_{i,j}}) , (\overline{a_{i,j}},b_{i,j}) \mid i \in [m] \rbrace$. Informally speaking, $P_{j}^{\mathrm{true}}$ and $P_{j}^{\mathrm{false}}$ will correspond to setting the variable $x_{j}$ to true or false, respectively. 

\begin{framed}[.9\textwidth]
\begin{itemize}
\item $a_{m,j}$ : $b_{m,j}>b_{m-1,j}>\overline{b_{m,j}}\ \ \forall j \in [n]$
\item $\overline{a_{m,j}}\ $: $\ b_{m,j}>b_{m-1,j}>\overline{b_{m,j}}\ \ \forall j \in [n]$
\item $a_{i,j}$ : $b_{i-1,j}>\overline{b_{i,j}}>\overline{b_{i+1,j}}>b_{i,j}\ \ \forall i \in \lbrace 2,\ldots,m-1 \rbrace$,$ j \in [n]$
\item $\overline{a_{i,j}}$ : $b_{i-1,j}>\overline{b_{i,j}}>\overline{b_{i+1,j}}>b_{i,j}\ \ \forall i \in \lbrace 2,\ldots,m-1 \rbrace$,$ j \in [n]$
\item $a_{1,j}$ : $\overline{b_{1,j}}>\overline{b_{2,j}}>b_{1,j}\ \ \forall j \in [n]$
\item $\overline{a_{1,j}}$ : $\overline{b_{1,j}}>\overline{b_{2,j}}>b_{1,j}\ \ \forall j \in [n]$
\end{itemize}
\end{framed}

\begin{myclaim}\label{claim:threesatclaim2}
An assignment $p : A(m,n) \rightarrow I(m,n) \cup \lbrace b_{\emptyset} \rbrace$ is pareto optimal in both $P_{1}(m,n)$ and $P_{2}(m,n)$ if and only if for each $j \in [n]$, either $P_{j}^{\mathrm{true}} \subseteq p$ or $P_{j}^{\mathrm{false}} \subseteq p$.
\end{myclaim} 
\begin{proof}
($\Rightarrow$): Assume that $p$ is pareto optimal in both $P_{1}(m,n)$ and $P_{2}(m,n)$. Towards a contradiction, suppose that there exists $j \in [n]$ satisfying that both $P_{j}^{\mathrm{true}} \nsubseteq p$ and $P_{j}^{\mathrm{false}} \nsubseteq p$. By \Clcref{claim:threesatclaim1}, there exist $i_{1},i_{2} \in [m]$ such that $i_{1}<i_{2}$, satisfying that either (1) $p(a_{i_{1},j})=b_{i_{1},j}$ and $p(a_{i_{2},j})=\overline{b_{i_{2},j}}$, or (2) $p(a_{i_{1},j})=\overline{b_{i_{1},j}}$ and $p(a_{i_{2},j})=b_{i_{2},j}$. For the first case, note that there must exist $i \in [m]$ such that $i_{1}\leq i < i_{2}$, $p(a_{i,j})=b_{i,j}$, and $p(a_{i+1,j})=\overline{b_{i+1,j}}$. Then we have that $p$ admits the trading cycle $(a_{i,j},b_{i,j},a_{i+1,j},\overline{b_{i+1,j}})$ in $P_{2}(m,n)$, a contradiction. For the second case, we have that there exists $i \in [m]$ such that $i_{1}\leq i < i_{2}$, $p(a_{i,j})=\overline{b_{i,j}}$, and $p(a_{i+1,j})=b_{i+1,j}$. By \Clcref{claim:threesatclaim1}, $p(\overline{a_{i,j}})={b_{i,j}}$ and $p(\overline{a_{i+1,j}})=\overline{b_{i+1,j}}$, then $p$ admits the trading cycle $(\overline{a_{i,j}},b_{i,j},\overline{a_{i+1,j}},\overline{b_{i+1,j}})$ in $P_{2}(m,n)$, a contradiction.

\begin{figure}[t]
\center
\begin{tikzpicture}
	\begin{pgfonlayer}{nodelayer}
		\node [style=vertex] (0) at (-14, 6) {$\overline{a_{i-1,j}}$};
		\node [style=vertex] (1) at (-10, 6) {$a_{i-1,j}$};
		\node [style=vertex] (2) at (-6, 6) {$a_{i,j}$};
		\node [style=vertex] (3) at (-2, 6) {$\overline{a_{i,j}}$};
		\node [style=rec] (7) at (-14, 2) {$\overline{b_{i-1,j}}$};
		\node [style=rec] (8) at (-10, 2) {$b_{i-1,j}$};
		\node [style=rec] (9) at (-6, 2) {$\overline{b_{i,j}}$};
		\node [style=rec] (10) at (-2, 2) {$b_{i,j}$};
	\end{pgfonlayer}
	\begin{pgfonlayer}{edgelayer}
		\draw [style=edge] (10) to (3);
		\draw [style=edge] (9) to (2);
		\draw [style=edge] (8) to (1);
		\draw [style=edge] (1) to (9);
		\draw [style=edge] (2) to (8);
		\draw [style=edge] (7) to (0);
	\end{pgfonlayer}
\end{tikzpicture}
\caption{Example of the trading graph for the first case.}
\label{fig:tgForTheFirstCase}
\end{figure}
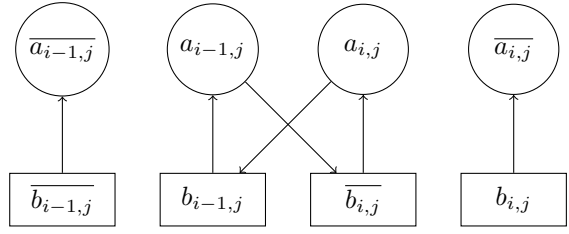

($\Leftarrow$): Assume that for each $j \in [n]$, either $P_{j}^{\mathrm{true}} \subseteq p$ or $P_{j}^{\mathrm{false}} \subseteq p$. By \Clcref{claim:threesatclaim1}, $p$ is pareto optimal in $P_{1}(m,n)$. Then by the construction of $P_{2}(m,n)$, observe that every possible trading cycle in $P_{2}(m,n)$ has one of the forms: (1) $(a_{i,j},\overline{b_{i,j}},a_{i-1,j},b_{i-1,j})$ (Shown in Figure \ref{fig:tgForTheFirstCase}) or (2) $(\overline{a_{i,j}},\overline{b_{i,j}},\overline{a_{i-1,j}},b_{i-1,j})$ where $i \in \lbrace 2,\ldots,m \rbrace$. This implies that there exist $j \in [n]$, and $i \in \lbrace 2,\ldots,m \rbrace$ such that either (1) $p(a_{i,j})=b_{i,j}$ and $p(a_{i-1,j})=\overline{b_{i,j}}$ or (2) $p(a_{i,j})=\overline{b_{i,j}}$ and $p(a_{i-1,j})=b_{i,j}$. Thus, both $P_{j}^{\mathrm{true}} \subsetneq p$ and $P_{j}^{\mathrm{false}} \subsetneq p$, a contradiction.\qed
\end{proof}

Let $D = (\XX,\CC)$ be an instance of \ThreeSat\ where $\XX = \lbrace x_{1},\ldots,x_{n} \rbrace$ is the set of variables, and $\CC = \lbrace C_{1},\ldots,C_{m} \rbrace$ is the set of clauses, each of size 3.
In order to construct the third preference profile $P_{3}(D)$, order the literals in each clause arbitrarily, such that $C_{i}=\ell_{i}^{1} \vee \ell_{i}^{2} \vee \ell_{i}^{3}$ for each $i \in [m]$. The third preference profile is responsible for the satisfiability of the formula. Let us define the following:

\begin{definition}
Let $C_{i}=\ell_{i}^{1} \vee \ell_{i}^{2} \vee \ell_{i}^{3} \in \CC$. We define $\ind_{D}(i,j)$ as the index of the variable which appears in the $j$-th literal in $C_{i}$ for each $j \in [3]$.
\end{definition}

For example, if $C_{i}=x_{3} \vee \overline{x_{1}} \vee x_{5}$, then $\ind_{D}(i,1)=3$, $\ind_{D}(i,2)=1$, and $\ind_{D}(i,3)=5$. Briefly put, each assignment $p$ that is pareto optimal in the profiles $P_{1}(m,n)$,$P_{2}(m,n)$, and $P_{3}(D)$ contains either $P_{j}^{\mathrm{true}} = \lbrace (a_{i,j},b_{i,j}),(\overline{a_{i,j}},\overline{b_{i,j}}) | i \in [m] \rbrace$ or $P_{j}^{\mathrm{false}} = \lbrace (a_{i,j},\overline{b_{i,j}}),(\overline{a_{i,j}},b_{i,j})| i \in [m] \rbrace$ for every $j \in [n]$, and for each clause $C_{i}=\ell_{i}^{1} \vee \ell_{i}^{2} \vee \ell_{i}^{3} \in \CC$, there exists at least one satisfied literal $\ell_{i}^{j}$. Being pareto optimal in $P_{3}(D)$ enforces $a_{i,\ind_{i}(1)},a_{i,\ind_{i}(2)}$ and $a_{i,\ind_{i}(3)}$ to admit a trading cycle if none of their corresponding literals is satisfied. We define the following:

\begin{definition}
For each $i \in [m]$, $j \in [3]$, we define
$
  \bb_{D}(i, j) =
  \begin{cases}
      \overline{b_{i,\ind_{D}(i,j)}} & \ell_{i}^{j} \text{ is negative}\\
    b_{i,\ind_{D}(i,j)} & \ell_{i}^{j} \text{ is positive}
  \end{cases} 
$.
\end{definition}
Intuitively, when $a_{i,\ind_{D}(i,j)}$ gets $\bb_{D}(i, j)$ and $\overline{a_{i,\ind_{D}(i,j)}}$ gets $\overline{\bb_{D}(i, j)}$, it means that $\ell_{i}^{j}$ is ``satisfied''.
Preference profile $P_{3}(D)$ is defined as follows:
\begin{framed}
\begin{itemize}
\item $a_{i,\ind_{D}(i,3)}\ $: $\ \bb_{D}(i,3)>\overline{\bb_{D}(i,2)}>\overline{\bb_{D}(i,3)}\ \ \forall i \in [m]$
\item $a_{i,\ind_{D}(i,2)}\ $: $\ \bb_{D}(i,2)>\overline{\bb_{D}(i,1)}>\overline{\bb_{D}(i,2)}\ \ \forall i \in [m]$
\item $a_{i,\ind_{D}(i,1)}\ $: $\ \bb_{D}(i,1)>\overline{\bb_{D}(i,3)}>\overline{\bb_{D}(i,1)}\ \ \forall i \in [m]$
\item $\overline{a_{i,\ind_{D}(i,r)}}\ $: $\ \bb_{D}(i,r)>\overline{\bb_{D}(i,r)}\ \ \forall i \in [m] , r \in [3]$
\item $a_{i,j}\ $: $\ b_{i,j}>\overline{b_{i,j}}\ \ \forall i \in [m], j \in [n]$ such that $x_{j}$ does not appear in $C_{i}$
\item $\overline{a_{i,j}}\ $: $\ b_{i,j}>\overline{b_{i,j}}\ \ \forall i \in [m], j \in [n]$ such that $x_{j}$ does not appear in $C_{i}$

\end{itemize}
\end{framed} 

\begin{myclaim}\label{claim:threesatclaim3}
An assignment $p : A(m,n) \rightarrow I(m,n) \cup \lbrace b_{\emptyset} \rbrace$ is pareto optimal in $P_{1}(m,n), P_{2}(m,n)$, and $P_{3}(D)$ if and only if:
\begin{itemize}
\item For each $j \in [n]$, either $P_{j}^{\mathrm{true}} \subseteq p$ or $P_{j}^{\mathrm{false}} \subseteq p$.
\item For each clause $C_{i}=\ell_{i}^{1} \vee \ell_{i}^{2} \vee \ell_{i}^{3}  \in \CC$, there exists at least one $j \in [3]$ such that $p(a_{i,\ind_{D}(i,j)}) = \bb_{D}(i,j)$.  
\end{itemize}
\end{myclaim} 

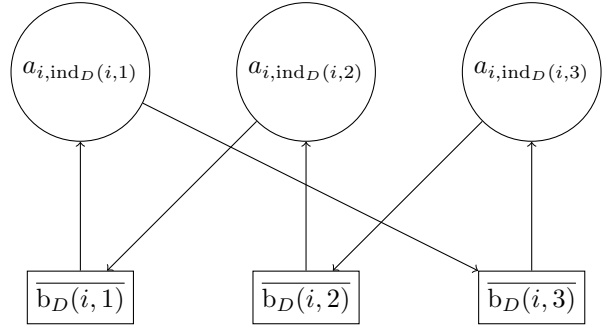
\begin{figure}[t]
\center
\begin{tikzpicture}
	\begin{pgfonlayer}{nodelayer}
		\node [style=vertex] (0) at (3, 6) {$a_{i,\ind_{D}(i,3)}$};
		\node [style=vertex] (1) at (-3, 6) {$a_{i,\ind_{D}(i,2)}$};
		\node [style=vertex] (2) at (-9, 6) {$a_{i,\ind_{D}(i,1)}$};
		\node [style=rec] (3) at (3, 0) {$\overline{\bb_{D}(i,3)}$};
		\node [style=rec] (4) at (-3, 0) {$\overline{\bb_{D}(i,2)}$};
		\node [style=rec] (5) at (-9, 0) {$\overline{\bb_{D}(i,1)}$};
	\end{pgfonlayer}
	\begin{pgfonlayer}{edgelayer}
		\draw [style=edge] (3) to (0);
		\draw [style=edge] (4) to (1);
		\draw [style=edge] (5) to (2);
		\draw [style=edge] (0) to (4);
		\draw [style=edge] (1) to (5);
		\draw [style=edge] (2) to (3);
	\end{pgfonlayer}
\end{tikzpicture}
\caption{If $p$ is pareto optimal in both $P_{1}(m,n)$ and $P_{2}(m,n)$, then every potential trading cycle in $P_{3}(D)$ has this form; and it occurs when $P_{\ind_{D}(i,1)}^{\mathrm{false}},P_{\ind_{D}(i,2)}^{\mathrm{false}},P_{\ind_{D}(i,3)}^{\mathrm{false}} \subseteq p$.}
\label{fig:tgForTheThirdClaim3Sat}
\end{figure}

\begin{proof}
($\Rightarrow$): Assume that $p$ is pareto optimal in $P_{1}(m,n), P_{2}(m,n)$, and $P_{3}(D)$. \Clcref{claim:threesatclaim2} implies that $p$ satisfies the first condition. We show that $p$ also satisfies the second condition: Observe that every tuple of the form $(a_{i,\ind_{D}(i,3)},\overline{\bb_{D}(i,3)} , a_{i,\ind_{D}(i,2)},\overline{\bb_{D}(i,2)} , a_{i,\ind_{D}(i,1)}$ $,\overline{\bb_{D}(i,1)})$ is a trading cycle in $P_{3}(D)$ (see Figure \ref{fig:tgForTheThirdClaim3Sat}). Then by \Prcref{prop:po-iff-no-tc-and-sl}, for each $i \in [m]$, there must exists $j \in [3]$ such that $p(a_{i,\ind_{D}(i,j)}) = \bb_{D}(i,j)$.

($\Leftarrow$): Assume that $p$ satisfies both conditions. By \Clcref{claim:threesatclaim2}, $p$ is pareto optimal in both $P_{1}(m,n)$ and $P_{2}(m,n)$. We claim that $p$ is also pareto optimal in $P_{3}(D)$: First, $p$ does not admit self loops in $P_{3}(D)$ since all the items are allocated. Second, observe that every trading cycle in $P_{3}(D)$ is of the form $(a_{i,\ind_{D}(i,3)},\overline{\bb_{D}(i,3)} , a_{i,\ind_{D}(i,2)},\overline{\bb_{D}(i,2)} , a_{i,\ind_{D}(i,1)},\overline{\bb_{D}(i,1)})$. By the second condition (in each clause there exists at least one ``satisfied'' literal), $p$ does not allow cycles of such form. Thus, we conclude that $p$ pareto optimal in $P_{3}(D)$.\qed
\end{proof}

\begin{lemma}\label{lemma:ThreeSatYesInstanceIffPO}
An instance $D = (\XX,\CC)$ of \ThreeSat\ such that $|\XX|=n$ and $|\CC|=m$ is a \yes-instance\ if and only if there exists an assignment $p : A(m,n) \rightarrow I(m,n) \cup \lbrace b_{\emptyset} \rbrace$ that is pareto optimal in $P_{1}(m,n)$, $P_{2}(m,n)$, and $P_{3}(D)$. 
\end{lemma}

\begin{proof}
($\Rightarrow$): Assume that $D$ is a \yesinstance. Then, there exists a boolean assignment $\varphi : \XX \rightarrow \lbrace T,F \rbrace$ that satisfies every clause in $\CC$. We construct an assignment $p: A(m,n) \rightarrow I(m,n) \cup \lbrace b_{\emptyset} \rbrace$ as follows. For each $x_{j} \in \XX$, if $\varphi(x_{j})=T$, we add $P_{j}^{\mathrm{true}}$ to $p$, and otherwise, we add $P_{j}^{\mathrm{false}}$ to $p$. We claim that $p$ is pareto optimal in all three preference profiles. First, by \Clcref{claim:threesatclaim2}, $p$ is pareto optimal in $P_{1}(m,n)$ and in $P_{2}(m,n)$. Second, it is pareto optimal in the third layer: Since $\varphi$ satisfies every clause, we have that for each $C_{i} = \ell_{i}^{1} \vee \ell_{i}^{2} \vee \ell_{i}^{3} \in \CC$, there exists $j \in [3]$ such that (1) if $\ell_{i}^{j}$ is positive, then $\varphi(x_{\ind_{D}(i,j)})=T$, and (2) if $\ell_{i}^{j}$ is negative, then $\varphi(x_{\ind_{D}(i,j)})=F$. By the definition of $\bb_{D}(i,j)$ and the construction of $p$, we have that $p(a_{i,j}) = \bb_{D}(i,j)$, and by \Clcref{claim:threesatclaim3}, $p$ is pareto optimal in $P_{3}(D)$.

($\Leftarrow$): Assume that there exists an assignment $p$ that is pareto optimal in $P_{1}(m,n),P_{2}(m,n)$ and $P_{3}(D)$. By \Clcref{claim:threesatclaim3}, for each $j \in [n]$, either $P_{j}^{\mathrm{true}} \subseteq p$ or $P_{j}^{\mathrm{false}} \subseteq p$. We define a boolean assignment for $(\XX,\CC)$ by $\varphi(x_{j}) = T$ if $P_{j}^{\mathrm{true}} \subseteq p$, and $\varphi(x_{j}) = F$ if $P_{j}^{\mathrm{false}} \subseteq p$. By \Clcref{claim:threesatclaim3}, for each clause $C_{i}= \ell_{i}^{1} \vee \ell_{i}^{2} \vee \ell_{i}^{3} \in \CC$, there exists $j \in [3]$ such that $p(a_{i,\ind_{D}(i,j)}) = \bb_{D}(i,j)$. By the construction of $\varphi$, $\varphi(\ell_{i}^{j}) = T$. Hence, $\varphi$ satisfies each clause in $C$, and therefore $D$ is a \yes-instance.\qed
\end{proof}

We rely on these results to design a polynomial reduction from \ThreeSat\ to \goassign.

\begin{theorem}\label{theorem:threesat-reduction}
\ThreeSat\ is polynomial-time reducible to \goassign\, where $\alpha = \ell = 3$ and $d = 3$ or where $\alpha = \ell = 4$ and $d = m$.
\end{theorem}
\begin{proof}
Given an instance $D = (\XX,\CC)$ of \ThreeSat, such that $|\XX| = n$, and $|\CC| = m$, we construct an instance of \goassign\ with the agent set $A(m,n)$, and the item set $I(m,n)$, consisting of three layers. The first layer contains $P_{1}(m,n)$, the second contains $P_{2}(m,n)$, and the third contains $P_{3}(D)$, and we finally set $\alpha = \ell = 3$. The reduction clearly can be done in polynomial time. The correctness of the reduction is derived by \Lecref{lemma:ThreeSatYesInstanceIffPO}, which implies that there exists a 3-globally optimal assignment for the constructed instance if and only if $D$ is a \yes-instance. Notice that we can add an additional layer to the constructed instance where the preference lists of all the agents are complete and equal to each other (in this case, $d = m$). In this case, all the agents respect each other and therefore each assignment that allocates all the items is pareto optimal. This implies that there exists a 4-globally optimal assignment for the constructed instance if and only if $D$ is a \yes-instance. \qed  
\end{proof}

Since $\ell  + d$ (or $\ell + (m - d)$) in the statement above is upper bounded by fixed constants, we conclude the following.
\begin{corollary} \label{corollary:np-hardness-for-ell-alpha-d} 
\goassign\ is \paraH\ with respect to the parameters $\ell + d$ and $\ell + (m - d)$.
\end{corollary}

\section{Non-Existence of Polynomial Kernels}
\label{sec:part1Kernelization}

Before we present our next results, let us provide the following observation on which we will rely in the next three cross-compositions.

\begin{observation} \label{observation:threeSatObservation}
Let $D = (\XX,\CC)$ be an instance of \ThreeSat\ of size $n$. Then there exists an equivalent \ThreeSat\ instance $U = (\XX_{U},\CC_{U})$ where $\XX_{U} = \lbrace x_{1},\ldots,x_{n} \rbrace$ and $|\CC_{U}| = n$.
\end{observation}

\begin{proof}
Since $D$ has size $n$, it has at most $n$ variables and at most $n$ clauses. Then there exists a one to one mapping from $\XX$ to $\XX_{U}$. We pick such arbitrary mapping and replace each variable in $D$ with its corresponding variable in $X_{U}$. We then append the clause set with copies of some existing clause until it has size $n$, and denote the resulting clause multi-set by $\CC_{U}$. The new instance $U = (\XX_{U},\CC_{U})$ is clearly equivalent to $D$.\qed
\end{proof}

\begin{theorem} \label{theorem:nopoly-kernel-agents-items-alpha} 
There does not exist a polynomial kernel for \goassign\ with respect to the parameters $k_{1} = n + m +\alpha$ and $k_{2} = n + m +(\ell - \alpha)$, unless \NP$\subseteq $\coNPpoly.
\end{theorem}
\begin{proof}
We provide two cross-compositions from \ThreeSat\ to \goassign. Given instances of \ThreeSat\ $D_{0}=(\XX_{0},\CC_{0}),\ldots,D_{t-1}=(\XX_{t-1},\CC_{t-1})$ of the same size $n \in \mathbb{N}$ for some $t \in \mathbb{N}$, we first modify each instance $D_{i}$ to have $\XX_{i} = \lbrace x_{1}, \ldots , x_{n} \rbrace$ and $|\CC_{i}|=n$ using \Obcref{observation:threeSatObservation}. The two algorithms construct instances of \goassign\ with $2n ^{2} + \lceil 2 \log{t} \rceil$ agents, $2n ^{2} + \lceil 2 \log{t} \rceil$ items, and $2+t$ layers, that share the same agent set and item set (Shown in Figure \ref{fig:generalKernel}). We first provide the first two layers, which are identical in both constructions, and then we provide the remaining $t$ layers for each construction separately. We create the agent sets $A(n,n) = \lbrace a_{i,j} \mid i,j \in [n] \rbrace \cup \lbrace \overline{a_{i,j}} \mid i,j \in [n] \rbrace$ and $A_{t} = \lbrace c_{i} \mid i \in [\lceil \log{t} \rceil] \rbrace \cup \lbrace \overline{c_{i}} \mid i \in [\lceil \log{t} \rceil] \rbrace$, and the item sets $I(n,n) = \lbrace b_{i,j} \mid i,j \in [n] \rbrace \cup \lbrace \overline{b_{i,j}} \mid i,j \in [n] \rbrace$ and $I_{t} = \lbrace d_{i} \mid i \in [\lceil \log{t} \rceil] \rbrace \cup \lbrace \overline{d_{i}} \mid i \in [\lceil \log{t} \rceil] \rbrace$. Both constructions are defined over the agent set $A = A(m,n) \cup A_{t}$ and the item set $I = I(m,n) \cup I_{t}$.
Intuitively, we use $A(m,n)$ and $I(m,n)$ in order to construct the preference profiles $P_{1}(m,n), P_{2}(m,n)$ and $P_{3}(D_{i})$ for each $i \in \lbrace 0 , \ldots, t-1  \rbrace$. The agents from $A_{t}$ ``choose'' a \yes-instance\ (if one exists) from $D_{0}, \ldots , D_{t-1}$; their goal is to enforce an assignment to be pareto optimal in (at most) one profile among $P_{3}(D_{0}),\ldots,P_{3}(D_{t-1})$. 

The first layer is a composition of $P_{1}(n,n)$ with additional $2 \lceil \log{t} \rceil$ preference lists that belong to the agents from $A_{t}$ and is defined as follows: 

\begin{framed}[.5\textwidth]
\begin{itemize}
\item $a_{i,j}$ : $b_{i,j}>\overline{b_{i,j}}\ \ \forall i,j \in [n]$
\item $\overline{a_{i,j}}$ : $b_{i,j}>\overline{b_{i,j}}\ \ \forall i,j \in [m]$
\item $c_{i}$ : $d_{i}>\overline{d_{i}}\ \ \forall i \in [\lceil \log{t} \rceil]$
\item $\overline{c_{i}}$ : $d_{i}>\overline{d_{i}}\ \ \forall i \in [\lceil \log{t} \rceil]$
\end{itemize}
\end{framed}

The second layer consists of the preference profile $P_{2}(n,n)$ together with the same preferences of the agents from $A_{t}$ as in the first layer: 

\begin{framed}[.9\textwidth]
\begin{itemize}
\item $a_{n,j}$ : $b_{n,j}>b_{n-1,j}>\overline{b_{n,j}}\ \ \forall j \in [n]$
\item $\overline{a_{n,j}}\ $: $\ b_{n,j}>b_{n-1,j}>\overline{b_{n,j}}\ \ \forall j \in [n]$
\item $a_{i,j}$ : $b_{i-1,j}>\overline{b_{i,j}}>\overline{b_{i+1,j}}>b_{i,j}\ \ \forall i \in \lbrace 2,\ldots,n-1 \rbrace$,$ j \in [n]$
\item $\overline{a_{i,j}}$ : $b_{i-1,j}>\overline{b_{i,j}}>\overline{b_{i+1,j}}>b_{i,j}\ \ \forall i \in \lbrace 2,\ldots,n-1 \rbrace$,$ j \in [n]$
\item $a_{1,j}$ : $\overline{b_{1,j}}>\overline{b_{2,j}}>b_{1,j}\ \ \forall j \in [n]$
\item $\overline{a_{1,j}}$ : $\overline{b_{1,j}}>\overline{b_{2,j}}>b_{1,j}\ \ \forall j \in [n]$
\item $c_{i}$ : $d_{i}>\overline{d_{i}}\ \ \forall i \in [\lceil \log{t} \rceil]$
\item $\overline{c_{i}}$ : $d_{i}>\overline{d_{i}}\ \ \forall i \in [\lceil \log{t} \rceil]$
\end{itemize}
\end{framed}

\begin{figure}[t]
\centering
\scalebox{0.7}{
\input{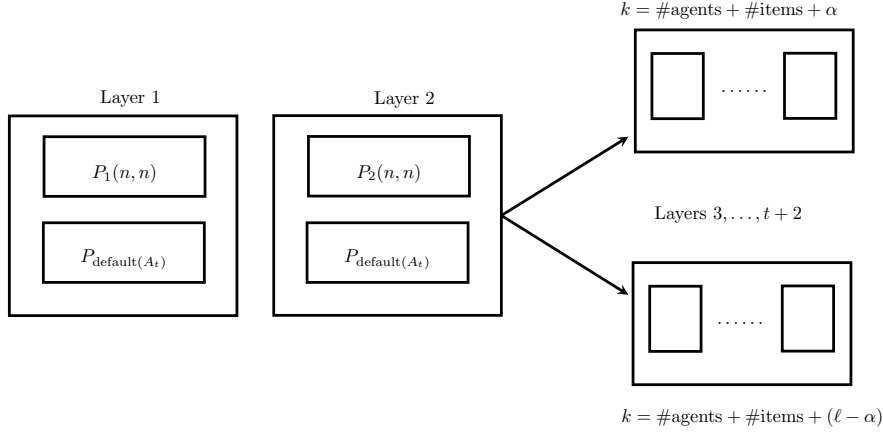}
}
\caption{The first two layers are identical in both cross-compositions; $P_{\mathrm{default}}(A_{t})$ is the preference profile in which both $c_{i}$ and $\overline{c_{i}}$ have the preference list $d_{i}>\overline{d_{i}}$ for each $i \in [\lceil \log{t} \rceil]$. The first layer is a composition of $P_{1}(n,n)$ and $P_{\mathrm{default}}(A_{t})$, and the second layer is a composition of $P_{2}(n,n)$ and $P_{\mathrm{default}}(A_{t})$. The rest $t$ layers will be defined separately for each parameter.}
\label{fig:generalKernel}
\end{figure}

\begin{definition}
For each $j \in [\lceil \log{t} \rceil]$ and $b \in \lbrace 0,1 \rbrace$, we define
$
  \dd(j,b) =
  \begin{cases}
    d_{j} & b=0\\
    \overline{d_{j}} & b=1
  \end{cases} 
$.
\end{definition}
We also denote $M_{j}^{0} = \lbrace (c_{j},d_{j}),(\overline{c_{j}},\overline{d_{j}})\rbrace$ and $M_{j}^{1} = \lbrace (c_{j},\overline{d_{j}}),(\overline{c_{j}},d_{j})\rbrace$. Informally speaking, the sets $M_{j}^{0}$ and $M_{j}^{1}$ will correspond to the value of the $j$-th bit in the binary representation of $i \in \lbrace 0 , \ldots, t-1  \rbrace$ such that $D_{i}$ is a \yes-instance. The sets $M_{j}^{b}$ that are contained in a pareto optimal assignment will ``encode'' the index of some \yes-instance\ (if one exists). For each $i \in \mathbb{N}, j \in [\lceil \log{i} \rceil]$, we define $i[j]$  as the $j$-th bit in the binary representation of $i$, i.e.~$i = i[1]i[2] \ldots i[\lceil \log{t} \rceil]$. We denote the set $M_{i}$ by $M_{i} = \bigcup_{j=1}^{\lceil \log{t} \rceil}{M_{j}^{i[j]}}$.
We first claim the following.

\begin{myclaim}\label{claim:cc1c1}
An assignment $p : A \rightarrow I \cup \lbrace b_{\emptyset} \rbrace$ is pareto optimal in the first two layers if and only if it is pareto optimal in both $P_{1}(n,n)$ and $P_{2}(n,n)$, and there exists exactly one $i \in \lbrace 0 , \ldots, t-1  \rbrace$ such that $M_{i} \subseteq p$. 
\end{myclaim}
\begin{proof}

($\Rightarrow$): Assume that $p$ is pareto optimal in both layers. Since $P_{1}(n,n)$ is contained in the first layer, $p$ is pareto optimal in $P_{1}(n,n)$. Similarly, $p$ is pareto optimal in $P_{2}(n,n)$ since $P_{2}(n,n)$ is contained in the second layer. Observe that for each $i \in [\lceil \log{t} \rceil]$, $d_{i}$ and $\overline{d_{i}}$ are both acceptable by $c_{i}$ and $\overline{c_{i}}$; and are the only items acceptable by them. Then, we have that $\lbrace p(c_{i}), p(\overline{c_{i}}) \rbrace = \lbrace d_{i},\overline{d_{i}} \rbrace$ (else we get a self loop), implying that for each $j \in [\lceil \log{t} \rceil]$, either $M_{j}^{0} \subseteq p$ or $M_{j}^{1} \subseteq p$. Let $i[j] = 0$ if $M_{j}^{0} \subseteq p$ and $i[j] = 1$ if $M_{j}^{1} \subseteq p$. Let $i \in \mathbb{N}$ be a number whose binary representation is $i = i[1] \ldots i[\lceil \log{t} \rceil]$; then, we have that $M_{i} \subseteq p$. We prove that $i$ is unique. Towards a contradiction, suppose there exists $e \in \lbrace 0 , \ldots, t-1  \rbrace$ such that $e \neq i$ and $p$ satisfies that both $M_{i} \subseteq p$ and $M_{e} \subseteq p$. Since $e \neq i$, there exists a place in their binary representations where they differ, i.e.~there exists $r \in [\lceil \log{t} \rceil]$ such that $e[r] \neq i[r]$. Thus, both $M_{r}^{0} \subseteq p$ and $M_{r}^{1} \subseteq p$, and this gives us a contradiction.   

($\Leftarrow$): Assume that $p$ is pareto optimal in both $P_{1}(n,n)$ and $P_{2}(n,n)$, and there exists exactly one $i \in \lbrace 0 , \ldots, t-1  \rbrace$ such that $M_{i} \subseteq p$. First, we prove that $p$ is pareto optimal in the first layer: Observe that all the agents pairwise respect each other in the first layer; then, by \Lecref{lemma:respectedAgents}, $p$ does not admit trading cycles in this layer. Moreover, by \Clcref{claim:threesatclaim1}, $p$ allocates all the items, hence it does not admit self loops. Thus, \Prcref{prop:po-iff-no-tc-and-sl} implies that $p$ is pareto optimal in the first layer. Second, we prove that $p$ is pareto optimal in the second layer: Observe that the agents from $A_{t}$ respect each other in the second layer; then by \Lecref{lemma:respectedAgents}, $p$ does not admit trading cycles among them. Furthermore, notice that the set of items acceptable by the agents from $A(n,n)$ is disjoint from the set of items acceptable by the agents from $A_{t}$, thus $p$ cannot admit trading cycles that consist of both agents from $A_{t}$ and agents from $A(n,n)$. Lastly, since $p$ is pareto optimal in $P_{2}(n,n)$, \Prcref{prop:po-iff-no-tc-and-sl} implies that it does not admit trading cycles among agents from $A(n,n)$. Then, we have by \Prcref{prop:po-iff-no-tc-and-sl} that $p$ is pareto optimal in the second layer.\qed
\end{proof}

We define the remaining $t$ layers separately for each construction.

\smallskip
\noindent\textbf{The parameter $k = n + m +\alpha$.}
The cross-composition for this parameter constructs layer $3+i$ as a composition of $P_{3}(D_{i})$ with a unique combination of preference lists of length $1$ that belong to agents from $A_{t}$. It is defined as follows: 

\begin{framed}
\begin{itemize}
\item $a_{q,\ind_{D_{i}}(q,3)}\ $: $\ \bb_{D_{i}}(q,3)>\overline{\bb_{D_{i}}(q,2)}>\overline{\bb_{D_{i}}(q,3)}\ \ \forall q \in [m]$
\item $a_{q,\ind_{D_{i}}(q,2)}\ $: $\ \bb_{D_{i}}(q,2)>\overline{\bb_{D_{i}}(q,1)}>\overline{\bb_{D_{i}}(q,2)}\ \ \forall q \in [m]$
\item $a_{q,\ind_{D_{i}}(q,1)}\ $: $\ \bb_{D_{i}}(q,1)>\overline{\bb_{D_{i}}(q,3)}>\overline{\bb_{D_{i}}(q,1)}\ \ \forall q \in [m]$
\item $\overline{a_{q,\ind_{D_{i}}(q,r)}}\ $: $\ \bb_{D_{i}}(q,r)>\overline{\bb_{D_{i}}(q,r)}\ \ \forall q \in [m] , r \in [3]$
\item $a_{q,j}\ $: $\ b_{q,j}>\overline{b_{q,j}}\ \ \forall q \in [m], j \in [n]$ such that $x_{j}$ does not appear in the $q$-th clause in $\CC_{i}$
\item $\overline{a_{q,j}}\ $: $\ b_{q,j}>\overline{b_{q,j}}\ \ \forall q \in [m], j \in [n]$ such that $x_{j}$ does not appear in the $q$-th clause in $\CC_{i}$
\item $c_{j}\ $:$\ \dd(j,i[j])\ \ \forall j \in [\lceil \log{t} \rceil]$
\item $\overline{c_{j}}\ $:$\ \overline{\dd(j,i[j])}\ \forall j \in [\lceil \log{t} \rceil]$
\end{itemize}
\end{framed}

\begin{figure}[t]
\centering
\scalebox{0.7}{
\ifx\du\undefined
  \newlength{\du}
\fi
\setlength{\du}{15\unitlength}
\begin{tikzpicture}
\pgftransformxscale{1.000000}
\pgftransformyscale{-1.000000}
\definecolor{dialinecolor}{rgb}{0.000000, 0.000000, 0.000000}
\pgfsetstrokecolor{dialinecolor}
\definecolor{dialinecolor}{rgb}{1.000000, 1.000000, 1.000000}
\pgfsetfillcolor{dialinecolor}
\definecolor{dialinecolor}{rgb}{1.000000, 1.000000, 1.000000}
\pgfsetfillcolor{dialinecolor}
\fill (16.925000\du,8.017363\du)--(16.925000\du,15.217363\du)--(25.125000\du,15.217363\du)--(25.125000\du,8.017363\du)--cycle;
\pgfsetlinewidth{0.100000\du}
\pgfsetdash{}{0pt}
\pgfsetdash{}{0pt}
\pgfsetmiterjoin
\definecolor{dialinecolor}{rgb}{0.000000, 0.000000, 0.000000}
\pgfsetstrokecolor{dialinecolor}
\draw (16.925000\du,8.017363\du)--(16.925000\du,15.217363\du)--(25.125000\du,15.217363\du)--(25.125000\du,8.017363\du)--cycle;
\definecolor{dialinecolor}{rgb}{0.000000, 0.000000, 0.000000}
\pgfsetstrokecolor{dialinecolor}
\node at (21.025000\du,11.857363\du){};
\definecolor{dialinecolor}{rgb}{1.000000, 1.000000, 1.000000}
\pgfsetfillcolor{dialinecolor}
\fill (18.175000\du,8.767363\du)--(18.175000\du,10.917363\du)--(23.975000\du,10.917363\du)--(23.975000\du,8.767363\du)--cycle;
\pgfsetlinewidth{0.100000\du}
\pgfsetdash{}{0pt}
\pgfsetdash{}{0pt}
\pgfsetmiterjoin
\definecolor{dialinecolor}{rgb}{0.000000, 0.000000, 0.000000}
\pgfsetstrokecolor{dialinecolor}
\draw (18.175000\du,8.767363\du)--(18.175000\du,10.917363\du)--(23.975000\du,10.917363\du)--(23.975000\du,8.767363\du)--cycle;
\definecolor{dialinecolor}{rgb}{0.000000, 0.000000, 0.000000}
\pgfsetstrokecolor{dialinecolor}
\node at (21.075000\du,10.082363\du){$P_{3}(D_{1})$};
\definecolor{dialinecolor}{rgb}{0.000000, 0.000000, 0.000000}
\pgfsetstrokecolor{dialinecolor}
\node[anchor=west] at (21.075000\du,9.842363\du){};
\definecolor{dialinecolor}{rgb}{0.000000, 0.000000, 0.000000}
\pgfsetstrokecolor{dialinecolor}
\node[anchor=west] at (19.657429\du,7.459730\du){Layer 3};
\definecolor{dialinecolor}{rgb}{1.000000, 1.000000, 1.000000}
\pgfsetfillcolor{dialinecolor}
\fill (18.150000\du,11.870000\du)--(18.150000\du,14.020000\du)--(23.950000\du,14.020000\du)--(23.950000\du,11.870000\du)--cycle;
\pgfsetlinewidth{0.100000\du}
\pgfsetdash{}{0pt}
\pgfsetdash{}{0pt}
\pgfsetmiterjoin
\definecolor{dialinecolor}{rgb}{0.000000, 0.000000, 0.000000}
\pgfsetstrokecolor{dialinecolor}
\draw (18.150000\du,11.870000\du)--(18.150000\du,14.020000\du)--(23.950000\du,14.020000\du)--(23.950000\du,11.870000\du)--cycle;
\definecolor{dialinecolor}{rgb}{0.000000, 0.000000, 0.000000}
\pgfsetstrokecolor{dialinecolor}
\node at (21.050000\du,13.185000\du){$P(A_{t},1)$};
\definecolor{dialinecolor}{rgb}{0.000000, 0.000000, 0.000000}
\pgfsetstrokecolor{dialinecolor}
\node[anchor=west] at (25.425000\du,11.580000\du){$\ldots\ldots$};
\definecolor{dialinecolor}{rgb}{0.74, 0.74, 0.74}
\pgfsetfillcolor{dialinecolor}
\fill (27.975000\du,8.037633\du)--(27.975000\du,15.237633\du)--(36.175000\du,15.237633\du)--(36.175000\du,8.037633\du)--cycle;
\pgfsetlinewidth{0.100000\du}
\pgfsetdash{}{0pt}
\pgfsetdash{}{0pt}
\pgfsetmiterjoin
\definecolor{dialinecolor}{rgb}{0.000000, 0.000000, 0.000000}
\pgfsetstrokecolor{dialinecolor}
\draw (27.975000\du,8.037633\du)--(27.975000\du,15.237633\du)--(36.175000\du,15.237633\du)--(36.175000\du,8.037633\du)--cycle;
\definecolor{dialinecolor}{rgb}{0.000000, 0.000000, 0.000000}
\pgfsetstrokecolor{dialinecolor}
\node at (32.075000\du,11.877633\du){};
\definecolor{dialinecolor}{rgb}{0,0.8,0}
\pgfsetfillcolor{dialinecolor}
\fill (29.225000\du,8.817363\du)--(29.225000\du,10.967363\du)--(35.025000\du,10.967363\du)--(35.025000\du,8.817363\du)--cycle;
\pgfsetlinewidth{0.100000\du}
\pgfsetdash{}{0pt}
\pgfsetdash{}{0pt}
\pgfsetmiterjoin
\definecolor{dialinecolor}{rgb}{0.000000, 0.000000, 0.000000}
\pgfsetstrokecolor{dialinecolor}
\draw (29.225000\du,8.817363\du)--(29.225000\du,10.967363\du)--(35.025000\du,10.967363\du)--(35.025000\du,8.817363\du)--cycle;
\definecolor{dialinecolor}{rgb}{0.000000, 0.000000, 0.000000}
\pgfsetstrokecolor{dialinecolor}
\node at (32.125000\du,10.132363\du){$P_{3}(D_{i})$};
\definecolor{dialinecolor}{rgb}{0.000000, 0.000000, 0.000000}
\pgfsetstrokecolor{dialinecolor}
\node[anchor=west] at (32.125000\du,9.892363\du){};
\definecolor{dialinecolor}{rgb}{0.000000, 0.000000, 0.000000}
\pgfsetstrokecolor{dialinecolor}
\node[anchor=west] at (30.707429\du,7.480000\du){Layer 3+i};
\definecolor{dialinecolor}{rgb}{0,0.8,0}
\pgfsetfillcolor{dialinecolor}
\fill (29.200000\du,11.920000\du)--(29.200000\du,14.070000\du)--(35.000000\du,14.070000\du)--(35.000000\du,11.920000\du)--cycle;
\pgfsetlinewidth{0.100000\du}
\pgfsetdash{}{0pt}
\pgfsetdash{}{0pt}
\pgfsetmiterjoin
\definecolor{dialinecolor}{rgb}{0.000000, 0.000000, 0.000000}
\pgfsetstrokecolor{dialinecolor}
\draw (29.200000\du,11.920000\du)--(29.200000\du,14.070000\du)--(35.000000\du,14.070000\du)--(35.000000\du,11.920000\du)--cycle;
\definecolor{dialinecolor}{rgb}{0.000000, 0.000000, 0.000000}
\pgfsetstrokecolor{dialinecolor}
\node at (32.100000\du,13.235000\du){$P(A_{t},i)$};
\definecolor{dialinecolor}{rgb}{0.000000, 0.000000, 0.000000}
\pgfsetstrokecolor{dialinecolor}
\node[anchor=west] at (36.475000\du,11.630000\du){$\ldots\ldots$};
\definecolor{dialinecolor}{rgb}{1.000000, 1.000000, 1.000000}
\pgfsetfillcolor{dialinecolor}
\fill (38.985165\du,8.050919\du)--(38.985165\du,15.250919\du)--(47.185165\du,15.250919\du)--(47.185165\du,8.050919\du)--cycle;
\pgfsetlinewidth{0.100000\du}
\pgfsetdash{}{0pt}
\pgfsetdash{}{0pt}
\pgfsetmiterjoin
\definecolor{dialinecolor}{rgb}{0.000000, 0.000000, 0.000000}
\pgfsetstrokecolor{dialinecolor}
\draw (38.985165\du,8.050919\du)--(38.985165\du,15.250919\du)--(47.185165\du,15.250919\du)--(47.185165\du,8.050919\du)--cycle;
\definecolor{dialinecolor}{rgb}{0.000000, 0.000000, 0.000000}
\pgfsetstrokecolor{dialinecolor}
\node at (43.085165\du,11.890919\du){};
\definecolor{dialinecolor}{rgb}{1.000000, 1.000000, 1.000000}
\pgfsetfillcolor{dialinecolor}
\fill (40.235165\du,8.830649\du)--(40.235165\du,10.980649\du)--(46.035165\du,10.980649\du)--(46.035165\du,8.830649\du)--cycle;
\pgfsetlinewidth{0.100000\du}
\pgfsetdash{}{0pt}
\pgfsetdash{}{0pt}
\pgfsetmiterjoin
\definecolor{dialinecolor}{rgb}{0.000000, 0.000000, 0.000000}
\pgfsetstrokecolor{dialinecolor}
\draw (40.235165\du,8.830649\du)--(40.235165\du,10.980649\du)--(46.035165\du,10.980649\du)--(46.035165\du,8.830649\du)--cycle;
\definecolor{dialinecolor}{rgb}{0.000000, 0.000000, 0.000000}
\pgfsetstrokecolor{dialinecolor}
\node at (43.135165\du,10.145649\du){$P_{3}(D_{i})$};
\definecolor{dialinecolor}{rgb}{0.000000, 0.000000, 0.000000}
\pgfsetstrokecolor{dialinecolor}
\node[anchor=west] at (43.135165\du,9.905649\du){};
\definecolor{dialinecolor}{rgb}{0.000000, 0.000000, 0.000000}
\pgfsetstrokecolor{dialinecolor}
\node[anchor=west] at (41.717594\du,7.493286\du){Layer 3+(t-1)};
\definecolor{dialinecolor}{rgb}{1.000000, 1.000000, 1.000000}
\pgfsetfillcolor{dialinecolor}
\fill (40.210165\du,11.933286\du)--(40.210165\du,14.083286\du)--(46.010165\du,14.083286\du)--(46.010165\du,11.933286\du)--cycle;
\pgfsetlinewidth{0.100000\du}
\pgfsetdash{}{0pt}
\pgfsetdash{}{0pt}
\pgfsetmiterjoin
\definecolor{dialinecolor}{rgb}{0.000000, 0.000000, 0.000000}
\pgfsetstrokecolor{dialinecolor}
\draw (40.210165\du,11.933286\du)--(40.210165\du,14.083286\du)--(46.010165\du,14.083286\du)--(46.010165\du,11.933286\du)--cycle;
\definecolor{dialinecolor}{rgb}{0.000000, 0.000000, 0.000000}
\pgfsetstrokecolor{dialinecolor}
\node at (43.110165\du,13.248286\du){$P(A_{t},t)$};
\end{tikzpicture}
}
\caption{The last $t$ layers constructed for the parameter $n+m+\alpha$. An assignment can be pareto optimal in at most one layer among layers $3,\ldots,3+(t-1)$. If $p$ is a 3-globally optimal assignment, then the layer $3+i$ in which it is pareto optimal corresponds to a \yesinstance\ $D_{i}$; This $i \in \lbrace 0 , \ldots, t-1  \rbrace$ is encoded by the allocations of the agents from $A_{t}$; $P(A_{t},i)$ is the preference profile in which $c_{j}$ accepts only $\dd(j,i[j])$ and $\overline{c_{j}}$ accepts only $\overline{\dd(j,i[j])}$, for each $j \in [\lceil \log{t} \rceil]$.}
\label{fig:kernelAgentsPlusItemsPlusAlpha}
\end{figure}
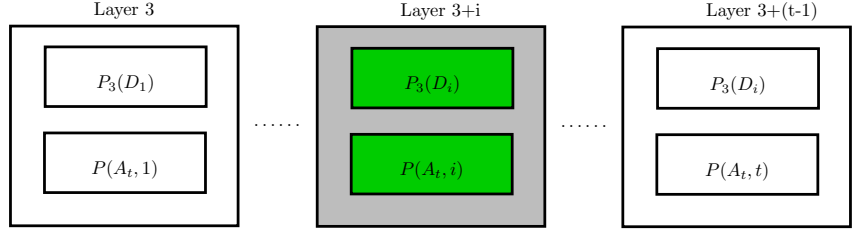

An illustration of the last $t$ layers in shown in Figure \ref{fig:kernelAgentsPlusItemsPlusAlpha}.Intuitively, the goal of the agents from $A_{t}$ is to enforce each assignment to be pareto optimal in at most one layer among $3,\ldots,3+(t-1)$. We finally set $\alpha=3$. Note that the construction can be done in time that is polynomial in $\Sigma_{i=0}^{t-1}|D_{i}|$. Before we prove the correctness of the construction, let us prove the following:

\begin{myclaim}\label{claim:cc1c2}
Let $p : A \rightarrow I \cup \lbrace b_{\emptyset} \rbrace$ be an assignment and $i \in \lbrace 0 , \ldots, t-1  \rbrace$. Then $p$ is pareto optimal in layer $3+i$ if and only if it is pareto optimal in $P_{3}(D_{i})$ and $M_{i} \subseteq p$. 
\end{myclaim}
\begin{proof}
($\Rightarrow$): Assume that $p$ is pareto optimal in layer $3+i$. First, since $P_{3}(D_{i})$ is contained in the preference profile of layer $3+i$, we infer that $p$ is pareto optimal in $P_{3}(D_{i})$. Second, observe that for each $i \in [\lceil \log{t} \rceil]$, $c_{i}$ only accepts $\dd(j,i[j])$ and $\overline{c_{j}}$ only accepts $\overline{\dd(j,i[j])}$. Furthermore, $\dd(j,i[j])$ and $\overline{\dd(j,i[j])}$ are only acceptable by these agents. This enforces $p$ to satisfy $p(c_{i})=\dd(j,i[j])$ and $p(\overline{c_{j}}) = \overline{\dd(j,i[j])}$ as otherwise, $p$ would admit self loops. Hence, $p$ satisfies $M_{i} \subseteq p$. 

($\Leftarrow$): Assume that $p$ is pareto optimal in $P_{3}(D_{i})$ and that $M_{i} \subseteq p$. Notice that $p$ allocates all the items, thus it cannot admit self loops in layer $3+i$. In addition, since each agent from $A_{t}$ only accepts a single item, $p$ cannot admit trading cycles among agents from $A_{t}$. Moreover, since $P_{3}(D_{i})$ is contained in the preference profile in layer $3+i$, and $p$ is pareto optimal in $P_{3}(D_{i})$, by \Prcref{prop:po-iff-no-tc-and-sl} we have that $p$ does not admit trading cycles among agents from $A(n,n)$. We have also that $p$ does not admit trading cycles that contain both agents from $A_{t}$ and agents from $A(n,n)$ since their sets of acceptable items are disjoint. Then we infer that $p$ is pareto optimal in layer $3+i$.\qed
\end{proof}

\begin{myclaim}\label{claim:cc1c3}
Let $p : A \rightarrow I \cup \lbrace b_{\emptyset} \rbrace$ be an assignment. Then there exists at most one $i \in \lbrace 0 , \ldots, t-1  \rbrace$ such that $p$ is pareto optimal in layer $3+i$. 
\end{myclaim}
\begin{proof}
Towards a contradiction, suppose there exist $i_{1},i_{2} \in \lbrace 0 , \ldots, t-1  \rbrace$ such that $i_{1} \neq i_{2}$ and $p$ is pareto optimal in both layers $3+i_{1}$ and $3+i_{2}$. By \Clcref{claim:cc1c2}, both $M_{i_{1}} \subseteq p$ and $M_{i_{2}} \subseteq p$, a contradiction to \Clcref{claim:cc1c1}.\qed
\end{proof}

We now prove that there exists $i \in \lbrace 0 , \ldots, t-1  \rbrace$ such that $D_{i}$ is a \yes-instance of \ThreeSat\ if and only if there exists a 3-globally optimal assignment for the constructed instance.

($\Rightarrow$): Assume there exists $i \in \lbrace 0 , \ldots, t-1  \rbrace$ such that $D_{i}$ is a \yes-instance. By \Lecref{lemma:ThreeSatYesInstanceIffPO}, there exists an assignment $p$ that is pareto optimal in the profiles $P_{1}(n,n), P_{2}(n,n)$ and $P_{3}(D_{i})$. We extend $p$ by $p \leftarrow p \cup M_{i}$. Since $p$ allocates all the items from $I$, it is pareto optimal in the first two layers by \Clcref{claim:cc1c1}. By \Clcref{claim:cc1c2}, it is pareto optimal in layer $3+i$. Thus, $p$ is 3-globally optimal for the constructed instance.

($\Leftarrow$): Suppose there exists a 3-globally optimal assignment for the constructed instance. By \Clcref{claim:cc1c3}, $p$ is pareto optimal in at most one among layers $3,\ldots,3+(t-1)$, say, layer $3+i$ where  $i \in \lbrace 0 , \ldots, t-1  \rbrace$. By \Clcref{claim:cc1c2}, $p$ is pareto optimal in $P_{3}(D_{i})$ and by \Clcref{claim:cc1c1}, it is pareto optimal in both $P_{1}(n,n)$ and $P_{2}(n,n)$. Thus we have by \Lecref{lemma:ThreeSatYesInstanceIffPO} that $D_{i}$ is a \yes-instance.

Since the constructed instance satisfies $n + m + \alpha = \OO(n^{2}+ \log{t})$, \Prcref{prop:noKern} implies that \goassign\ does not admit a polynomial kernel with respect to the parameter $ k = n + m + \alpha$, unless \NP$\subseteq $\coNPpoly.

\smallskip
\noindent\textbf{The parameter $k = n + m +(\ell - \alpha)$.}
Similarly to the previous cross-composition, the goal of the agents in $A_{t}$ here is to encode some $i \in \lbrace 0 , \ldots, t-1  \rbrace$ such that $D_{i}$ is a \yes-instance (if one exists). This is done by enforcing every $\ell$-pareto optimal assignment to be pareto optimal in some profile $P_{3}(D_{i})$ and preventing the existence of trading cycles in all layers $3+j$ where $j \neq i$. Let $i \in \lbrace 0 , \ldots, t-1  \rbrace$, we define layer $3+i$ as follows:

\begin{framed}
\begin{itemize}
\item $a_{q,\ind_{D_{i}}(q,3)}\ $: $\ \bb_{D_{i}}(q,3)>\overline{\bb_{D_{i}}(q,2)}>\overline{\bb_{D_{i}}(q,3)}\ \ \forall q \in [m]$
\item $a_{q,\ind_{D_{i}}(q,2)}\ $: $\ \bb_{D_{i}}(q,2)>\overline{\bb_{D_{i}}(q,1)}>\overline{\bb_{D_{i}}(q,2)}\ \ \forall q \in [m]$
\item $a_{q,\ind_{D_{i}}(q,1)}\ $: $\ \bb_{D_{i}}(q,1)>\dd(\lceil \log{t} \rceil , i[\lceil \log{t} \rceil])>\overline{\bb_{D_{i}}(q,1)}\ \ \forall q \in [m]$
\item $\overline{a_{q,\ind_{D_{i}}(q,r)}}\ $: $\ \bb_{D_{i}}(q,r)>\overline{\bb_{D_{i}}(q,r)}\ \ \forall q \in [m] , r \in [3]$
\item $a_{q,j}\ $: $\ b_{q,j}>\overline{b_{q,j}}\ \ \forall q \in [m], j \in [n]$ such that $x_{j}$ does not appear in the $q$-th clause in $\CC_{i}$
\item $\overline{a_{q,j}}\ $: $\ b_{q,j}>\overline{b_{q,j}}\ \ \forall q \in [m], j \in [n]$ such that $x_{j}$ does not appear in the $q$-th clause in $\CC_{i}$
\item $c_{j}\ $:$\ \overline{\dd(j,i[j])}>\dd(j-1,i[j-1])>\dd(j,i[j])\ \ \forall j \in \lbrace 2,\ldots,\lceil \log{t} \rceil \rbrace$
\item $c_{1}\ $:$\ \overline{\dd(1,i[1])}> \lbrace \overline{\bb_{D_{i}}(q,3)}\ |\  \forall q \in [m] \rbrace > \dd(1,i[1])$
\item $\overline{c_{j}}\ $:$\ \overline{\dd(j,i[j])}>\dd(j,i[j])\ \ \forall j \in [\lceil \log{t} \rceil ]$
\end{itemize}
\end{framed}

\begin{figure}[t]
\centering
\scalebox{0.7}{
\input{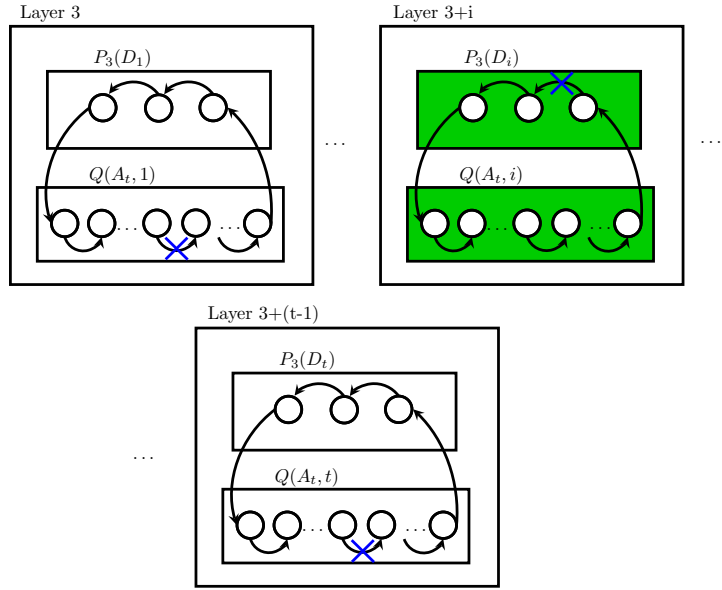}
}
\caption{The last $t$ layers constructed for the parameter $n+m+(\ell - \alpha)$. $Q(A_{t},i)$ consists of the preference lists of the agents from $A_{t}$ in layer $3+i$, for each $i \in \lbrace 0 , \ldots, t-1  \rbrace$. All the possible cycles in layers $3,\ldots,3+(t-1)$ can be decomposed into a path with 3 items in the trading graph of $P_{3}(D_{i})$, connected to a path with $\lceil \log{t} \rceil$ items in the trading graph of $Q(A_{t},i)$, and an edge connecting the end point of the second path to the start point of the first path. For an $\ell$-globally optimal assignment, the second path will be cut in most layers, then no cycles will exist. There will exist exactly one layer $3+i$ for which the trading graph of $Q(A_{t},i)$ contains such path. This implies that the trading graph of $P_{3}(D_{i})$ must not contain paths with 3 items, then $D_{i}$ is a \yesinstance.}
\label{fig:kernelAgentsPlusItemsPlusEllMinusAlpha}
\end{figure}

The construction is illustrated in Figure \ref{fig:kernelAgentsPlusItemsPlusEllMinusAlpha}. We now claim the following.

\begin{myclaim}\label{claim:cc1c4}
Let $p : A \rightarrow I \cup \lbrace b_{\emptyset} \rbrace$ be an assignment such that $p$ is pareto optimal in the first two layers and let $i \in \lbrace 0 , \ldots, t-1  \rbrace$. Then $p$ is pareto optimal in layer $3+i$ if and only if: (1) $p$ allocates all the items in $I$, and (2) $M_{i} \nsubseteq p$ or $p$ is pareto optimal in $P_{3}(D_{i})$.
\end{myclaim}
\begin{proof}

($\Rightarrow$): 
Assume that $p$ is pareto optimal in layer $3+i$. By \Clcref{claim:cc1c1}, $p$ is pareto optimal in both $P_{1}(n,n)$ and $P_{2}(n,n)$ and there exists exactly one $j \in \lbrace 0 , \ldots, t-1  \rbrace$ such that $M_{j} \subseteq p$. By the construction of the first layer, $p$ must allocate all the items since otherwise it would admit self loops. Notice that the only possible trading cycle in layer $3+i$, for any $i \in \lbrace 0 , \ldots, t-1  \rbrace$, has the form: $(a_{q,ind_{D_{i}}(q,3)},\overline{\bb_{D_{i}}(q,3)},a_{q,ind_{D_{i}}(q,2)},\overline{\bb_{D_{i}}(q,2)},a_{q,ind_{D_{i}}(q,1)}$ $,\overline{\bb_{D_{i}}(q,1)},c_{\lceil \log{t} \rceil},\dd(\lceil \log{t} \rceil,i[\lceil \log{t} \rceil]),\ldots,c_{1},\dd(1,i[1]))$ and it exists only when $p$ is not pareto optimal in $P_{3}(D_{i})$ and $M_{i} \subseteq p$.

($\Leftarrow$): Suppose that $p$ satisfies that $M_{i} \nsubseteq p$ or $p$ is pareto optimal in $P_{3}(D_{i})$. By the observation in the previous direction, $p$ does not admit trading cycles in layer $3+i$. Thus, $p$ is pareto optimal in layer $3+i$.\qed
\end{proof}

We prove that there exists $i \in \lbrace 0 , \ldots, t-1  \rbrace$ such that $D_{i}$ is a \yes-instance\ if and only if there exists an $\ell$-globally optimal assignment for the constructed instance.

($\Rightarrow$): Suppose there exists $i \in \lbrace 0 , \ldots, t-1  \rbrace$ such that $D_{i}$ is a \yes-instance. We extend $p$ by $p \leftarrow p \cup M_{i}$. By \Lecref{lemma:ThreeSatYesInstanceIffPO}, $p$ is pareto optimal in $P_{1}(n,n),P_{2}(n,n)$ and $P_{3}(D_{i})$. First, since all the items in $I_{t}$ are allocated, by Claims \ref{claim:cc1c2} and \ref{claim:cc1c3} we have that $p$ is pareto optimal in the first two layers. Second, by \Clcref{claim:cc1c4}, we have that $p$ is pareto optimal in layer $3+i$. Third, let $j \in \lbrace 0 , \ldots, t-1  \rbrace \setminus \lbrace i \rbrace$. By the construction of $p$ we have that $M_{j} \nsubseteq p$, then by \Clcref{claim:cc1c4}, $p$ is pareto optimal in layer $3+j$. Then we conclude that $p$ is $\ell$-globally optimal for the constructed instance.
 
($\Leftarrow$): Suppose that there exists an $\ell$-globally optimal assignment for the constructed instance $p$. By Claims \ref{claim:cc1c2} and \ref{claim:cc1c3}, $p$ is pareto optimal in both $P_{1}(n,n)$ and $P_{2}(n,n)$ and there exists $i \in \lbrace 0 , \ldots, t-1  \rbrace$ such that $M_{i} \subseteq p$. \Clcref{claim:cc1c4} implies that $p$ is pareto optimal in $P_{3}(D_{i})$. Thus, by \Lecref{lemma:ThreeSatYesInstanceIffPO}, $D_{i}$ is a \yes-instance.\qed
\end{proof}

We show the the same result holds also for the parameter $\nitems + \ell$.
\begin{theorem}\label{theorem:nopoly-kernel-items-l}
There does not exist a polynomial kernel for \goassign\ with respect to $k=\nitems + \ell$, unless \NP$\subseteq $\coNPpoly.
\end{theorem}
\begin{proof}
We present a cross-composition from \ThreeSat\ to \goassign. Given instances of \ThreeSat\ $D_{0}=(\XX_{0},\CC_{0}),\ldots,D_{t-1}=(\XX_{t-1},\CC_{t-1})$ of the same size $n \in \mathbb{N}$ for some $t \in \mathbb{N}$, we first modify each instance $D_{i}$ to have $\XX_{i} = \lbrace x_{1}, \ldots , x_{n} \rbrace$ and $|\CC_{i}|=n$ using \Obcref{observation:threeSatObservation}. We define an agent set $A_{i}(n,n)$ of $n^{2}$ agents for each instance $D_{i}$ by $A_{i}(n,n)=\lbrace a_{r,j}^{i}, \overline{a_{r,j}^{i}} \mid r,j \in [n] \rbrace$ and we set $A = \bigcup_{i=0}^{t-1}A_{i}(n,n)$. We also create the item set $I = I(n,n)$. The constructed instance is defined over $A$ and $I$, and it consists of $2\lceil \log{t} \rceil + 2$ layers. Notice that we have a total number of $2tn^{2}$ agents and $2n^{2}$ items, then in every assignment for the constructed instance, there will exist agents which get no items. Intuitively, the goal of the first $2 \lceil \log{t} \rceil$ layers is to enforce every $\alpha$-globally optimal assignment to allocate the items in $I$ only to agents that correspond to a \yes-instance\ (if one exists). They are constructed as compositions the profile $P_{1}(n,n)$ over the agent set $A_{i}(n,n)$ and the item set $I(n,n)$ for each $i \in \lbrace 0 , \ldots, t-1  \rbrace$ (namely, replacing each $a_{q,j}$ with $a_{q,j}^{i}$ in $P_{1}(n,n)$). Layer $2 \lceil \log{t} \rceil + 1$ is constructed as a composition of the profile $P_{2}(n,n)$ over $A_{i}(n,n)$ and $I(n,n)$ for each $i \in \lbrace 0 , \ldots, t-1  \rbrace$, and the last layer is a composition of the profiles $P_{3}(D_{i})$ over $A_{i}(n,n)$ and $I(n,n)$ for each $i \in \lbrace 0 , \ldots, t-1  \rbrace$. We first construct the first $\lceil \log{t} \rceil$ layers. Informally speaking, for $i \in [\lceil \log{t} \rceil]$, the preference profile in layer $i$ requires an assignment to allocate $b_{\emptyset}$ to any agent whose corresponding instance is $D_{j}$ such that the $i$-th bit in the binary representation of $j$ is $0$. Layer $i$ is formally defined as follows:

\begin{framed}[.9\textwidth]
\begin{itemize}
\item $a_{q,j}^{r}$ : $b_{q,j}>\overline{b_{q,j}}\ \ \forall q,j \in [n], r\in \lbrace 0 , \ldots, t-1 \rbrace \text{ such that } r[i]=0$
\item $\overline{a_{q,j}^{r}}$ : $b_{q,j}>\overline{b_{q,j}}\ \ \forall q,j \in [n], r\in \lbrace 0 , \ldots, t-1  \rbrace \text{ such that } r[i]=0$
\item $a_{q,j}^{r}$ : $\emptyset\ \ \forall q,j \in [n], r\in \lbrace 0 , \ldots, t-1  \rbrace \text{ such that } r[i]=1$
\item $\overline{a_{q,j}^{r}}$ : $\emptyset\ \ \forall q,j \in [n], r\in \lbrace 0 , \ldots, t-1  \rbrace \text{ such that } r[i]=1$
\end{itemize}
\end{framed}

We define the next $\lceil \log{t} \rceil$ layers. Intuitively, layer $\lceil \log{i} \rceil + i $ is different from layer $i$ such that it requires an assignment to allocate $b_{\emptyset}$ to any agent whose corresponding instance is $D_{j}$ such that the $i$-th bit in the binary representation of $j$ is $1$ (instead of $0$). For each $i \in [\lceil \log{t} \rceil]$, layer $\lceil \log{t} \rceil + i$ is defined as follows:

\begin{framed}[.9\textwidth]
\begin{itemize}
\item $a_{q,j}^{r}$ : $b_{q,j}>\overline{b_{q,j}}\ \ \forall q,j \in [n], r\in \lbrace 0 , \ldots, t-1  \rbrace \text{ such that } r[i]=1$
\item $\overline{a_{q,j}^{r}}$ : $b_{q,j}>\overline{b_{q,j}}\ \ \forall q,j \in [n], r\in \lbrace 0 , \ldots, t-1  \rbrace \text{ such that } r[i]=1$
\item $a_{q,j}^{r}$ : $\emptyset\ \ \forall q,j \in [n], r\in \lbrace 0 , \ldots, t-1  \rbrace \text{ such that } r[i]=0$
\item $\overline{a_{q,j}^{r}}$ : $\emptyset\ \ \forall q,j \in [n], r\in \lbrace 0 , \ldots, t-1  \rbrace \text{ such that } r[i]=0$
\end{itemize}
\end{framed}

Layer $2\lceil \log{t} \rceil +1$ is a composition of the profile $P_{2}(n,n)$ over the agent set $A_{i}(n,n)$ and the item set $I(n,n)$ for each $i \in [t]$.

\begin{framed}
\begin{itemize}
\item $a_{n,j}^{i}$ : $b_{n,j}>b_{n-1,j}>\overline{b_{n,j}}\ \ \forall i \in \lbrace 0 , \ldots, t-1  \rbrace, j \in [n]$
\item $\overline{a_{n,j}^{i}}\ $: $\ b_{n,j}>b_{n-1,j}>\overline{b_{n,j}}\ \ \forall i \in \lbrace 0 , \ldots, t-1  \rbrace,j \in [n]$
\item $a_{q,j}^{i}$ : $b_{q-1,j}>\overline{b_{q,j}}>\overline{b_{q+1,j}}>b_{q,j}\ \ \forall i \in \lbrace 0 , \ldots, t-1  \rbrace,q \in \lbrace 2,\ldots,n-1 \rbrace$,$ j \in [n]$
\item $\overline{a_{q,j}^{i}}$ : $b_{q-1,j}>\overline{b_{q,j}}>\overline{b_{q+1,j}}>b_{q,j}\ \ \forall i \in \lbrace 0 , \ldots, t-1  \rbrace,q \in \lbrace 2,\ldots,n-1 \rbrace$,$ j \in [n]$
\item $a_{1,j}^{i}$ : $\overline{b_{1,j}}>\overline{b_{2,j}}>b_{1,j}\ \ \forall i \in \lbrace 0 , \ldots, t-1  \rbrace,j \in [n]$
\item $\overline{a_{1,j}^{i}}$ : $\overline{b_{1,j}}>\overline{b_{2,j}}>b_{1,j}\ \ \forall i \in \lbrace 0 , \ldots, t-1  \rbrace,j \in [n]$
\end{itemize}
\end{framed}

Layer $2\lceil \log{t} \rceil +2$ is a composition of the profiles $P_{3}(D_{i})$ over $A_{i}(n,n)$ and $I(n,n)$ for each $i \in \lbrace 0 , \ldots, t-1  \rbrace$.

\begin{framed}
\begin{itemize}
\item $a_{j,\ind_{D_{i}}(j,3)}^{i}\ $: $\ \bb_{D_{i}}(j,3)>\overline{\bb_{D_{i}}(j,2)}>\overline{\bb_{D_{i}}(j,3)}\ \ \forall i \in \lbrace 0 , \ldots, t-1  \rbrace , j \in [n]$
\item $a_{j,\ind_{D_{i}}(j,2)}^{i}\ $: $\ \bb_{D_{i}}(j,2)>\overline{\bb_{D_{i}}(j,1)}>\overline{\bb_{D_{i}}(j,2)}\ \ \forall i \in \lbrace 0 , \ldots, t-1  \rbrace , j \in [n]$
\item $a_{j,\ind_{D_{i}}(j,1)}^{i}\ $: $\ \bb_{D_{i}}(j,1)>\overline{\bb_{D_{i}}(j,3)}>\overline{\bb_{D_{i}}(j,1)}\ \ \forall i \in \lbrace 0 , \ldots, t-1  \rbrace, j \in [n]$
\item $\overline{a_{j,\ind_{D_{i}}(j,r)}^{i}}\ $: $\ \bb_{D_{i}}(j,r)>\overline{\bb_{D_{i}}(j,r)}\ \ \forall i \in \lbrace 0 , \ldots, t-1  \rbrace , j \in [n] , r \in [3]$
\item $a_{q,r}^{i}\ $: $\ b_{q,r}>\overline{b_{q,r}}\ \ \forall i \in \lbrace 0 , \ldots, t-1  \rbrace , q,r \in [n]$ such that $x_{r}$ does not appear in the $q$-th clause in $\CC_{i}$
\item $\overline{a_{q,r}^{i}}\ $: $\ b_{q,r}>\overline{b_{q,r}}\ \ \forall i \in \lbrace 0 , \ldots, t-1  \rbrace , q,r \in [n]$ such that $x_{r}$ does not appear in the $q$-th clause in $\CC_{i}$
\end{itemize}
\end{framed} 
We finally set $\alpha = \lceil \log{t} \rceil +2$.
We claim the following:

\begin{myclaim} \label{claim1:ItemsPlusEll}
Let $p : A \rightarrow I \cup \lbrace b_{\emptyset} \rbrace$ be an assignment for the constructed instance and let $i \in \lceil \log{t} \rceil$. Then $p$ satisfies the following:
\begin{itemize}
\item $p$ is pareto optimal in layer $i$ if and only if it allocates all the items in $I$ to agents from the set $\lbrace a_{q,j}^{r} , \overline{a_{q,j}^{r}} \mid q,j \in [n], r \in \lbrace 0 , \ldots, t-1 \rbrace \text{ such that $r[i]=0$} \rbrace$.
\item $p$ is pareto optimal in layer $\lceil \log{t} \rceil + i$ if and only if it allocates all the items in $I$ to agents from the set $\lbrace a_{q,j}^{r} , \overline{a_{q,j}^{r}} \mid q,j \in [n], r \in \lbrace 0 , \ldots, t-1 \rbrace \text{ such that $r[i]=1$} \rbrace$.
\end{itemize}
\end{myclaim}
\begin{proof}
We provide a proof for the first condition (the proof for the second is symmetric).
 
($\Rightarrow$): Assume that $p$ is pareto optimal in layer $i$. Since the preference lists of agents $a_{q,j}^{r}$ and $\overline{a_{q,j}^{r}}$ for all $q,j \in [n]$ and $r \in \lbrace 0 , \ldots, t-1 \rbrace$ such that $r[i]=1$ are empty, $p$ allocates $b_{\emptyset}$ to all such agents, and therefore allocates items only to agents from $\lbrace a_{q,j}^{r} , \overline{a_{q,j}^{r}} \mid q,j \in [n], r \in \lbrace 0 , \ldots, t-1 \rbrace \text{ such that $r[i]=0$} \rbrace$. In addition, observe that each item in $I$ is acceptable by at least two agents from the set $\lbrace a_{q,j}^{r} , \overline{a_{q,j}^{r}} \mid q,j \in [n], r \in \lbrace 0 , \ldots, t-1 \rbrace \text{ such that $r[i]=0$} \rbrace$. By \Prcref{prop:po-iff-no-tc-and-sl}, $p$ does not admit self loops. Thus we have that all the items in $I$ are allocated.

($\Leftarrow$): Suppose that $p$ allocates all the items in $I$ to agents from the set $\lbrace a_{q,j}^{r} , \overline{a_{q,j}^{r}} \mid q,j \in [n], r \in \lbrace 0 , \ldots, t-1 \rbrace \text{ such that $r[i]=0$} \rbrace$. First, since $p$ allocates all the items, it does not admit self loops in layer $i$. Second, observe that all the agents pairwise respect each other in this layer, thus by \Lecref{lemma:respectedAgents}, $p$ does not admit trading cycle. By \Prcref{prop:po-iff-no-tc-and-sl}, $p$ is pareto optimal in layer $i$.\qed
\end{proof}

\begin{myclaim}\label{claim2:ItemsPlusEll}
Let $p : A \rightarrow I \cup \lbrace b_{\emptyset} \rbrace$ be an assignment for the constructed instance. Then there does not exist $i \in \lceil \log{t} \rceil$ such that $p$ is pareto optimal in both layers $i$ and $\lceil \log{t} \rceil + i$. 
\end{myclaim}
\begin{proof}
Towards a contradiction, suppose there exists $i \in \lceil \log{t} \rceil$ such that $p$ is pareto optimal in both layers $i$ and $\lceil \log{t} \rceil + i$. By \Clcref{claim1:ItemsPlusEll}, $p$ allocates all the items in $I$ to agents from both $\lbrace a_{q,j}^{r} , \overline{a_{q,j}^{r}} \mid q,j \in [n], r \in \lbrace 0 , \ldots, t-1 \rbrace \text{ such that $r[i]=0$} \rbrace$ and $\lbrace a_{q,j}^{r} , \overline{a_{q,j}^{r}} \mid q,j \in [n], r \in \lbrace 0 , \ldots, t-1 \rbrace \text{ such that $r[i]=1$} \rbrace$, but the intersection of these two sets is empty, a contradiction.\qed
\end{proof}

\begin{myclaim}\label{corollary1:ItemsPlusEll}
Let $p : A \rightarrow I \cup \lbrace b_{\emptyset} \rbrace$ be an assignment for the constructed instance. Then $p$ is pareto optimal in at most $\lceil \log{t} \rceil$ layers among the first $2 \lceil \log{t} \rceil$ layers. 
\end{myclaim}
\begin{proof}
\Clcref{claim2:ItemsPlusEll} implies that for each $i \in \lceil \log{t} \rceil$, $p$ is pareto optimal in at most one among layers $i$ and $\lceil \log{t} \rceil + i$. Implying that $p$ is pareto optimal in at most $\lceil \log{t} \rceil$ layers among the first $2 \lceil \log{t} \rceil$ layers.\qed
\end{proof}

Let $i \in \lbrace 0 , \ldots, t-1 \rbrace$, and let $p$ be an assignment for the constructed instance. We say that $p$ {\em encodes} $i$ in the first $2 \lceil \log{t} \rceil$ layers if it satisfies the following:
\begin{itemize}
\item For each $j \in [\lceil \log{t} \rceil]$ such that $i[j]=0$, $p$ is pareto optimal in layer $j$.
\item For each $j \in [\lceil \log{t} \rceil]$ such that $i[j]=1$, $p$ is pareto optimal in layer $\lceil \log{t} \rceil + j$.
\end{itemize}

\begin{myclaim}\label{corollary2:ItemsPlusEll}
Let $p : A \rightarrow I \cup \lbrace b_{\emptyset} \rbrace$ be an $\alpha$-globally optimal assignment for the constructed instance. Then there exists $i \in \lbrace 0 , \ldots, t-1 \rbrace$ that $p$ encodes in the first $2\lceil \log{t} \rceil$ layers. 
\end{myclaim}
\begin{proof}
Since $\alpha = \lceil \log{t} \rceil +2$, \Clcref{corollary1:ItemsPlusEll} implies that $p$ is pareto optimal in exactly $\lceil \log{t} \rceil$ layers among the first $2 \lceil \log{t} \rceil$ layers. By \Clcref{claim2:ItemsPlusEll}, for each $j \in [\lceil \log{t} \rceil]$, $p$ is pareto optimal in either layer $j$ or layer $\lceil \log{t} \rceil +j$. Denote $i[j]=0$ if $p$ is pareto optimal in layer $j$ and $i[j]=1$ if $p$ is pareto optimal in layer $\lceil \log{t} \rceil + j$. Observe that $p$ encodes $i=i[1]i[2]\ldots i[\lceil \log{t} \rceil]$ in the first $2 \lceil \log{t} \rceil$ layers.\qed
\end{proof}

\begin{myclaim}\label{claim3:ItemsPlusEll}
Let $p : A \rightarrow I \cup \lbrace b_{\emptyset} \rbrace$ be an assignment for the constructed instance. Then there exists $i \in \lbrace 0 , \ldots, t-1 \rbrace$ such that $p$ encodes $i$ in the first $2 \lceil \log{t} \rceil$ layers if and only if $p$ is pareto optimal in $P_{1}(n,n)$ over the agent set $A_{i}(n,n)$ and the item set $I(n,n)$. 
\end{myclaim}
\begin{proof}
($\Rightarrow$): Assume that there exists $i \in \lbrace 0 , \ldots, t-1 \rbrace$ that $p$ encodes in the first $2 \lceil \log{t} \rceil$ layers. 
Observe that for each $j \in [\lceil \log{t} \rceil]$, if $i[j]=0$ then $p$ is pareto optimal in layer $j$, and if $i[j]=1$ then $p$ is pareto optimal in layer $\lceil \log{t} \rceil +j$. By \Clcref{claim1:ItemsPlusEll}, $p$ allocates all the items in $I$ to agents from $\bigcap_{j=1}^{\lceil \log{t} \rceil}{\lbrace a_{q,s}^{r},\overline{a_{q,s}^{r}} \mid q,s \in [n], r \in \lbrace 0 , \ldots, t-1 \rbrace \text{ such that $r[j]=i[j]$} \rbrace }= A_{i}(n,n)$. Since the agents in $A_{i}(n,n)$ respect each other in $P_{1}(n,n)$ over $A_{i}(n,n)$ and $I(n,n)$, and all the items are allocated to $A_{i}(n,n)$, we have by \Lecref{lemma:respectedAgents} and \Prcref{prop:po-iff-no-tc-and-sl} that $p$ is pareto optimal in $P_{1}(n,n)$ over $A_{i}(n,n)$ and $I(n,n)$.

($\Leftarrow$): Suppose that $p$ is pareto optimal in $P_{1}(n,n)$ over $A_{i}(n,n)$ and $I(n,n)$. We claim that $p$ encodes $i$ in the first $2 \lceil \log{t} \rceil$ layers: By \Clcref{claim1:ItemsPlusEll}, for each $j \in [\lceil \log{t} \rceil]$, if $i[j] = 0$ then $p$ is pareto optimal in layer $j$, and if $i[j]=1$ then $p$ is pareo optimal in layer $\lceil \log{t} \rceil + j$. Then we have that $p$ encodes $i$ in the first $2 \lceil \log{t} \rceil$ layers.\qed
\end{proof}

\begin{myclaim} \label{claim4:ItemsPlusEll}
Let $p : A \rightarrow I \cup \lbrace b_{\emptyset} \rbrace$ be an assignment such that there exists $i \in \lbrace 0 , \ldots, t-1 \rbrace$ such that $p$ encodes $i$ in the first $2 \lceil \log{t} \rceil$ layers. Then $p$ is pareto optimal in layer $2 \lceil \log{t} \rceil +1 $ if and only if it is pareto optimal in $P_{2}(n,n)$ over $A_{i}(n,n)$ and $I(n,n)$.
\end{myclaim}
\begin{proof}

($\Rightarrow$): Assume that $p$ is pareto optimal in layer $2 \lceil \log{t} \rceil +1$. Since $P_{2}(n,n)$ over $A_{i}(n,n)$ and $I(n,n)$ is contained in the preference profile of layer $2 \lceil \log{t} \rceil +1 $, $p$ must be pareto optimal in this sub-profile.
 
($\Leftarrow$): Assume that $p$ is pareto optimal in $P_{3}(n,n)$ over $A_{i}(n,n)$ and $I(n,n)$. By the construction of $P_{3}(n,n)$ over $A_{i}(n,n)$ and $I(n,n)$, $p$ must allocate all the items to agents from $A_{i}(n,n)$, as otherwise, it would admit self loops. Observe that agents from $A \setminus A_{i}(n,n)$ do not admit trading cycles since they are assigned to $b_{\emptyset}$, and they also do not admit self loops since all their acceptable items are already allocated by $p$.\qed
\end{proof}
\begin{myclaim} \label{claim5:ItemsPlusEll}
Let $p : A \rightarrow I \cup \lbrace b_{\emptyset} \rbrace$ be an assignment such that there exists $i \in \lbrace 0 , \ldots, t-1 \rbrace$ such that $p$ encodes $i$ in the first $2 \lceil \log{t} \rceil$ layers. Then $p$ is pareto optimal in layer $2 \lceil \log{t} \rceil +2 $ if and only if it is pareto optimal in $P_{3}(D_{i})$ over $A_{i}(n,n)$ and $I(n,n)$.
\end{myclaim}
\begin{proof}

($\Rightarrow$): Assume that $p$ is pareto optimal in layer $2 \lceil \log{t} \rceil +2$. Since $P_{3}(n,n)$ over $A_{i}(n,n)$ and $I(n,n)$ is contained in the preference profile of layer $2 \lceil \log{t} \rceil +1 $, $p$ must be pareto optimal in this sub-profile.
 
($\Leftarrow$): Assume that $p$ is pareto optimal in $P_{3}(n,n)$ over $A_{i}(n,n)$ and $I(n,n)$. By the construction of $P_{3}(n,n)$ over $A_{i}(n,n)$ and $I(n,n)$, $p$ must allocate all the items to agents from $A_{i}(n,n)$, as otherwise, it would admit self loops. Observe that agents in $A \setminus A_{i}(n,n)$ do not admit trading cycles since they are assigned to $b_{\emptyset}$, they also do not admit self loops since all their acceptable items are already allocated by $p$.\qed
\end{proof}

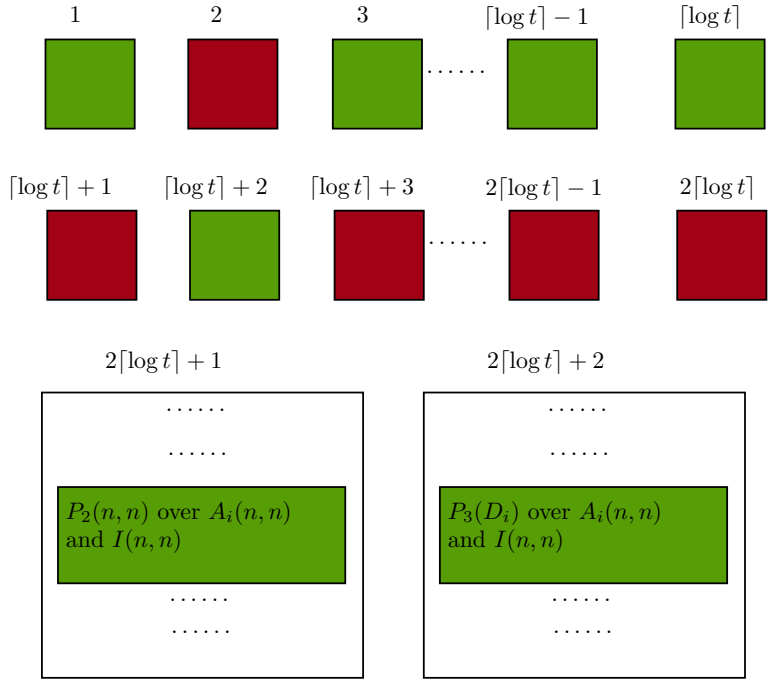
\begin{figure}[t]
\centering
\scalebox{0.9}{
\tikzset{every picture/.style={line width=0.75pt}} 

\begin{tikzpicture}[x=0.75pt,y=0.75pt,yscale=-1,xscale=1]

\draw  [fill={rgb, 255:red, 87; green, 158; blue, 6 }  ,fill opacity=1 ] (60,30) -- (110,30) -- (110,80) -- (60,80) -- cycle ;
\draw  [fill={rgb, 255:red, 163; green, 2; blue, 22 }  ,fill opacity=1 ] (140,30) -- (190,30) -- (190,80) -- (140,80) -- cycle ;
\draw  [fill={rgb, 255:red, 87; green, 158; blue, 6 }  ,fill opacity=1 ] (221,30) -- (271,30) -- (271,80) -- (221,80) -- cycle ;
\draw  [fill={rgb, 255:red, 87; green, 158; blue, 6 }  ,fill opacity=1 ] (319,30) -- (369,30) -- (369,80) -- (319,80) -- cycle ;
\draw  [fill={rgb, 255:red, 87; green, 158; blue, 6 }  ,fill opacity=1 ] (413,30) -- (463,30) -- (463,80) -- (413,80) -- cycle ;
\draw  [fill={rgb, 255:red, 163; green, 2; blue, 22 }  ,fill opacity=1 ] (61,126) -- (111,126) -- (111,176) -- (61,176) -- cycle ;
\draw  [fill={rgb, 255:red, 87; green, 158; blue, 6 }  ,fill opacity=1 ] (141,126) -- (191,126) -- (191,176) -- (141,176) -- cycle ;
\draw  [fill={rgb, 255:red, 163; green, 2; blue, 22 }  ,fill opacity=1 ] (222,126) -- (272,126) -- (272,176) -- (222,176) -- cycle ;
\draw  [fill={rgb, 255:red, 163; green, 2; blue, 22 }  ,fill opacity=1 ] (320,126) -- (370,126) -- (370,176) -- (320,176) -- cycle ;
\draw  [fill={rgb, 255:red, 163; green, 2; blue, 22 }  ,fill opacity=1 ] (414,126) -- (464,126) -- (464,176) -- (414,176) -- cycle ;
\draw  [fill={rgb, 255:red, 87; green, 158; blue, 6 }  ,fill opacity=1 ] (67,281) -- (228,281) -- (228,335) -- (67,335) -- cycle ;
\draw   (58,228) -- (238,228) -- (238,388) -- (58,388) -- cycle ;
\draw  [fill={rgb, 255:red, 87; green, 158; blue, 6 }  ,fill opacity=1 ] (281,281) -- (442,281) -- (442,335) -- (281,335) -- cycle ;
\draw   (272,228) -- (452,228) -- (452,388) -- (272,388) -- cycle ;

\draw (72,10) node [anchor=north west][inner sep=0.75pt]   [align=left] {$\displaystyle 1$};
\draw (151,10) node [anchor=north west][inner sep=0.75pt]   [align=left] {$\displaystyle 2$};
\draw (233,10) node [anchor=north west][inner sep=0.75pt]   [align=left] {$\displaystyle 3$};
\draw (272,46) node [anchor=north west][inner sep=0.75pt]   [align=left] {$\displaystyle \dotsc \dotsc $};
\draw (304,10) node [anchor=north west][inner sep=0.75pt]   [align=left] {$\displaystyle \lceil \log t\rceil -1$};
\draw (413,9) node [anchor=north west][inner sep=0.75pt]   [align=left] {$\displaystyle \lceil \log t\rceil $};
\draw (37,105) node [anchor=north west][inner sep=0.75pt]   [align=left] {$\displaystyle \lceil \log t\rceil +1$};
\draw (124,105) node [anchor=north west][inner sep=0.75pt]   [align=left] {$\displaystyle \lceil \log t\rceil +2$};
\draw (206,105) node [anchor=north west][inner sep=0.75pt]   [align=left] {$\displaystyle \lceil \log t\rceil +3$};
\draw (273,142) node [anchor=north west][inner sep=0.75pt]   [align=left] {$\displaystyle \dotsc \dotsc $};
\draw (305,106) node [anchor=north west][inner sep=0.75pt]   [align=left] {$\displaystyle 2\lceil \log t\rceil -1$};
\draw (414,105) node [anchor=north west][inner sep=0.75pt]   [align=left] {$\displaystyle 2\lceil \log t\rceil $};
\draw (92,203) node [anchor=north west][inner sep=0.75pt]   [align=left] {$\displaystyle 2\lceil \log t\rceil +1$};
\draw (126,235.33) node [anchor=north west][inner sep=0.75pt]   [align=left] {$\displaystyle \dotsc \dotsc $};
\draw (127,260) node [anchor=north west][inner sep=0.75pt]   [align=left] {$\displaystyle \dotsc \dotsc $};
\draw (70,287.33) node [anchor=north west][inner sep=0.75pt]   [align=left] {$\displaystyle P_{2}( n,n)$ over $\displaystyle A_{i}( n,n)$\\and $\displaystyle I( n,n)$};
\draw (128,342) node [anchor=north west][inner sep=0.75pt]   [align=left] {$\displaystyle \dotsc \dotsc $};
\draw (129,360) node [anchor=north west][inner sep=0.75pt]   [align=left] {$\displaystyle \dotsc \dotsc $};
\draw (306,203) node [anchor=north west][inner sep=0.75pt]   [align=left] {$\displaystyle 2\lceil \log t\rceil +2$};
\draw (340,235.33) node [anchor=north west][inner sep=0.75pt]   [align=left] {$\displaystyle \dotsc \dotsc $};
\draw (341,260) node [anchor=north west][inner sep=0.75pt]   [align=left] {$\displaystyle \dotsc \dotsc $};
\draw (284,287.33) node [anchor=north west][inner sep=0.75pt]   [align=left] {$\displaystyle P_{3}( D_{i})$ over $\displaystyle A_{i}( n,n)$\\and $\displaystyle I( n,n)$};
\draw (342,342) node [anchor=north west][inner sep=0.75pt]   [align=left] {$\displaystyle \dotsc \dotsc $};
\draw (343,360) node [anchor=north west][inner sep=0.75pt]   [align=left] {$\displaystyle \dotsc \dotsc $};

\end{tikzpicture}
}
\caption{If $p$ is $(\lceil \log{t} \rceil +2)$-globally optimal for the constructed instance, then the $\lceil \log{t} \rceil$ layers in which it is pareto optimal among the first $2\lceil \log{t} \rceil$ layers encode $i \in \lbrace 0 , \ldots, t-1 \rbrace$ such that: (1) $p$ allocates all the items from $I(n,n)$ to all the agents from $A_{i}(n,n)$, and (2) $p$ is pareto optimal in $P_{1}(n,n)$ over $A_{i}(n,n)$ and $I(n,n)$. We have that $p$ must be pareto optimal in both layers $2 \lceil \log{t} \rceil +1$ and $2 \lceil \log{t} \rceil +2$, implying that it is pareto optimal in both $P_{2}(n,n)$ and $P_{3}(D_{i})$ over $A_{i}(n,n)$ and $I(n,n)$.}
\label{fig:ccItemsPlusEll}
\end{figure}

We now prove the correctness of the construction. Namely, there exists $i \in \lbrace 0 , \ldots, t-1 \rbrace$ such that $D_{i}$ is a \yes-instance\ of \ThreeSat\ if and only if there exists an $\alpha$-globally optimal for the constructed instance.

($\Rightarrow$):
Suppose there exists $i \in \lbrace 0 , \ldots, t-1 \rbrace$ such that $D_{i}$ is a \yes-instance. By \Lecref{lemma:ThreeSatYesInstanceIffPO}, there exists an assignment $q : A(n,n) \rightarrow I(n,n) \cup \lbrace b_{\emptyset} \rbrace$ such that $q$ is pareto optimal in $P_{1}(n,n), P_{2}(n,n)$ and $P_{3}(D_{i})$. We define an assignment $p: A \rightarrow I \cup \lbrace b_{\emptyset} \rbrace$ by $p(a_{r,j}^{i}) = q(a_{r,j}) , p(\overline{a_{r,j}^{i}}) = q(\overline{a_{r,j}})$ for each $j,r \in [n]$, and $p(a_{r,j}^{y}) = b_{\emptyset} , p(\overline{a_{r,j}^{y}}) = b_{\emptyset}$ for each $j,r \in [n] , y \in \lbrace 0 , \ldots, t-1 \rbrace \setminus \lbrace i \rbrace$. By Claims \ref{claim1:ItemsPlusEll} and \ref{corollary1:ItemsPlusEll}, we have that $p$ is pareto optimal in exactly $\lceil \log{t} \rceil $ layers among the first $2\lceil \log{t} \rceil$ layers and that it encodes $i$. By Claims \ref{claim4:ItemsPlusEll} and \ref{claim5:ItemsPlusEll}, $p$ is pareto optimal in both layers $2 \lceil \log{t} \rceil +1$ and $2 \lceil \log{t} \rceil +2$. Thus, $p$ is an $\alpha$-globally optimal assignment for the constructed instance.

($\Leftarrow$): Assume that there exists $p : A \rightarrow I \cup \lbrace b_{\emptyset} \rbrace$ that is $\alpha$-globally optimal for the constructed instance (see Figure \ref{fig:ccItemsPlusEll}). By Claims \ref{corollary1:ItemsPlusEll} and \ref{corollary2:ItemsPlusEll}, $p$ is pareto optimal in exactly $\lceil \log{t} \rceil$ layers among the first $2 \lceil \log{t} \rceil$ layers, and there exists $i \in \lbrace 0 , \ldots, t-1 \rbrace$ such that $p$ encodes $i$ in the first $2 \lceil \log{t} \rceil$ layers. By \Clcref{claim3:ItemsPlusEll}, $p$ is pareto optimal in $P_{1}(n,n)$ over $A_{i}(n,n)$ and $I(n,n)$. Since $p$ is $\alpha$-globally optimal, it is pareto optimal in layers $2 \lceil \log{t} \rceil +1$ and $2 \lceil \log{t} \rceil +2$. Claims \ref{claim4:ItemsPlusEll} and \ref{claim5:ItemsPlusEll} imply that $p$ is pareto optimal in both $P_{2}(n,n)$ and $P_{3}(D_{i})$ over $A_{i}(n,n)$ and $I(n,n)$. Hence, by \Lecref{lemma:ThreeSatYesInstanceIffPO}, we have that $D_{i}$ is a \yes-instance.  

Since the constructed instance satisfies $m + \ell = \OO(n^{2}+ \log(t))$, we have by \Prcref{prop:noKern} that \goassign\ does not admit a polynomial kernel with respect to the parameter $ k = m + \ell$, unless \NP$\subseteq $\coNPpoly.\qed
\end{proof}
\section{Conclusion and Future Research}
\label{sec:conclusion}
In this paper, we introduced a new variant of the \assign\ problem where each agent is equipped with multiple incomplete preference lists, and we defined a corresponding notion of global optimality, that naturally extends pareto optimality. We considered several natural parameters, and presented a comprehensive picture of the parameterized complexity of the problem with respect to them.

The results show that the problem of finding an $\alpha$-globally optimal assignment is, in general, computationally hard, but that it admits more positive results when the parameter depends on $n = \nagents$ (and $\alpha$ or $\ell$) than when it depends on $m = \nitems$ (and $\alpha$ or $\ell$). We proved that the problem admits an \XP\ algorithm with respect to $m$, but is unlikely to admit one with respect to $\ell + d$ and $\ell + (m - d)$. We provided an $\OO^{*}(n!)$-time algorithm and an exponential kernel with respect to $m + \ell$. Both results showed that the problem is \FPT\ with respect to these parameters. In addition, we proved that $\OO^{*}(k!)$ is essentially a tight lower bound on the running time under \ETH\ for even larger parameters than $n$ such as $k = n + m + \alpha$ and $k = n + m + (\ell - \alpha)$. Moreover, we proved that the problem admits a polynomial kernel with respect to  $m + \ell$, but is unlikely to admit one with respect to all the other parameters that we considered. We also proved that the problem is \WOH\ with respect to $m + \alpha$ and $m + (\ell - \alpha)$. However, two questions are still open:
\begin{enumerate}
  \item  Is it possible to obtain a (not polynomial) better kernel for $m + \ell$ with size {\em substantially} smaller than $\OO^{*}((m!)^{\ell+1})$?
  \item Is it possible to obtain an algorithm with a better running time than $\OO^{*}(k!)$ for $k = n + m + \ell$? It can be shown using a reduction from \ThreeSat\ (similar to the reduction from \ThreeSat\ in Section \ref{sec:NPHardness}) that the problem cannot be solved in time $2^{o(k)}$ under \ETH. However, can a $2^{\OO(k)}$-algorithm be achieved?
\end{enumerate}

Continuing our research, it might be interesting to study ``weaker" definitions of optimality, for example: finding an assignment such that every group of $k$ agents has some $\alpha$ layers where they (1) do not admit trading cycles; (2) are not parts of larger trading cycles; or (3) do not admit the same trading cycle. Verification variants of these problems can also be suggested, i.e.~given an assignment $p$, check whether it is optimal.

Another direction is to study the particular case where the preferences of the agents are complete since it may provide more positive algorithmic results under some parameterizations. In addition, notice that a solution to \goassign\ can be seen as an approximation to the ``optimal'' solution in which an assignment is pareto optimal in a maximum number of layers (this is similar to the \textsc{Vertex Cover} problem, where the parameter $k$ is somewhat an ``approximation'' to the size of the minimum vertex cover). In this approach, we can define the problem as an approximation problem and study it from the perspective of parameterized approximation.

In this paper, we considered the basic ``unweighed'' model of the problem (since this is the first study of this kind). Another direction is to consider a weighted version in which some criteria (layers) may have higher importance than others. A straightforward way to model this is by having several copies of layers. However, if weights are high and varied, this might lead to inefficiency.
\newpage
\bibliographystyle{splncs04}
\bibliography{Refs}
\end{document}